\documentclass[preprint, twoside, imslayout]{imsart}
\RequirePackage[OT1]{fontenc}
\RequirePackage{amsthm,amsmath}
\RequirePackage[numbers]{natbib}
\RequirePackage[colorlinks,citecolor=blue,urlcolor=blue]{hyperref}
\usepackage{listings}

\startlocaldefs
\numberwithin{equation}{section}
\theoremstyle{plain}
\newtheorem{thm}{Theorem}[section]
\endlocaldefs

\newcommand{\E}{\mathbb{E}}

\newcommand{\tP}{\tilde{P}}	
\newcommand{\tPg}{\tilde{P}_{\gamma}}	
\newcommand{\tpi}{\tilde{\pi}}	
\newcommand{\tpig}{\tilde{\pi}_{\gamma}}
\newcommand{\tpigLi}{\tilde{P}_{\gamma,L,i}} 
\newcommand{\tpigLk}{\tilde{P}_{\gamma,L,k}} 
\newcommand{\tpigJi}{\tilde{P}_{\gamma,J,i}} 
\newcommand{\tpigJk}{\tilde{P}_{\gamma,J,k}} 
\newcommand{\RJik}{R_{\gamma,J,ik}} 

\newcommand{\RJi}{R_{\gamma,J,i}}
\newcommand{\RJk}{R_{\gamma,J,k}}
\newcommand{\RLi}{R_{\gamma,L,i}} 
\newcommand{\RLk}{R_{\gamma,L,k}}

\newcommand{\agik}{a_{\gamma,ik}} 
\newcommand{\agki}{a_{\gamma,ki}}  
  
\newcommand{\wgi}{w_{\gamma,i}} 
\newcommand{\wgk}{w_{\gamma,k}}
\newcommand{\wgj}{w_{\gamma,j}}
\newcommand{\giY}{\gamma \in \mathcal{Y}} 
\newcommand{\bb}[1]{\mathbb{#1}}

\theoremstyle{plain}

\newtheorem{lemma}[thm]{Lemma}
\newtheorem{cor}[thm]{Corollary}
\theoremstyle{definition}

\newtheorem{defi}[thm]{Definition}

\theoremstyle{remark}
\newtheorem{remark}[thm]{Remark}

\usepackage{geometry}
\usepackage[utf8]{inputenc}
\usepackage{amsmath, amsthm, amssymb}
\usepackage{enumerate}
\usepackage{verbatim}
\usepackage{graphicx}
\usepackage{algorithm}
\usepackage{algpseudocode}
\usepackage{listings}
\usepackage{color}
\usepackage{subcaption}
\usepackage{hyperref}
\usepackage{float}
\usepackage[T1]{fontenc}
\usepackage{array,multirow,graphicx}
\usepackage{tikz}
\usepackage[most]{tcolorbox} 
\usepackage{pgfplots} 
\graphicspath{ {./Plots/} }
\usetikzlibrary{arrows.meta}
\usetikzlibrary{arrows, positioning, shapes.geometric}
\usepackage{placeins} 

\usepackage{array, makecell}%
\usepackage{placeins}
\setcellgapes{4pt}
\begin{document}
\algnewcommand{\algorithmicgoto}{\textbf{go to}}%
\algnewcommand{\Goto}[1]{\algorithmicgoto~\ref{#1}}%

\begin{frontmatter}
\title{A Framework for Adaptive MCMC Targeting Multimodal Distributions}
\runtitle{Adaptive MCMC for Multimodal Distributions}

\begin{aug}
\author{\fnms{Emilia} \snm{Pompe}\thanksref{t1,m1}\ead[label=e1]{emilia.pompe@stats.ox.ac.uk}},
\author{\fnms{Chris} \snm{Holmes}\thanksref{t2,m1}\ead[label=e2]{cholmes@stats.ox.ac.uk}}
\and
\author{\fnms{Krzysztof} \snm{{\L}atuszy\'nski}\thanksref{t3,m2}
\ead[label=e3]{K.G.Latuszynski@warwick.ac.uk}
}

\thankstext{t1}{Supported by the EPSRC and MRC Centre for Doctoral Training in Next Generation Statistical Science: the Oxford-Warwick Statistics Programme, EP/L016710/1, and the Clarendon Fund.}
\thankstext{t2}{Supported by the MRC, the EPSRC, the Alan Turing Institute, Health Data Research UK, and the Li Ka Shing foundation.}
\thankstext{t3}{Supported by the Royal Society via the University Research
Fellowship scheme.}

\runauthor{E. Pompe, C. Holmes and K. {\L}atuszy\'nski}
\affiliation{University of Oxford\thanksmark{m1} and University of Warwick\thanksmark{m2}}

\address{Department of Statistics \\
University of Oxford\\
24–29 St Giles'\\
Oxford OX1 3LB\\
United Kingdom\\
\printead{e1}\\
\phantom{E-mail:\ }\printead*{e2}}

\address{Department of Statistics \\
University of Warwick\\
Coventry, CV4 7AL\\
United Kingdom\\
\printead{e3}\\
}
\end{aug}

\begin{abstract} 
	We propose a new Monte Carlo method for sampling from multimodal distributions. The idea of this technique is based on splitting the task into two: finding the modes of a target distribution $\pi$ and sampling, given the knowledge of the locations of the modes. The sampling algorithm relies on steps of two types: local ones, preserving the mode; and jumps to regions associated with different modes. Besides, the method learns the optimal parameters of the algorithm while it runs, without requiring user intervention.  Our technique should be considered as a flexible framework, in which the design of moves can follow various strategies known from the broad MCMC literature. 
	
	In order to design an adaptive scheme that facilitates both local and jump moves, we introduce an auxiliary variable representing each mode and we define a new target distribution $\tilde{\pi}$ on an augmented state space $\mathcal{X}~\times~\mathcal{I}$, where $\mathcal{X}$ is the original state space of $\pi$ and $\mathcal{I}$ is the set of the modes. As the algorithm runs and updates its parameters, the target distribution $\tilde{\pi}$ also keeps being modified. This motivates a new class of algorithms, Auxiliary Variable Adaptive MCMC. We prove general ergodic results for the whole class before specialising to the case of our algorithm.
\end{abstract}



\end{frontmatter}

\section{Introduction}\label{section:introduction}
Poor mixing of standard Markov Chain Monte Carlo (MCMC) methods, such as the Metropolis-Hastings algorithm or Hamiltonian Monte Carlo, on multimodal target distributions with isolated modes is a well-described problem in statistics. Due to their dynamics these algorithms struggle with crossing low probability barriers separating the modes and thus take a long time before moving from one mode to another, even in low dimensions. Sequential Monte Carlo (SMC) has often empirically proven to outperform MCMC on this task, its robust behaviour, however, relies strongly on the good between-mode mixing of the Markov kernel used within the SMC algorithm (see \cite{jasra2015error}). Therefore, constructing an MCMC algorithm which enables fast exploration of the state space for complicated target functions is of great interest, especially as multimodal distributions are common in applications. The examples include, but are not limited to, problems in genetics \cite{craiu2009learn, kou2006discussion}, astrophysics \cite{feroz2009multinest, feroz2013importance, tak2017repelling} and sensor network localisation~\cite{ihler2005nonparametric}.

Moreover, multimodality is an inherent issue of Bayesian mixture models (e.g. \cite{jasra2005markov}), where it is caused by label-switching, if both the prior and the likelihood are symmetric with respect to all components, or more generally, it may be caused by model identifiability issues or model misspecification (see \cite{drton2004multimodality}).

Designing MCMC algorithms for sampling from a multimodal target distribution $\pi$ on a $d$-dimensional space $\mathcal{X}$ needs to address three fundamental challenges: \begin{itemize} \item[(1)] Identifying high probability regions where the modes are located; \item[(2)] Moving between the modes by crossing low probability barriers; \item[(3)] Sampling efficiently within the modes by accounting for inhomogeneity between them and their local geometry.
\end{itemize}
These challenges are fundamental in the sense that the high probability regions in (1) are typically exponentially small in dimension with respect to a reference measure on the space $\mathcal{X}$. Furthermore, the design of basic reversible Markov chain kernels prevents them from visiting and crossing low energy barriers, hence from overcoming (2). Besides, accounting for inhomogeneity of the modes in (3) requires dividing the $d$-dimensional space $\mathcal{X}$ into regions, which is an intractable task on its own that requires detailed a priori knowledge of~$\pi$.

Existing MCMC methodology for multimodal distributions usually identifies these challenges separately and a systematic way of addressing (1-3) is not available. Section \ref{section:other_approaches} discusses main areas of the abundant literature on the topic in more detail.

In this paper we introduce a unifying framework for addressing (1-3) simultaneously  via a novel design of Auxiliary Variable Adaptive MCMC. The framework allows us to split the sampling task into mode finding, between-region jump moves and local moves. In addition, it incorporates parameter adaptations for optimisation of the local and jump kernels, and identification of local regions. Unlike other state space augmentation techniques for multimodal distributions, where the auxiliary variables are introduced to improve mixing on the extended state space, auxiliary variables in our approach help to design an efficient adaptive scheme. We present the adaptive mechanics and main properties of the resulting algorithm in Section \ref{section:contribution} after reviewing the literature.

\subsection{Other approaches}\label{section:other_approaches} Numerous MCMC methods have been proposed to address the issue of multimodality and we review briefly the main strands of the literature.

The most popular approach is based on tempering. The idea behind this type of methods relies on an observation  that raising a multimodal distribution $\pi$ to the power $\beta \in(0,1)$ makes the modes "flatter" and as a result, it is more likely to accept moves to the low probability regions. Hence, it is easier to explore the state space and find the regions where the modes of $\pi$ are located, addressing challenge (1) above, and also to move between these regions, addressing challenge (2). The examples of such methods, which incorporate $\pi^{\beta}$ by augmenting the state space, are parallel tempering  proposed by \cite{geyer1991markov} and its adaptive version \cite{miasojedow2013adaptive}, simulated tempering \cite{marinari1992simulated}, tempered transitions \cite{neal1996sampling} and the equi-energy sampler \cite{kou2006discussion}. Despite their popularity, tempering-based approaches, as noticed by \cite{woodard2009sufficient}, tend to mix between modes exponentially slowly in dimension if the modes have different local covariance structures. Addressing this issue is an area of active research~\cite{tawn2018weight}.

Another strand of research is optimisation-based methods, which address challenge (1) by running preliminary optimisation searches in order to identify local maxima of the target distribution. They use this information in their between-mode proposal design to overcome challenge (2). A method called Smart Darting Monte Carlo, introduced in \cite{andricioaei2001smart}, relies on moves of two types: jumps between the modes, allowed only in non-overlapping $\epsilon$-spheres around the local maxima identified earlier; and local moves (Random Walk Metropolis steps).  This technique was generalised in \cite{sminchisescu2011generalized} by allowing the jumping regions to overlap and have an arbitrary volume and shape. \cite{ahn2013distributed}~went one step further by introducing updates of the jumping regions and parameters of the proposal distribution at regeneration times, hence the name of their method Regeneration Darting Monte Carlo (RDMC). This includes a possibility of adding new locations of the modes at regeneration times if they are detected by optimisation searches running on separate cores.

Another optimisation-based method, Wormhole Hamiltonian Monte Carlo, was introduced by \cite{lan2014wormhole} as an extension of Riemanian Manifold HMC (see \cite{duane1987hybrid} and \cite{girolami2011riemann}).
The main underlying idea here is to construct a network of "wormholes" connecting the modes (neighbourhoods of straight line segments between the local maxima of $\pi$). The Riemannian metric used in the algorithm is a weighted mixture of a standard metric responsible for local HMC-based moves and another metric, influential in the vicinity of the wormholes, which shortens the distances between the modes. As before, updates of the parameters of the algorithm, including the network system, are allowed at regeneration times. 

As we will see later, the algorithm we propose also falls into the category of optimisation-based methods.

The Wang-Landau algorithm \cite{wang2001efficient, wang2001determining} or its adaptive version proposed by~\cite{bornn2013adaptive} belong to the exploratory strategies that aim to push the algorithm away from well-known regions and visit new ones, hence addressing challenge (1). The multi-domain sampling technique,  proposed in \cite{zhou2011multi}, combines the idea of the Wang-Landau algorithm with the optimisation-based approach. This algorithm relies on partitioning the state space into domains of attraction of the modes. Local moves are Random Walk Metropolis steps proposed from a distribution depending on the domain of attraction of the current state. Jumps between the modes follow the independence sampler scheme, where the new states are proposed from a mixture of Gaussian distributions approximating~$\pi$.

Other common approaches include Metropolis-Hastings algorithms with a special design of the proposal distribution accounting for the necessity of moving between the modes \cite{tjelmeland2001mode, tak2017repelling} and MultiNest algorithms based on nested sampling \cite{feroz2009multinest, feroz2013importance}.

\subsection{Contribution}\label{section:contribution} As mentioned before, the existing MCMC methods for multimodal distributions struggle to tackle challenges (1-3) simultaneously. In particular, challenge (3) typically fails to be addressed. The difficulty behind this challenge is that when modes have distinct shapes, different local proposal distributions will work well in regions associated with different modes. Note that the majority of the methods described above (all tempering-based techniques, the equi-energy sampler, the adaptive and non-adaptive Wang-Landau algorithm)  only employ a single transition kernel, regardless of the region. 

In applied problems optimal parameters of the MCMC kernels are unknown, therefore recent approaches involve tuning them while the algorithm runs. In case of unimodal target distributions Adaptive MCMC techniques prove to be useful \cite{roberts2009examples, MR2461882, MR3360496}. The parameters, such as covariance matrices, the scaling, the step size and the number of leapfrog steps of the involved Metropolis-Hastings, MALA, or HMC kernels can be learned on the fly as the simulation progresses, based on the samples observed so far. The adaptive algorithms remain ergodic under suitable regularity conditions \cite{MR2260070, roberts2007coupling, MR3012408, MR2648752, chimisov2018air}.

In case of multimodal distributions an analogous idea would be to apply these Adaptive MCMC methods separately to regions associated with different modes,  to improve the within-mode mixing. Note that in order to sample from different proposal distributions in regions associated with different modes, one needs to control at each step of the algorithm which region the current state belongs to. Besides, adapting parameters of the local proposal distributions on the fly must be based on samples that actually come from the corresponding region. The only known approach to assigning samples to regions is that of the multi-domain sampler \cite{zhou2011multi}. However, in their setting keeping track of the regions requires running a gradient ascent procedure at each MCMC step, which
imposes a high computational burden on the whole algorithm. Other optimisation-based approaches known in the literature (e.g.  \cite{sminchisescu2011generalized} and  \cite{ahn2013distributed}) tend to
ignore the necessity of assigning samples to regions and the possibility of moving between the modes via local  steps.

An issue that we have not raised so far is that the adaptive optimisation-based methods presented above, such as those of \cite{ahn2013distributed} and \cite{lan2014wormhole}, allow for adaptations only at regeneration times. Although this approach seems appealing from the point of view of the theory, since no further proofs of convergence are needed, it does not work well in practice in high dimensions. The reason for this is that regenerations happen rarely in large dimensions, which makes the adaptive scheme prohibitively inefficient.  Besides, identifying regeneration times using the method of \cite{mykland1995regeneration}, as authors of both algorithms propose, requires case-specific calculations which precludes any generic implementation of an algorithm based on regenerations. Moreover, the resulting  identified regenerations are of orders of magnitude more infrequent than the "true" ones which are already rare.

We aim to remedy these shortcomings by proposing a framework for designing an adaptive algorithm on an augmented state space $\mathcal{X} \times \mathcal{I}$, where $\mathcal{I} = \{1, \dots, N\}$, and the auxiliary variable $i$ of the resulting sample $(x, i)$ encodes the corresponding region for $x$. Local MCMC kernels update $x$ only, while jump kernels that move between the modes update $x$ and $i$ simultaneously. Furthermore, the design of the target distribution on the augmented state space prevents the algorithm from moving to a region associated with a different mode via local steps. In the sequel we make specific choices for the adaptive scheme, the local and jump kernels, as well as the burn-in routine used for setting up initial values of the parameters of the algorithm. However, the design is modular and different approaches can be incorporated in the framework. Besides, it allows for a multicore implementation of a large part of the algorithm.

This approach motivates introducing the Auxiliary Variable Adaptive MCMC class, where not only transition kernels are allowed to be modified on the fly, but also the augmented target distributions. It turns out that apart from our method, there is a wide range of algorithms that belong to this class, including adaptive parallel tempering or adaptive versions of pseudo-marginal MCMC. Thus our general ergodicity results, proved for the whole class under standard regularity conditions, can potentially be useful for analysing other methods.

The remainder of the paper is organised as follows. In Section \ref{section:adapt_mcmc} we present our algorithm, the Jumping Adaptive Multimodal Sampler (JAMS) and discuss its properties. In Section \ref{section:auxiliary_variable} we define  the Auxiliary Variable Adaptive MCMC class and establish convergence in distribution and a Weak Law of Large Numbers for this class, under the uniform and the non-uniform scenario. We present theoretical results specialised to the case of our proposed algorithm in Section \ref{section:ergodicty_amcmc}. Ergodicity is derived here from the analogues of the Containment and Diminishing Adaptation conditions introduced in \cite{roberts2007coupling}, as opposed to identifying regeneration times, which allows us to circumvent the issues described above.  The proofs of all our theorems along with some additional comments about the theoretical results are gathered in Supplementary Material A. Section \ref{section:examples} demonstrates the performance of our method on two synthetic and one real data example. Additional details of our numerical experiments are available in Supplementary Material B. We conclude with a summary of our results in Section \ref{section:summary}.

\section{Jumping Adaptive Multimodal Sampler (JAMS)}~\label{section:adapt_mcmc}~\subsection{Main algorithm}\label{subsection:main_algorithm}
Let $\pi$ be the multimodal target distribution of interest defined on $(\mathcal{X}, \mathcal{B}(\mathcal{X}))$. We introduce a collection of target distributions $\{\tpig \}_{\giY}$ on the augmented state space $\mathcal{X} \times \mathcal{I}$, where $\mathcal{I}:= \{1, \ldots, N\}$ is the finite set of indices of the modes of $\pi$. We defer the discussion about finding the modes to Section \ref{subsection:burn_in_algorithm}. Here $\gamma$ denotes the design parameter of the algorithm that may be adapted on the fly. For a fixed $\giY$, $\tpig$ is defined as
\begin{equation}\label{eq:pi_gamma_defined}
\tpig(x,i): = \pi(x) \frac{\wgi Q_i(\mu_i, \Sigma_{\gamma,i})(x)}{\sum_{j\in \mathcal{I}} \wgj Q_j(\mu_j, \Sigma_{\gamma,j})(x)},
\end{equation}
where $Q_i(\mu_i, \Sigma_{\gamma,i})$ is an elliptical distribution (such as the normal or the multivariate $t$ distribution) centred at $\mu_i$ with covariance  matrix $\Sigma_{\gamma, i}$. We shall think of $\{\mu_i\}_{i\in \mathcal{I}}$ and  $\{\Sigma_{\gamma, i}\}_{i\in \mathcal{I}}$ as locations and covariances of the modes of $\pi$, respectively.
Firstly, notice that constructing a Markov chain targeting~$\tpig$ provides a natural way of identifying the mode at each step by recording the auxiliary variable $i$. Besides, for each $B \in \mathcal{B}(\mathcal{X})$ and $\giY$  we have
\begin{equation}\label{eq:jams_main_equation}
\tpig(B \times \mathcal{I}) = 
\int_{B}\sum_{i \in \mathcal{I}} \pi(x) \frac{\wgi Q_i(\mu_i, \Sigma_{\gamma,i})(x)}{\sum_{j\in \mathcal{I}} \wgj Q_j(\mu_j, \Sigma_{\gamma,j})(x)} dx=  \int_{B}  \pi(x) \cdot 1  dx= \pi(B).
\end{equation}
Hence, $\pi$ is the marginal distribution of $\tpig$ for each $\giY$, so sampling from~$\tpig$ can be used to generate samples from $\pi$.

The sampling algorithm that we propose is summarised in Algorithm \ref{alg:mode_jumping}.  The method relies on MCMC steps of two types, performed with probabilities $1-\epsilon$ and $\epsilon$, respectively.

\begin{itemize}
	\item  \textbf{Local move}: Given the current state of the chain $(x, i)$ and the current parameter $\gamma$, a local kernel $\tP_{\gamma, L, i}$ invariant with respect to $\tpig$ is used to update $x$, while $i$ remains fixed, hence the mode is preserved. 
	\item  \textbf{Jump move}: Given the current state of the chain $(x, i)$ and the current parameter $\gamma$, a new mode $k$ is proposed with probability $\agik$. Then a new point $y$ is proposed using a distribution $\RJik(x, \cdot)$. The new pair is accepted or rejected using the standard Metropolis-Hastings formula such that jump kernel is invariant with respect to $\tpig$.
\end{itemize}

Our  choice for the local kernel is Random Walk Metropolis (RWM) with proposal $\RLi(x, \cdot)$ that follows either the normal or the $t$ distribution. This allows us to employ well-developed adaptation strategies for RWM and build on its stability properties to establish ergodicity of JAMS in Section \ref{section:ergodicty_amcmc}. However, in practice any other MCMC kernel, such as MALA or HMC, may be used. The standard Metropolis-Hastings acceptance probability formula that admits  $\tpig$ as the  invariant distribution for the local move becomes:
\begin{align}\label{eq:local_move_acceptance}
\begin{split}
\alpha_{\gamma, L}\left((x, i) \to (y, i)\right) & = \min \left[1, \frac{\tpig(y, i)}{\tpig(x, i)}\right] \\
&=   \min \left[1, \frac{\pi(y) Q_i(\mu_i, \Sigma_{\gamma,i})(y)}{\pi(x) Q_i(\mu_i, \Sigma_{\gamma,i})(x)}  \frac{\sum_{j\in \mathcal{I}} \wgj Q_j(\mu_j, \Sigma_{\gamma,j})(x)}{\sum_{j\in \mathcal{I}} \wgj Q_j(\mu_j, \Sigma_{\gamma,j})(y)}\right].
\end{split}
\end{align}

As for the jump moves, we consider two different methods of proposing a new point $y$ associated with mode $k$. The first one, which we call \emph{independent proposal jumps}, is to draw $y$ from an elliptical distribution centred at $\mu_k$ with covariance matrix $\Sigma_{\gamma, k}$, independently from the current point $(x,i)$. Since there is no dependence on $x$ and $i$, in case of independent proposal jumps the proposal distribution to mode $k$ will be denoted by $\RJk (\cdot)$.  For independent proposal jumps, the acceptance probability is equal to
\begin{align} \label{eq:independent_jump_move_acceptance}
\begin{split}
\alpha_{\gamma, J}\left((x, i) \to (y, k)\right) & = \min \left[1, \frac{\tpig(y, k)}{\tpig(x, i)} \frac{\agki \RJi(x)} {\agik \RJk(y)} \right].
\end{split}
\end{align}

Alternatively, given that the current state is $(x, i)$, we can propose a "corresponding" point $y$ in mode $k$ such that
\begin{equation*}
(x-\mu_i)^T \Sigma_{\gamma, i}^{-1} (x-\mu_i) = 	(y-\mu_k)^T \Sigma_{\gamma, k}^{-1} (y-\mu_k).
\end{equation*}
The required equality is satisfied for 
\begin{equation}\label{eq:deterministic_proposal}
y : = \mu_k + \Lambda_{\gamma, k}\Lambda_{\gamma,i}^{-1}(x - \mu_i),
\end{equation}
where 
$$
\Sigma_{ \gamma, i} = \Lambda_{\gamma,i } \Lambda_{ \gamma, i}^T \quad \text{and} \quad \Sigma_{\gamma, k} = \Lambda_{\gamma,k} \Lambda_{\gamma, k}^T.
$$
Herein this method will be called \emph{deterministic jumps}. The acceptance probability is then given by
\begin{equation}\label{eq:acc_prob_deterministic}
\alpha_{\gamma, J}\left((x, i) \to (y, k)\right) = 	\min\left[1, \frac{\tilde{\pi}(y,k)}{\tilde{\pi}(x,i)} \frac{\agki \sqrt{\det \Sigma_{\gamma,k}}}{\agik \sqrt{\det \Sigma_{\gamma,i}} } \right].
\end{equation}
Note that in both cases the design of the jump moves takes into account the shapes of the two modes involved, which helps achieving high acceptance rates and consequently improves the between-mode mixing.

\begin{algorithm}[!ht]
	\caption{JAMS: main algorithm prototype, iteration $n+1$}\label{alg:mode_jumping}
	\begin{algorithmic}[1]	
		\State {\textbf{Input:} current point $(x_n, i_n)$, list of modes $\{\mu_1, \ldots \mu_N \}$, constant $\epsilon \in (0,1)$, parameter $\gamma_{n} = \lbrace\Sigma_{\gamma_{n},i}, w_{\gamma_{n},i}, a_{\gamma_{n}, ik}\rbrace_{i,k \in \{1, \ldots, N\}}$, empirical means $m_1, \ldots, m_N$ and covariance matrices $S_1, \ldots, S_N$.} 
		\State {Generate $u \sim U[0,1]$.}
		\If {$u > \epsilon$}
		\State  {\textbf{Local move}:}
		\State{Propose a new value $y \sim R_{\gamma_{n}, L, i_n}(x_n, \cdot)$.}
		\State{Accept $y$ with probability $\alpha_{\gamma_{n}, L}\left((x_n, i_n) \to (y, i_n) \right)$}.
		\If {$y$ accepted}
		\State{$(x_{n+1}, i_{n+1}) := (y, i_n)$.}
		\Else
		\State{$(x_{n+1}, i_{n+1}) := (x_n, i_n)$.}
		\EndIf
		\Else
		\State  {\textbf{Jump move:}}
		\State {Propose a new mode $k \sim (a_{\gamma_n, i1}, \ldots, a_{\gamma_n, iN})$. }
		\State{Propose a new value $y \sim R_{\gamma_{n},J,ik}(x_n, \cdot)$.}
		\State{Accept $(y, k)$ with probability $\alpha_{\gamma_{n}, J}\left((x_n, i_n) \to (y, k) \right)$.}
		\If {$(y, k)$ accepted}
		\State{$(x_{n+1}, i_{n+1}) := (y, k)$.}
		\Else
		\State{$(x_{n+1}, i_{n+1}) := (x_n, i_n)$.}
		\EndIf
		\EndIf
		\State{Update the empirical mean $m_{i_{n+1}}$ and covariance matrix $S_{i_{n+1}}$ by including $x_{n+1}$.}
		\State{Update the parameter $\gamma_{n}$ to $\gamma_{n+1}$ according to Algorithm \ref{alg:params_updates}.}
		\State \Return {New sample $(x_{n+1}, i_{n+1})$, parameter $\gamma_{n+1}$, $m_{i_{n+1}}$ and $S_{i_{n+1}}$.}
	\end{algorithmic}
\end{algorithm}

As presented in Algorithm \ref{alg:mode_jumping},  the method involves learning the parameters on the fly. We design an adaptation scheme of three lists of parameters: covariance matrices (used both for adapting the target distribution $\tpig$ and the proposal distributions), weights $\wgi$ and probabilities $\agik$ of proposing mode $k$ in a jump from mode $i$. Hence, formally  $\mathcal{Y}$ refers to the product space of $\Sigma_{\gamma,i}$, $w_{\gamma,i}$ and $a_{\gamma,ik}$ for $i, k \in \{1, \ldots, N \}$ restricted by $\sum_{j \in \mathcal{I}} \wgj =1$ and $\sum_{k \in \mathcal{I}} \agik =1$ for each $\giY$ and each $i \in \mathcal{I}$. An adaptive scheme for $w_{\gamma,i}$ and $a_{\gamma,ik}$ that follows an intuitive heuristic is discussed briefly in Section~10  of Supplementary Material B.

\begin{algorithm}[ht]
	\caption{Updating the parameters, iteration $n+1$}\label{alg:params_updates}
	\begin{algorithmic}[1]	
		\State {\textbf{Input:} (in addition to the parameters of iteration $n+1$ of the main algorithm) number of samples observed so far in each mode $n_1, \ldots, n_N$, auxiliary matrices $\tilde{\Sigma}_{1}, \ldots, \tilde{\Sigma}_{N}$, positive integers $AC_1$ and $AC_2$, constants $\alpha, \alpha_{\text{opt}} \in(0,1)$, $\beta>0$. }
		\If {$n_{i_{n+1}}< AC_1$  }
		\If{\textbf{Local move}}
		\State{$\tilde{\Sigma}_{i_{n+1}}: = \exp \left(n_{i_{n+1}}^{-\alpha} \left(\alpha_{\gamma_n,L} - \alpha_{\text{opt}}\right)\right) \tilde{\Sigma}_{i_{n+1}} $.}
		\State{$\Sigma_{\gamma_{n+1}, i_{n+1}} := \tilde{\Sigma}_{i_{n+1}}+~\beta I_d$. }
		\EndIf
		\Else
		\If {$n_{i_{n+1}}$ is divisible by $AC_2$} 
		\State{$\Sigma_{\gamma_{n+1}, i_{n+1}} := S_{i_{n+1}}+\beta I_d.$} 
		\State{Update $w_{\gamma_{n},i}$ and $a_{\gamma_{n}, ik}$  to $w_{\gamma_{n+1},i}$ and $a_{\gamma_{n+1}, ik}$ for $i,k=1, \ldots, N$.}
		\EndIf
		\EndIf
		
	\end{algorithmic}
\end{algorithm}

Our method of adapting the covariance matrices $\Sigma_{\gamma,i}$ is presented in Algorithm~\ref{alg:params_updates}. For every $i \in \mathcal{I}$ the matrix $\Sigma_{\gamma,i}$ is based on the empirical covariance matrix of the samples from the region associated with mode $i$ obtained so far. This is possible in our framework by keeping track of the auxiliary variable $i$. Updates are performed every certain number of iterations (denoted by $AC_2$ in Algorithm \ref{alg:params_updates}). This method follows the classical Adaptive Metropolis methodology (cf. \cite{haario2001adaptive, roberts2009examples}) applied separately to the covariance structure associated with each mode. For the local proposal distributions the covariance matrices are additionally scaled by the factor $2.38^2/d$, which is commonly used as optimal for Adaptive Metropolis algorithms \cite{roberts1997weak, roberts2001optimal}. Since representing a covariance matrix in high dimensions reliably typically requires a large number of samples, we do not apply this method straight away. Instead, we perform adaptive scaling, aiming to achieve the optimal acceptance rate (typically  fixed at 0.234; see \cite{roberts1997weak, roberts2001optimal}) for local moves, until the number of samples observed in a given mode exceeds a pre-specified constant (denoted by $AC_1$ in Algorithm \ref{alg:params_updates}).

It is worth outlining that this special construction of the target distribution $\tpig$ makes it unlikely for the algorithm to escape  via local steps from the mode it is assigned to and settle in another one. Indeed, if a proposed point $y$ is very distant from the current mode $\mu_i$ and close to another mode $\mu_k$, the acceptance probability becomes very small due to the expression $Q_i(\mu_i, \Sigma_{\gamma,i})(y)$ in the numerator of (\ref{eq:local_move_acceptance})  and $Q_k(\mu_k, \Sigma_{\gamma,k})(y)$ in the denominator, as $Q_i(\mu_i, \Sigma_{\gamma,i})(y)$ will typically be tiny in such case and $Q_k(\mu_k, \Sigma_{\gamma,k})(y)$ will be large. This allows for controlling from which mode a given state of the chain was sampled. The property of our algorithm described above is crucial for its efficiency as it enables estimating  matrices $ \Sigma_{\gamma,i}$ based on samples that are indeed close to mode $\mu_i$, which in turn improves both the within-mode and the between-mode mixing. If we were working directly with $\pi$, the corresponding acceptance probability would be given by $\min\left[1, \frac{\pi(y)}{\pi(x)} \right]$ and we would not have a natural mechanism for preventing the sampler from visiting different regions via local moves.

\subsection{Burn-in algorithm}\label{subsection:burn_in_algorithm}
Note that Algorithm \ref{alg:mode_jumping} takes mode locations $\{\mu_1, \ldots, \mu_N\}$ and initial values of the matrices $\{\Sigma_{\gamma_0,1}, \ldots, \Sigma_{\gamma_0, N} \}$ as input. Recall also that further improvements in the estimation of $\Sigma_{\gamma, i}$ are possible after some samples in mode $i$ have been observed (see Algorithm \ref{alg:params_updates}). Hence, the  matrices $\{\Sigma_{\gamma_0,1}, \ldots, \Sigma_{\gamma_0, N}\}$ need to represent well the shapes of the corresponding modes so that jumps to all the modes are accepted reasonably quickly.

We address the issues of finding the local maxima of $\pi$, setting up the starting values of the covariance matrices and other aspects of the implementation of our method by introducing a burn-in algorithm, summarised by Algorithm~\ref{alg:burn_in_algorithm}. 

The burn-in algorithm  runs in advance, before the main MCMC sampler (Algorithm~\ref{alg:mode_jumping}) is started, in order to provide initial values of the parameters. Since it needs to find the locations of the modes of $\pi$, and this may be arbitrarily hard, one may prefer a version of this method in which the burn-in algorithm continues running in parallel to the main sampler on multiple cores. These cores communicate with the main sampler every certain number of iterations so that it can incorporate recently discovered modes into the augmented target distribution~$\tpig$.  For clarity of presentation, we focus on the sequential setting where the burn-in routine runs before the main algorithm, and in Sections \ref{section:auxiliary_variable} and \ref{section:ergodicty_amcmc} we develop ergodic theory that covers this case. However, the ergodic theory is immediately applicable to the version where the burn-in and the main algorithm run in parallel, as explained in Remark~\ref{rem:parallel_theory}. 

We sketch different stages of the burn-in routine below, in Sections \ref{subsubsection:starting_points} -- \ref{subsubsection:initial_matrix_estimation}, additional details are given in Supplementary Material B. The flowchart illustrating how the full algorithm works is shown in Figure \ref{fig:flowchart}. 

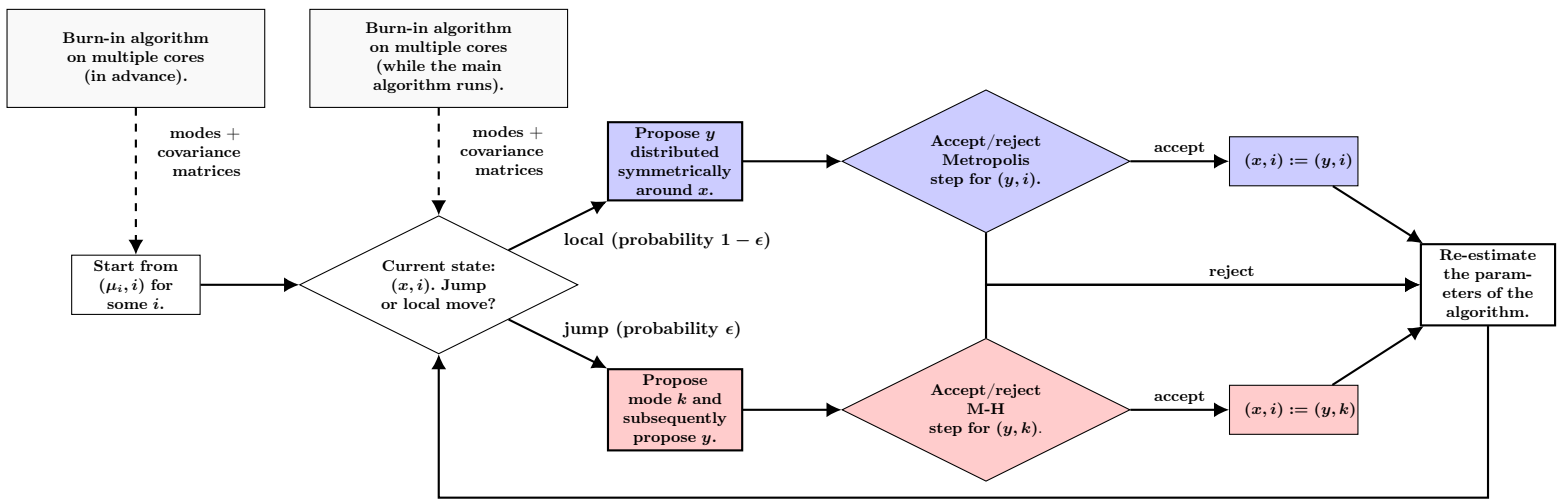
\begin{figure}
	\begin{tikzpicture}[
	scale=0.65, every node/.style={transform shape},
	trim left = -3.0cm,
	node distance =1cm and 2cm,
	process/.style = {rectangle, draw, fill=white,  thick,
		minimum width=1cm, minimum height=1cm, align=center, text width =25mm},
	mode/.style = {rectangle, draw, fill=white,
		minimum width=3cm, minimum height=2cm, align=center, text width =50mm},
	process2/.style = {rectangle, draw = none, fill=none, 
		minimum width=3cm, minimum height=1cm, align=center, text width =20mm},
	process1/.style = {rectangle, draw, fill=white,
		minimum width=2.6cm, minimum height=1cm, align=center, text width =20mm},
	decision/.style ={diamond, draw, fill=white, 
		minimum width=3cm, minimum height=1cm, align=center,aspect =2, text width = 30mm},
	arrow/.style = { thick, -stealth, -{Latex[length=2mm, width = 2mm]} },
	core/.style = {rectangle, draw = none, fill=none,
		minimum width=3cm, minimum height=1cm, align=center, text width = 250mm},
	arrownew/.style = { thick, -stealth, -{Latex[length=2mm, width = 2mm]}, dashed}
	]
	
	
	\node (dec1)  [decision]   {\textbf{Current state: $\boldsymbol{(x,i)}$. Jump or local move?}};
	\node (local1)  [process,above right=of dec1, fill =blue!20]   {\textbf{Propose $\boldsymbol{y}$ distributed symmetrically around $\boldsymbol{x}$.}};
	
	\node (start)  [process1,left=of dec1]{\textbf{Start from $\boldsymbol{(\mu_i,i)}$ for some  $\boldsymbol{i}$.}};
	\node (mode2)  [mode, above = 3cm of start, fill = gray!5] {\textbf{Burn-in algorithm \\ on multiple cores\\ (in advance).}};
	\node (mode)  [mode,  right= 0.9cm of mode2, fill = gray!5] {\textbf{Burn-in algorithm \\ on multiple cores \\ (while the main algorithm runs).}};
	
	\node (pro6b) [process2, above left= 4cm of dec1]   {};
	\node (pro7b) [process2, above right= 4cm of dec1]   {};
	
	\node (local2)  [decision,right=of local1, fill = blue!20]   {\textbf{Accept/reject Metropolis} \\ \textbf{step for $\boldsymbol{(y,i)}$.}};
	\node (local3)  [process1,right=of local2, fill = blue!20]   {$\boldsymbol{(x,i) := (y,i)}$};
	\node (jump1)  [process,below right=of dec1, fill = red!20]   {\textbf{Propose mode $\boldsymbol{k}$  and subsequently propose $\boldsymbol{y}$.}};
	\node (jump2)  [decision, right=of jump1, fill = red!20]   {\textbf{Accept/reject M-H\\ step for} $\boldsymbol{(y,k)}$.};
	
	\node (jump3)  [process1,right=of jump2, fill = red!20]   {$\boldsymbol{(x,i):= (y,k)}$};
	
	\node (pro2b) [process2, right=of dec1]   {};
	\node (pro3b) [process2, right=of pro2b]  {};
	\node (pro4b) [process2, right=of pro3b]  {};
	
	\node (reestimate)  [process, right = of pro4b]   {\textbf{Re-estimate the parameters of the algorithm.}};

	\draw [arrow] (dec1) -+ (jump1.north west) node[midway,above right] {\small{\textbf{jump (probability $\boldsymbol{\epsilon}$)}}};
	\draw [arrow] (dec1) -+ (local1.south west) node[midway,below right] {\small{\textbf{local (probability $\boldsymbol{1-\epsilon}$)}}};
	\draw [arrow] (local1) -+ (local2);
	\draw [arrow] (local2) -+ (local3)node[midway,above]{\textbf{accept}};
	\draw [arrow] (jump1) -+ (jump2);
	\draw [arrow] (jump2) -+ (jump3)node[midway,above]{\textbf{accept}};
	\draw [arrow] (jump3) -+ (reestimate);
	\draw [arrow] (local3) -+ (reestimate);
	\draw [arrow] (reestimate.south) --+ (0cm,-3.5cm) -|(dec1.south);
	\draw [arrow] (local2.south) --+ (0cm,-1cm) |-(reestimate)node[near end,above right]{\textbf{reject}};
	\draw [arrow] (jump2.north) --+ (0cm,0cm) |-(reestimate);

	\draw [arrownew] (mode2) --+ (start)  node[pos = 0.3, align = right, right, text width = 2cm]{\textbf{modes + covariance matrices}};
	\draw [arrownew] (mode) --+ (dec1)node[pos = 0.4, align = right, right, text width = 2cm]{\textbf{modes + covariance matrices}};
	\draw [arrow] (start) --+ (dec1);
	
	\end{tikzpicture}
	\caption{\footnotesize Flowchart illustrating the communication between the main algorithm and the burn-in algorithm (running before the main algorithm is initialised and in parallel). }\label{fig:flowchart}
\end{figure}
\begin{algorithm}[!ht]
	\caption{JAMS: burn-in algorithm}\label{alg:burn_in_algorithm}
	\begin{algorithmic}[1]	
		\State {\textbf{Input:}  end points of the intervals $L_1, U_1, \ldots,L_d, U_d$ (or the prior $\pi_0$), number of starting points $n$, small positive value $q$, threshold $b_{\text{acc}}$.}
		\State {Generate $s_1, \ldots, s_n$ uniformly from $[L_1, U_1] \times \ldots \times [L_d, U_d]$ (or $s_i \sim \pi_0$ for $i = 1, \ldots, n.$)}
		\State{Run BFGS from $s_1,  \ldots, s_n$ to minimize $- \log(\pi(x))$.} \Comment{ \textbf(in parallel)}
		\State{Denote the optimum points by $m_1, \ldots, m_n$ and their corresponding Hessian matrices by $H_1, \ldots, H_n$.}
		\State{Set $\mu_1:= m_1$, $H_{\mu_{1}}:= H_1$, $N:=1$.}
		\For{$i = 2, \ldots, n$}
		\If {$\min_{j \in 1, \ldots, N} \frac{1}{2} \left((\mu_j-m_i)^T H_{\mu_j} (\mu_j-m_i) + (\mu_j-m_i)^T H_i (\mu_j-m_i)\right)<q$}
		\State{$k := \arg\min_{j \in 1, \ldots, N} \frac{1}{2} \left((\mu_j-m_i)^T H_{\mu_j} (\mu_j-m_i) + (\mu_j-m_i)^T H_i (\mu_j-m_i)\right)$}
		\If{$\pi(\mu_k)< \pi(m_i)$}
		\State{Set $\mu_k: =m_i$ and $H_{\mu_k}: =H_i$.}
		\EndIf
		\Else \State{Set $\mu_{N+1}:= m_i$ and $H_{\mu_{N+1}}:= H_i$.}
		\State{Set $N:=N+1$.}
		\EndIf
		\EndFor
		\State{Define $\Sigma_{0,i} := H_{\mu_i}^{-1}$ for $i =1, \ldots, N$. }
		\State{Set $k=0$ and define $\tpi(x,i) := \pi(x) \frac{\frac{1}{N} Q_i(\mu_i, \Sigma_{0,i})(x)}{\sum_{j=1}^N \frac{1}{N} Q_j(\mu_j, \Sigma_{0,j})(x)}$.}
		
		\While{inhomogeneity factors $b_{k,1}, \ldots, b_{k,N}$ satisfy: $\max_{i \in \{1, \ldots, N \}} b_{k,i} > b_{\text{acc}}$}
		\State{k:= k+1}
		\State{Run Algorithm \ref{alg:mode_jumping} targeting $\tpi$ with $\epsilon =0$  from  $\mu_1, \ldots \mu_N$.}\Comment{ \textbf(in parallel)}
		\State{Define $\Sigma_{k,i}$ for $i =1, \ldots N$ as the matrix of mode $\mu_i$ updated in round $k$.}
		\State{Update the target distribution $\tpi$ by setting
			$\tpi(x,i) := \pi(x) \frac{\frac{1}{N} Q_i(\mu_i, \Sigma_{k,i})(x)}{\sum_{j=1}^N \frac{1}{N} Q_j(\mu_j, \Sigma_{k,j})(x)}$.}
		\EndWhile
		\State{Set $K:=k$ and $\Sigma_{\gamma_0,i} := \Sigma_{K,i}$, $\tilde{\Sigma}_{i}: = \Sigma_{K,i}$, $S_i:= \Sigma_{K,i}$ for $i =1, \ldots, N$.}
		\State \Return {List of $\{\mu_i, \Sigma_{\gamma_0,i}, \tilde{\Sigma}_{i}, S_i\}$ for $i=1, \ldots, N$.\newline Target distribution $\tpi_{\gamma_{0}}(x,i) := \pi(x) \frac{\frac{1}{N}Q_i(\mu_i, \Sigma_{\gamma_{0},i})(x)}{\sum_{j=1}^N \frac{1}{N} Q_j(\mu_j, \Sigma_{\gamma_{0},j})(x)}$ that will be used as input for Algorithm~\ref{alg:mode_jumping}.}
	\end{algorithmic}
\end{algorithm}
\subsubsection{Starting points for the optimisation procedure}\label{subsubsection:starting_points}
We  sample the starting points for optimisation searches uniformly on a compact set  which is a product of intervals provided by the user
$$
[L_1, U_1] \times \ldots \times [L_d, U_d],
$$ 
where $d$ is the dimension of the state space $\mathcal{X}$. Note that if the domain of attraction of each mode overlaps with $[L_1, U_1] \times \ldots \times [L_d, U_d]$, then asymptotically all modes will be found, as we will have at least one starting point in each domain.

When dealing with Bayesian models, one can alternatively sample the starting points from the prior distribution.

\subsubsection{Mode finding via an optimisation procedure}\label{subsubsection:BFGS}
The BFGS optimisation algorithm \cite{nocedal2006numerical} is initiated from every starting point. The BFGS method method provides the optimum point and the Hessian matrix at this point which is particularly useful in the next step of mode merging.

For numerical reasons, instead of working directly with $\pi$, we typically use the BFGS algorithm to find the local minima of $-\log(\pi)$.

\subsubsection{Mode merging}\label{subsubsection:mode_merging} 
Starting the optimisation procedure from different points belonging to the same basin of attraction will take us to points which are close to the true local maxima, but numerically different, an issue that seems to be ignored in optimisation-based MCMC literature. 

We deal with this in a heuristic way (lines 5-16 of Algorithm \ref{alg:burn_in_algorithm}) by classifying two vectors $m_i$ and $m_j$ as corresponding to the same mode if the squared Mahalanobis distance between them is smaller than some pre-specified value $q$. If we let $H_i$ and $H_j$ denote the Hessian matrices of $-\log(\pi)$ at  $m_i$ and $m_j$, respectively, the above Mahalanobis distance is calculated for $H_i^{-1}$ and $H_j^{-1}$ (for symmetry, we average over these two values). This method is scale invariant as the Hessian captures the local shape and scale.

\subsubsection{Initial covariance matrix estimation}\label{subsubsection:initial_matrix_estimation}
In order to find initial covariance matrix estimates $\Sigma_{\gamma_0,1}, ..., \Sigma_{\gamma_0, N}$ that accurately reflect the geometry of different modes, we employ the augmented target machinery of Algorithm~\ref{alg:mode_jumping} in the following way. We run Algorithm~\ref{alg:mode_jumping} without jumps, i.e. with $\epsilon =0$, in parallel, starting from each of the modes $\mu_1, \ldots, \mu_N$.  This implies that we run $N$ chains and each of them adapts only the matrix $\Sigma_i$ corresponding to the mode $\mu_i$ which was its starting point. We make a number of rounds (denoted by $K$) of this procedure and after each round we update the target distribution $\tilde{\pi}$ by exchanging the knowledge about the adapted covariance matrices between cores. The final covariance matrices passed to the main MCMC sampler are calculated based on the samples collected in all rounds.

The reason why we exchange information between rounds, despite the additional cost of communication between cores, is that we want the sampler adapting $\Sigma_{k,i}$ to know where the regions associated with other modes are so that it is less likely to visit those regions and contaminate the estimate. Essentially the initial covariance estimation revisits the problem of collecting samples only from the corresponding regions, discussed in the previous parts of this paper.

The initial value of the matrix corresponding to mode $i$ is the inverse of the Hessian evaluated at $\mu_i$ (see line 17 of Algorithm~\ref{alg:burn_in_algorithm}). The values of the other parameters of the algorithm, such as  $\alpha$, $\beta$ and $AC_2$, are set to be the same as in the main algorithm. The values of $w_{\gamma, i}$ and $a_{\gamma, ik}$ are not updated during those runs. Besides, $w_{\gamma, i}$ are set to $1/N$.

The intuition for the choice of the number of rounds $K$ of the above procedure is to stop the burn-in algorithm when running an additional round does not yield much improvement in the accuracy of the estimation of $\Sigma_1, \ldots, \Sigma_N$. We use the inhomogeneity factor (see \cite{roberts2009examples} and \cite{rosenthal2011optimal}), a well-established measure of covariance estimation accuracy in the MCMC context, to choose $K$ automatically. We quantify the dissimilarity  between $\Sigma_{k-1, i}$ and $\Sigma_{k,i}$ for $i \in \mathcal{I}$ by their inhomogeneity factor, denoted by $b_{k,i}$, and stop the covariance estimation when this factor drops below a pre-specified threshold $b_{\text{acc}}$ for all $i \in \mathcal{I}$. Details are given in Supplementary Material~B.

\subsection{Further comments}
From the point of view of the fundamental challenges (1-3) discussed in Section \ref{section:introduction}, JAMS deals with (1) through its mode finding stage. Challenges (2) and (3) are addressed via jumps and local moves, respectively. As explained in Section~\ref{subsection:main_algorithm}, the auxiliary variable approach  facilitates moving efficiently between modes as well as accounting for inhomogeneity between them by using different local proposal distributions in different regions.

It is important to point out that the auxiliary variable approach presented above should be thought of as a flexible framework rather than one specific method. The BFGS algorithm used for mode finding could be replaced with another optimisation procedure and similarly, local moves could be performed using a different MCMC sampler, e.g. HMC. One could also consider another scheme for updating the parameters, for example, combining adaptive scaling with covariance matrix estimation (see \cite{vihola2011stability}).

\section{Auxiliary Variable Adaptive MCMC}\label{section:auxiliary_variable}
We introduce a general class of Auxiliary Variable Adaptive MCMC
algorithms, as follows.

Recall that $\pi(\cdot)$ is a fixed target probability density on
$(\mathcal{X}, \mathcal{B}(\mathcal{X}))$. For an auxiliary pair
$(\Phi, \mathcal{B}(\Phi))$, define $\tilde{\mathcal{X}}:=
\mathcal{X} \times \Phi,$ and for an index set $\mathcal{Y}$, consider a family of probability
measures  $\{\tilde{\pi}_{\gamma}(\cdot)\}_{\gamma \in \mathcal{Y}}$ on
$(\tilde{\mathcal{X}},  \mathcal{B}(\tilde{\mathcal{X}})
),$ such that
\begin{equation} \label{marginal_ok}
\tilde{\pi}_{\gamma}(B \times \Phi) = \pi(B) \quad \textrm{for
	every } B \in \mathcal{B}(\mathcal{X})  \textrm{ and } 
\gamma\in \mathcal{Y}.
\end{equation} 
Let $\{ \tilde{P}_{\gamma} \}_{\gamma \in \mathcal{Y}}$ be a collection of
Markov chain transition kernels on $(\tilde{\mathcal{X}},  \mathcal{B}(\tilde{\mathcal{X}})
),$ such that each $\tilde{P}_{\gamma} $ has $\tilde{\pi}_{\gamma}$ as its
invariant distribution  and is Harris ergodic,
i.e. for all 
$\gamma \in \mathcal{Y},$
\begin{equation}\label{ker_erg}
(\tilde{\pi}_{\gamma} \tilde{P}_{\gamma})
(\cdot) = \tilde{\pi}_{\gamma}(\cdot) \textrm{ and }  \lim_{n \to \infty}
\|\tilde{P}^{n}_{\gamma}(\tilde{x}, \cdot) - \tilde{\pi}_{\gamma}(\cdot)\|_{TV}
= 0 \textrm{ for all }  \tilde{x}:= (x,\phi) \in
\tilde{\mathcal{X}}.
\end{equation}
Here
$\| \cdot - \cdot\|_{TV}$ is the usual total variation distance,
defined  
for two probability measures $\mu$ and $\nu$ on a $\sigma-$algebra of
sets $\mathcal{G}$ as $\|\mu(\cdot) -
\nu(\cdot)\|_{TV} = \sup_{B \in
	\mathcal{G}} |\mu(B) - \nu(B)|.$

To define the dynamics of the Auxiliary Variable Adaptive MCMC sequence
$\{(\tilde{X}_n, \Gamma_n)\}_{n=0}^{\infty},$ where $\Gamma$
represents a random variable taking values in $(\mathcal{Y}, \mathcal{B}(\mathcal{Y}))$, denote
its filtration as \[\mathcal{G}_n:=
\sigma \{\tilde{X}_0,\dots, \tilde{X}_n, \Gamma_0, \dots, \Gamma_n
\}.\] Now, the conditional distribution of $\Gamma_{n+1}$ given
$\mathcal{G}_n$ will be specified by the adaptive algorithm being
used, such as Algorithm \ref{alg:mode_jumping}, while the dynamics of the $\tilde{X}$
coordinate follows
\begin{equation}\label{dyn_tilde_x}
\mathbb{P}\big[\tilde{X}_{n+1} \in \tilde{B} | \tilde{X}_{n}=\tilde{x},
\Gamma_n = \gamma, \mathcal{G}_{n-1}\big] = \tilde{P}_{\gamma}(\tilde{x},
\tilde{B}) \textrm{ for } \tilde{x} \in \tilde{\mathcal{X}}, \gamma \in
\mathcal{Y}, \tilde{B} \in \mathcal{B}(\tilde{\mathcal{X}}).
\end{equation} 
Note that depending on the adaptive update rule for $\Gamma_n$, the sequence $\{(\tilde{X}_n, \Gamma_n)\}_{n=0}^{\infty},$
defined above is not necessarily a Markov chain. By
$\tilde{A}_{n}^{\mathcal{G}_t}(\cdot )$ denote the distribution of the $\tilde{\mathcal{X}}$-marginal
of $\{(\tilde{X}_n, \Gamma_n)\}_{n=0}^{\infty}$  at time
$n$, conditionally on the history up to time $t,$ i.e.
\begin{equation*}
\tilde{A}_{n}^{\mathcal{G}_t} (\tilde{B} ):=
\mathbb{P}\big[\tilde{X}_n \in \tilde{B} | \tilde{X}_0 = \tilde{x}_0,
\dots, \tilde{X}_t = \tilde{x}_t,
\Gamma_0 = \gamma_0, \dots, \Gamma_t = \gamma_t \big]
\end{equation*}
for $\tilde{B} \in\mathcal{B}(\tilde{\mathcal{X}})$, and in particular for $t=0$, we shall write
\begin{equation*}
\tilde{A}_{n}^{(\tilde{x}, \gamma)}(\tilde{B}
):=\tilde{A}_{n}^{\mathcal{G}_0} (\tilde{B} ) = 
\mathbb{P}\big[\tilde{X}_n \in \tilde{B} | \tilde{X}_0 = \tilde{x},
\Gamma_0 = \gamma \big] \quad  \textrm{for }\tilde{B} \in\mathcal{B}(\tilde{\mathcal{X}}).
\end{equation*}
By $A_{n}^{\mathcal{G}_t}(\cdot )$ and $A_{n}^{(\tilde{x}, \gamma)}(
\cdot )$ denote the further
marginalisation of $\tilde{A}_{n}^{\mathcal{G}_t}(\cdot )$ and $\tilde{A}_{n}^{(\tilde{x}, \gamma)}(
\cdot )$, respectively, onto the space of interest $\mathcal{X},$ where the
target measure $\pi(\cdot)$ lives, namely
\begin{equation*}
A_{n}^{\mathcal{G}_t}( B ):= \tilde{A}_{n}^{\mathcal{G}_t}( B
\times \Phi) \textrm{ and } A_{n}^{(\tilde{x}, \gamma)} ( B ):= \tilde{A}_{n}^{(\tilde{x}, \gamma)} ( B
\times \Phi), \quad \textrm{for } B \in\mathcal{B}(\mathcal{X}).
\end{equation*}
Finally, in order to define ergodicity of the Auxiliary Variable
Adaptive MCMC, let
\begin{equation*}
T_{n}(\tilde{x}, \gamma) := \|A_{n}^{(\tilde{x}, \gamma)}( \cdot ) -
\pi(\cdot)\|_{TV} = \sup_{B \in \mathcal{B}(\mathcal{X})}
|A_{n}^{(\tilde{x}, \gamma)}( B ) - \pi(B)|.
\end{equation*}

\begin{defi}\label{aux_erg} We say that the Auxiliary Variable
	Adaptive MCMC algorithm generating $\{(\tilde{X}_n,
	\Gamma_n)\}_{n=0}^{\infty},$ is \emph{ergodic}, if
	\begin{equation*}
	\lim_{n \to \infty} T_{n}(\tilde{x}, \gamma) =0 \quad \textrm{for
		all} \quad \tilde{x} \in\tilde{\mathcal{X}}, \gamma \in \mathcal{Y}.
	\end{equation*}
\end{defi}
As we shall see in Section \ref{section:ergodicty_amcmc}, JAMS belongs to the class defined above. There exist other algorithms falling into this category, therefore the results presented in this paper, in particular Theorems \ref{thm:uniform}, \ref{thm:non_uniform} and \ref{thm:lln}, may be useful for analysing their ergodicity. Examples of other algorithms in this class include adaptive parallel tempering \cite{miasojedow2013adaptive} and adaptive versions of pseudo-marginal algorithms \cite{andrieu2009pseudo, andrieu2015convergence}. A more detailed discussion on this may be found in Supplementary Material A. 
\subsection{Theoretical results for the class}
The two main approaches to verifying ergodicity of Adaptive MCMC are based on martingale approximations \cite{MR2260070,  MR3012408, MR2648752} or coupling \cite{roberts2007coupling}. Here we extend the latter to the Auxiliary Variable Adaptive MCMC class by constructing explicit couplings. In particular, ergodicity of this class of algorithms will be verified for the uniform and the non-uniform case, providing results analogous to Theorems 1 and 2 of
\cite{roberts2007coupling}.

For the uniform case analogues
of the usual conditions of Simultaneous Uniform Ergodicity and
Diminishing Adaptation will be required. 

\begin{thm}[Ergodicity -- uniform case] \label{thm:uniform} Consider an Auxiliary Variable
	Adaptive MCMC algorithm on a state space \mbox{$\tilde{\mathcal{X}} =
		\mathcal{X} \times \Phi$,} following dynamics \eqref{dyn_tilde_x} with 
	a family
	of transition kernels $\{ \tilde{P}_{\gamma} \}_{\gamma \in
		\mathcal{Y}}$ satisfying  \eqref{marginal_ok} and \eqref{ker_erg}. If~conditions (a) and (b) below are
	satisfied, then the algorithm is ergodic in the sense of Definition~\ref{aux_erg}.
	\begin{enumerate}
		\item[(a)] (Simultaneous Uniform Ergodicity). For all $\varepsilon >
		0,$ there exists $N= N(\varepsilon) \in \mathbb{N}$ such that
		\[
		\|\tilde{P}_{\gamma}^N(\tilde{x}, \cdot) -
		\tilde{\pi}_{\gamma}(\cdot)\|_{TV} \leq \varepsilon, \qquad
		\textrm{for all} \quad \tilde{x} \in \tilde{\mathcal{X}} \; \textrm{
			and } \; \gamma \in \mathcal{Y}.
		\]  
		\item[(b)] (Diminishing Adaptation). The random variable
		\[
		D_n:= \sup_{\tilde{x} \in \tilde{\mathcal{X}}} \|
		\tilde{P}_{\Gamma_{n+1}}(\tilde{x}, \cdot) - \tilde{P}_{\Gamma_{n}}(\tilde{x}, \cdot) \|_{TV}
		\] converges to $0$ in probability.
	\end{enumerate}
\end{thm}

In fact assumption (a) of Theorem \ref{thm:uniform} can be relaxed. To this end, define the $\varepsilon-$convergence time as
\begin{equation}\label{egn:epsconvtime} M_{\varepsilon}(\tilde{x}, \gamma) := \inf \{k\geq 1: \| \tilde{P}_{\gamma}^k(\tilde{x}, \cdot)-\tilde{\pi}_{\gamma}(\cdot) \|_{TV} \leq \varepsilon \}. \end{equation} It is enough  that the random variable $
M_{\varepsilon}(\tilde{X}_{n}, \Gamma_n)$ is bounded in probability. Precisely, the following ergodicity result holds for the non-uniform case.

\begin{thm}[Ergodicity -- non-uniform case] \label{thm:non_uniform} Consider an Auxiliary Variable
	Adaptive MCMC algorithm, under the assumptions of Theorem \ref{thm:uniform} and replace condition (a) with the following:
	\begin{enumerate}
		\item[(a)](Containment). For all $\varepsilon>0$ and all $\tilde{\delta}>0$, there exists $N=N(\varepsilon, \tilde{ \delta})$ such that 
		\begin{equation}\label{eq:containment_condition_1}
		\mathbb{P}\left(M_{\varepsilon}(\tilde{X}_{n}, \Gamma_n) > N | \tilde{X}_0 =\tilde{x}, \Gamma_0 = \gamma \right) \leq \tilde{ \delta}
		\end{equation}
		for all $n \in \mathbb{N}$.
	\end{enumerate}
	Then the algorithm is ergodic in the sense of Definition \ref{aux_erg}.
\end{thm}

We establish the Weak Law of Large Numbers for the class of Auxiliary Variable Adaptive MCMC algorithms for both the uniform and the non-uniform case. By letting $\Phi$ be a singleton, our result applies to the standard Adaptive MCMC setting and extends the result of \cite{roberts2007coupling} where the WLLN was provided for the uniform case only.

\begin{thm}[WLLN] \label{thm:lln}    Consider an Auxiliary Variable
	Adaptive MCMC algorithm, as in Theorem \ref{thm:non_uniform}, together with assumptions a) and b) of this theorem.  Let $g: \mathcal{X} \to \mathbb{R}$ be a bounded measurable function. Then
	$$
	\frac{\sum_{i=1}^n g(X_i)}{n} \to \pi(g)
	$$
	in probability as $n \to \infty$. 
\end{thm}

While Containment is a weaker condition than Simultaneous Uniform Ergodicity, it is less tractable and in the standard Adaptive MCMC setting drift conditions are typically used to verify it \cite{roberts2007coupling,MR2849670}. Lemma \ref{thm:drift_condition} helps verifying Containment via geometric drift conditions in the Auxiliary Variable framework. The lemma additionally assumes that the adaptation happens on a compact set only (cf. condition e) below). Adapting on a compact set has been theoretically investigated in \cite{MR3404645} and used in certain adaptive Gibbs sampler contexts in \cite{chimisov2018adapting}. We shall use Lemma \ref{thm:drift_condition} as the main tool for establishing ergodic theorems for JAMS.
\begin{lemma}\label{thm:drift_condition}
	Assume that the following conditions are satisfied.
	\begin{enumerate}[a)]
		\item For each $\gamma \in \mathcal{Y}$ $\|\tPg^k(\tilde{x},\cdot) -\tpig(\cdot)\|_{TV} \to 0$ as $k \to \infty$.
		\item There exists $\lambda <1$, $b < \infty$ and a collection of functions  
		$V_{\tpig}: \tilde{\mathcal{X}} \to [1, \infty)$ for $\gamma \in \mathcal{Y}$, such that 
		the following simultaneous drift condition is satisfied:
		
		\begin{equation}\label{eq:main_drift_condition}
		\tPg V_{\tpig}(\tilde{x}) \leq \lambda V_{\tpig}(\tilde{x}) + b   \quad \text{for all }\tilde{x}\in \tilde{\mathcal{X}} \text{ and }\gamma \in \mathcal{Y},
		\end{equation}
		where for $\tilde{x} \in \tilde{\mathcal{X}}$
		$$
		\tPg V_{\tpig}(\tilde{x}) := \E \left(V_{\tpig}(\tilde{X}_{n+1}) \big| \tilde{X}_n = \tilde{x}, \Gamma_n = \gamma\right).
		$$
		Moreover, $V_{\tpig}(\tilde{x})$ is bounded on compact sets as a function of $(\tilde{x}, \gamma)$.	
		
		\item There exist $\delta >0$, $v >2n_0b/(1-\lambda^{n_0})$ and a positive integer $n_0$, 
		such that the following minorisation condition holds: for each $\giY$ we can find a probability measure $\nu_{\gamma}$ on $\tilde{\mathcal{X}}$ satisfying
		\begin{equation}\label{eq:minorisation_condititon}
		\tPg^{n_0}(\tilde{x}, \cdot) \geq \delta \nu_{\gamma}(\cdot) \quad \text{for all } \tilde{x} \text{ with } V_{\tpig}(\tilde{x}) \leq v.
		\end{equation}
		\item $\mathcal{Y}$ is compact in some topology.
		\item There exists a compact set $A$ such that if $X_n\notin A$, then $\Gamma_{n+1} = \Gamma_n$.
		
		\item $\E V_{\tpi_{\Gamma_0}}(\tilde{X}_0) < \infty$.
	\end{enumerate}
	Then the Containment condition \eqref{eq:containment_condition_1} holds. 
\end{lemma}

\subsection{Adaptive Increasingly Rarely version of the class}\label{subsection:airmcmc}
Adaptive Increasingly Rarely (AIR) MCMC algorithms were introduced in \cite{chimisov2018air} as an alternative to classical Adaptive MCMC methods. While they share the same self-tuning properties, their ergodic properties are mathematically easier to analyse and their computational cost of adaptation is smaller. 

The key idea behind the AIR algorithms is to allow the updates of parameters only at pre-specified times $N_j$ with and increasing sequence of lags $n_k$ between them. $N_j$ is therefore defined as
$$
N_j = \sum_{k=1}^j n_k \quad \text{with } N_0 = 0 \text{ and } n_0=0.
$$
For the sequence $\{n_k\}_{k>1}$ \cite{chimisov2018air} proposed using any scheme  that satisfies
$c_2k^{\kappa}  \geq n_k \geq c_1 k^{\kappa}
$
for some positive $c_1$, $c_2$ and $\kappa$. In order to ensure that the random variable
$$
D_n= \sup_{\tilde{x} \in \tilde{\mathcal{X}}} \|
\tilde{P}_{\Gamma_{n+1}}(\tilde{x}, \cdot) - \tilde{P}_{\Gamma_{n}}(\tilde{x}, \cdot) \|_{TV}
$$ converges to $0$ in probability (which is equivalent to Diminishing Adaptation), the following modification is introduced. The updates happen at times $N^{*}_j$, where
$$
N^{*}_j = \sum_{k=1}^j n^{*}_k \quad \text{with } N^{*}_0 = 0 \text{ and } n^{*}_0=0.
$$
and 
$$
n^{*}_k = n_k + \text{Uniform}[0, \lfloor{k^{\kappa^{*}}}\rfloor] \quad \text{for some } \kappa^{*} \in (0, \kappa).
$$
Observe that $D_n$ is only positive if $n+1 \in \{N^{*}_j\}_{j \geq 1}$. Besides, if 
\mbox{$n+1 > N_k$} then
$
\mathbb{P}(D_n >0) \leq \frac{1}{ \lfloor{k^{\kappa^{*}}}\rfloor},
$
so in particular $D_n$ goes to 0 as $n$ tends to infinity. 

We apply the same idea to Auxiliary Variable Adaptive MCMC algorithms, by adapting the parameters of the transition kernels and the target distributions only at times $N_j^*$, as described above, so that Diminishing Adaptation is automatically satisfied for these algorithms. In Section \ref{section:ergodicty_amcmc} we study in detail an AIR version of JAMS (see Algorithm \ref{alg:mode_jumping_2}). 
\section{Ergodicity of the Jumping Adaptive Multimodal Sampler}\label{section:ergodicty_amcmc}
We will use our results from Section \ref{section:auxiliary_variable} to prove ergodicity of JAMS. Firstly observe that this algorithm indeed belongs to the Auxiliary Variable Adaptive MCMC class. To see this, recall that the method utilises  a collection of distributions $\{\tpig(\cdot)\}_{\giY}$ on $\mathcal{\tilde{X}} := \mathcal{X} \times \mathcal{I}$, which corresponds to the notation introduced for the Auxiliary Variable Adaptive MCMC class, with $\Phi=\mathcal{I}$. Indeed, for each $B \in \mathcal{B}(\mathcal{X})$ and $\giY$  we have
$\tpig(B \times \mathcal{I})= \pi(B)$ (see \eqref{eq:jams_main_equation}).

Let $\tpigLi$ denote the kernel associated with the local move around \mbox{mode $i$} and analogously let  $\tpigJi$ be the kernel of the jump to \mbox{mode $i$}. The full transition kernel $\tPg$ is thus defined as
$$
\tPg\left((x,i), (dy,k)\right) := (1-\epsilon)\tpigLi\left((x,i), (dy,k)\right) \delta_{i=k} + \epsilon \agik \tpigJk \left((x,i), (dy,k)\right).
$$
It is easily checked that the acceptance probabilities \eqref{eq:local_move_acceptance} and \eqref{eq:independent_jump_move_acceptance} or \eqref{eq:acc_prob_deterministic} ensure that  detailed balance holds for the above kernels $\tPg$, admitting $\tpig$ as their invariant distributions. They also satisfy the Harris ergodicity condition. The above discussion shows that the algorithm indeed falls into the category of the Auxiliary Variable Adaptive MCMC, so Theorems \ref{thm:uniform} and \ref{thm:non_uniform} can be used to establish its ergodicity.

The main results of this section are stated in Theorems \ref{thm:heavy_tailed_proposal} and \ref{thm:light_tailed_proposal}, which establish convergence of our algorithm to the correct limiting distribution under the uniform and the non-uniform scenario, respectively.
\subsection{Overview of the assumptions}\label{section:overview_assumptions}
In order to prove ergodic results for JAMS, we consider Algorithm \ref{alg:mode_jumping_2}, which is a slightly modified version of Algorithm \ref{alg:mode_jumping}. While being easier to analyse mathematically, it inherits the main properties of Algorithm \ref{alg:mode_jumping}. The modifications are two-fold: firstly, we update the parameters only if the most recent sample $(x_n, i_n)$ is such that $x_n$ belongs to some fixed compact set $A_{i_n}$ and secondly, we adapt them "increasingly rarely" (see Section \ref{subsection:airmcmc}).  If jumps are proposed deterministically, we additionally assume that they are allowed only on "jumping regions" $JR_{\gamma, i}$ defined as
\begin{equation}\label{eq:jumping_region_def}
JR_{\gamma,i} = \{x \in \mathcal{X}: (x-\mu_i)^T\Sigma_{\gamma,i}^{-1}(x-\mu_i) \leq R\} 
\end{equation}
for $i\in \mathcal{I}$  and some  $R>0$. 
Note that equation \eqref{eq:deterministic_proposal} ensures that if $x$ belongs to $JR_{\gamma,i}$ and we propose a deterministic jump from $(x,i)$ to  $(y,k)$, then  $y$ must be in $JR_{\gamma,k}$. Thus the detailed balance condition is satisfied. The reasons for these modifications will become clearer when we present the proofs of the ergodic theorems. 

Even though the theory presented below works for any choice of the compact sets $A_1, \ldots, A_N$, we propose to define these sets in the following way. Recall that the burn-in routine (Algorithm \ref{alg:burn_in_algorithm}) provides the list of mode locations $\{\mu_1, ..., \mu_N\}$ and initial estimates of covariance matrices $\{\Sigma_{\gamma_0,1}, ..., \Sigma_{\gamma_0,N}\}$. By $\lambda_i$ denote the maximum eigenvalue of $\Sigma_{\gamma_0,i}$ and let $\lambda_M = \max\{\lambda_1, ...,\lambda_N\}$. Let $C$ be the convex hull of $\{\mu_1, ..., \mu_N\}$ and $D_C$ its diameter. Define
$$
A_i := \big\{ x \in \mathcal{X}: \|x-\mu_i\| \leq 2D_C + 100 (d\lambda_M)^{1/2} \big\}, 
$$
where $d$ is the dimension of $\mathcal{X}.$  

Observe that Algorithm \ref{alg:mode_jumping_2} is constructed in such a way that all the covariance matrices $\Sigma_{\gamma,i}$ are based on samples belonging to a compact set $A_{i}$. This implies that these matrices are bounded from above. Since we keep adding $\beta I_d$ to the covariance matrix at each step, they are also bounded from  below. Recall also that the covariance matrices for the local proposal distributions are scaled by a fixed factor $2.38^2/d$. Consequently, there exist positive constants $m$ and $M$ for which 

\begin{equation}\label{eq:matrices_estimation}
mI_d \preceq \Sigma_{\gamma,i}\preceq MI_d \text{ and } mI_d \preceq 2.38^2/d \Sigma_{\gamma,i}  \preceq MI_d \text{ for all } \giY \text{ and } i \in \mathcal{I}.
\end{equation}
As for the adaptive scheme for $\wgi$ and $\agik$, we only require that these values be bounded away from 0, i.e. there exist $\epsilon_a$ and $\epsilon_w$ such that 
\begin{equation}\label{eq:weights_epsilon}
\wgi > \epsilon_w \quad \text{and} \quad \agik > \epsilon_a \quad \text{for all } \giY \text{ and } i,k \in \mathcal{I}.
\end{equation}
Therefore, the parameter space $\mathcal{Y}$ may be considered as compact.

\begin{algorithm}[!ht]
	\caption{JAMS: main algorithm, iteration $n+1$}\label{alg:mode_jumping_2}
	\begin{algorithmic}[1]	
		\State {\textbf{Input:} current point $(x_n, i_n)$, list of modes $\{\mu_1, \ldots \mu_N \}$, constant $\epsilon \in (0,1)$, parameter $\gamma_{n} = \lbrace\Sigma_{\gamma_{n},i}, w_{\gamma_{n},i}, a_{\gamma_{n}, ik}\rbrace_{i,k \in \{1, \ldots, N\}}$, empirical means $m_1, \ldots, m_N$ and covariance matrices $S_1, \ldots, S_N$, integer $N_j^* \geq n+1$, next element of the lag sequence $n_{j+1}$.} 
		\State {Generate $u \sim U[0,1]$.}
		\If {$u > \epsilon$}
		\State  {\textbf{Local move}:}\label{local_m}
		\State{Propose a new value $y \sim R_{\gamma_{n}, L, i_n}(x_n, \cdot)$.}
		\State{Accept $y$ with probability $\alpha_{\gamma_{n}, L}\left((x_n, i_n) \to (y, i_n) \right)$}.
		\If {$y$ accepted}
		\State{$(x_{n}, i_{n}) := (y, i_n)$.}
		\Else
		\State{$(x_{n}, i_{n}) := (x_n, i_n)$.}
		\EndIf
		\Else
		\If{$x_n \notin JR_{\gamma_n,i_n}$}
		\Goto{local_m} \Comment{\textbf{(only for deterministic jumps)}}
		\EndIf
		\State{\textbf{Jump move:}}
		\State {Propose a new mode $k \sim (a_{\gamma_n, i1}, \ldots, a_{\gamma_n, iN})$. }
		\State{Propose a new value $y \sim R_{\gamma_{n},J,ik}(x_n, \cdot)$.}
		\State{Accept $(y, k)$ with probability $\alpha_{\gamma_{n}, J}\left((x_n, i_n) \to (y, k) \right)$.}
		\If {$(y, k)$ accepted}
		\State{$(x_{n+1}, i_{n+1}) := (y, k)$.}
		\Else
		\State{$(x_{n+1}, i_{n+1}) := (x_n, i_n)$.}
		\EndIf
		\EndIf
		\If{$x_{n+1} \in A_{i_{n+1}}$} 
		\State{Update the empirical mean $m_{i_{n+1}}$ and covariance matrix $S_{i_{n+1}}$ by \mbox{including $x_{n+1}$.}}
		\If{$n+1$ is equal to $N_j^{*}$}
		\State{Update the parameter $\gamma_{n}$ to $\gamma_{n+1}$ according to Algorithm \ref{alg:params_updates}.}
		\State{Sample $N_{j+1}^* = N_j^{*} + n_{j+1} + \text{Uniform} \left[0, \lfloor(j+1)^{\kappa^*}\rfloor \right]$.  } 
		\EndIf
		\EndIf
		\State \Return {New sample $(x_{n+1}, i_{n+1})$, parameter $\gamma_{n+1}$, $m_{i_{n+1}}$ and $S_{i_{n+1}}$.}
	\end{algorithmic}
\end{algorithm}

\subsection{Theoretical results for JAMS}
We begin with the case when the jump moves are proposed independently from distributions $\RJi$ with heavier tails than the tails of the target distribution $\pi$ for all $i \in \mathcal{I}$ and $\gamma \in \mathcal{Y}$, i.e.
\begin{equation}\label{eq:heavy_tailed_proposal}
\sup_{x \in \mathcal{X}} \sup_{\gamma \in \mathcal{Y}} \frac{\pi(x)}{\RJi(x)} < \infty \quad \text{for  each } i \in \mathcal{I}.
\end{equation}
We prove that under this assumption Simultaneous Uniform Ergodicity is satisfied for Algorithm \ref{alg:mode_jumping_2} and consequently, by Theorem \ref{thm:uniform}, the algorithm is ergodic.
\begin{thm}\label{thm:heavy_tailed_proposal}
	Consider Algorithm  \ref{alg:mode_jumping_2} and assume  that the relationship between the target distribution $\pi$ and the proposal distributions $\RJi$ satisfies \eqref{eq:heavy_tailed_proposal}. Then Algorithm \ref{alg:mode_jumping_2} is ergodic. 
\end{thm}
When the tails of the distribution $\pi$ are heavier then the tails of the proposal distributions $\RJi$, or when the jumps follow the deterministic scheme, Simultaneous Uniform Ergodicity  does not hold. However, it turns out that under some additional regularity conditions Algorithm \ref{alg:mode_jumping_2} is still ergodic, as it satisfies the assumptions of Lemma \ref{thm:drift_condition}. 
\newpage
\begin{thm}\label{thm:light_tailed_proposal} 
	Consider Algorithm  \ref{alg:mode_jumping_2} 
	and assume that the following conditions are satisfied.	
	\begin{enumerate}[a)]
		\item For each $i \in \mathcal{I}, \giY$ the proposal distribution for local moves $\RLi$ follows an elliptical distribution  parametrised by $\Sigma_{\gamma,i}$.  Furthermore, the family of distributions $\RLi(\textbf{0}, \cdot)$, $\giY$, has uniformly bounded probability density functions, and for any compact set $C \subset \mathcal{X}$ we have
		\begin{equation}\label{eq:local_proposals_compact_set_assumption}
		\inf_{x,y\in C} \inf_{\giY}\RLi(x,y) >0 \quad \text{for each } i \in \mathcal{I}.
		\end{equation}
		\item Let $r_{\gamma,i}(x)$ be the rejection set for local moves, i.e. $r_{\gamma,i}(x) := \{y \in \mathcal{X}: \tpi_{\gamma}(y,i)< \tpi_{\gamma}(x,i)\}$. We assume that 
		\begin{equation}\label{eq:assumption_limsup}
		\limsup_{|x| \to \infty} \sup_{\gamma \in \mathcal{Y}} \int_{r_{\gamma,i}(x)}\RLi(x,y) dy < 1 \quad \text{for each } i \in \mathcal{I}.
		\end{equation}
		\item The target distribution $\pi$ is super-exponential, i.e. it is positive with continuous first derivatives and satisfies 
		\begin{equation}\label{eq:super_exponential}
		\lim_{|x| \to \infty} \frac{x}{|x|} \cdot \nabla \log \pi(x) = - \infty.
		\end{equation}
		\item Every $Q_i,$ $i \in \mathcal{I},$ is an elliptical distribution parametrised by $\Sigma_{\gamma,i}$ positive on $\mathcal{X}$ and additionally, the following condition is satisfied:
		\begin{equation}\label{eq:Q_i_assumption}
		\sup_{x\in \mathcal{X}} \frac{Q_i(\mu_i, \Sigma_{\gamma_1,i})(x)}{Q_k(\mu_k, \Sigma_{\gamma_2,k})(x)} < \infty \quad \textrm{for all } i,k \in \mathcal{I} \textrm{ and } \gamma_1, \gamma_2 \in \mathcal{Y}.
		\end{equation}
	\end{enumerate} 	
	Additionally, one of the following two conditions for jump moves holds. 
	\begin{enumerate}[e1)]
		\item Jump moves follow the procedure for deterministic jumps, as described in Section \ref{subsection:main_algorithm}.
		\item Jump moves follow the independent proposal procedure, as described in Section \ref{subsection:main_algorithm}. The proposal distributions for jumps have uniformly bounded probability density functions and satisfy
		\begin{equation}\label{eq:jump_proposals_compact_set_assumption}
		\inf_{x\in B\left(\mu_i,r\right)} \inf_{\giY}\RJi(x) >0 \quad \text{for each } i \in \mathcal{I} \text{ and  some } r>0,
		\end{equation} 
		where $B\left(\mu_i,r\right)$ is a ball of radius $r$ and centre $\mu_i$.
		Moreover, the relationship between the target distribution  $\RJi$ is given by
		\begin{equation}\label{eq:light_tailed_proposal}
		\sup_{x \in \mathcal{X}} \sup_{\gamma \in \mathcal{Y}} \frac{\RJi(x)}{\pi(x)^{s_J}} < \infty \quad \text{for  each } i \in \mathcal{I} \text{ and  some } s_J\in(0,1].
		\end{equation}
	\end{enumerate}
	Then Algorithm \ref{alg:mode_jumping_2} 
	is ergodic.
\end{thm}
When proving the above result, we will refer to the proof of Theorem 4.1 of~\cite{jarner2000geometric}. Assumptions b) and c) are analogues of the regularity conditions considered in \cite{jarner2000geometric}. Condition a) holds automatically for our algorithm if we assume that the proposal distributions for local moves follow either the normal or the $t$ distribution  (see Section \ref{subsection:main_algorithm}) and when \eqref{eq:matrices_estimation} holds. Condition \eqref{eq:jump_proposals_compact_set_assumption} is satisfied if the proposal distributions for jumps follow, for example, the normal distribution. Condition d) can be easily verified if every $Q_i$, $i \in \mathcal{I}$ follows the $t$ distribution with the same number of degrees of freedom.

The result stated below establishes the Weak Law of Large Numbers for our algorithm.
\begin{thm}\label{thm:lln_jams}
	Consider Algorithm \ref{alg:mode_jumping_2} and assume that conditions of either Theorem \ref{thm:heavy_tailed_proposal} or Theorem \ref{thm:light_tailed_proposal} are satisfied. Then the Weak Law of Large Numbers holds for all bounded and measurable functions. 
\end{thm}

\begin{remark}\label{rem:parallel_theory}
	Note that Theorems \ref{thm:heavy_tailed_proposal}, \ref{thm:light_tailed_proposal} and \ref{thm:lln_jams} are based on an assumption that the list of modes is fixed. Let us now consider Algorithm~\ref{alg:mode_jumping_2} in the version with mode finding running in parallel to the main MCMC sampler, as shown in Figure \ref{fig:flowchart}. Assume additionally that 
	\begin{equation}\label{eqn_last_mode} 
	\mathbb{P}(\tau < t) \to 1 \quad \textrm{as} \quad t\to \infty,
	\end{equation} where $\tau$ is the time of adding the last mode.  In this case Theorems \ref{thm:heavy_tailed_proposal},   \ref{thm:light_tailed_proposal} and \ref{thm:lln_jams} still hold. Indeed, as the parallel burn-in algorithm runs independently of JAMS, we can rephrase all the probabilistic limiting statements in the proofs on the set $C_t:= \{\tau < t\}$ and then let $t \to \infty.$
\end{remark}

The following lemmas are useful in verifying assumption b) of Theorem~\ref{thm:light_tailed_proposal}.
\begin{lemma}\label{thm:useful_lemma_1} Let $r(x): = \{y \in \mathcal{X}: \pi(y)< \pi(x)\}$ and $a(x): = \{y \in \mathcal{X}: \pi(y)\geq \pi(x)\}$.
	Consider Algorithm \ref{alg:mode_jumping_2} together with  conditions a), c) and d) of Theorem \ref{thm:light_tailed_proposal}. Assume additionally that for some $\gamma^{*}\in \mathcal{Y}$
	\begin{equation}\label{eq:assumption_single_gamma} 
	\limsup_{|x| \to \infty}  \int_{r(x)}R_{\gamma^{*}, L, i}(x,y) dy < 1 \quad \text{for each } i \in \mathcal{I}.
	\end{equation}
	Then condition \eqref{eq:assumption_limsup} holds.
\end{lemma}
\begin{lemma}\label{thm:useful_lemma_2}
	Consider Algorithm \ref{alg:mode_jumping_2} together with  conditions a), c) and d) of Theorem \ref{thm:light_tailed_proposal}. Assume additionally that the target distribution $\pi$ satisfies
	\begin{equation}\label{eq:simplifying_assumption}
	\limsup_{|x| \to \infty}\frac{x}{|x|} \cdot \frac{\nabla \pi(x)}{|\nabla \pi(x)|} < 0.
	\end{equation}
	Then condition \eqref{eq:assumption_limsup} holds.
\end{lemma}

The following corollary shows Algorithm \ref{alg:mode_jumping_2} in a standard setting is successful at targeting mixtures of normal distributions. 
\begin{cor}\label{thm:corrollary}
	Let the target distribution $\pi$ be given by 
	$$
	\pi(x) \propto w_1 \exp\left(-p_1(x) \right) + \ldots + w_n \exp\left(-p_n(x) \right),
	$$
	where $w_i>0$ and $p_i$ is a polynomial of order $\geq 2$  for each $i = 1, \ldots, n$. If additionally $Q_i$ for $i \in \mathcal{I}$ follows the multivariate $t$ distribution with the same number of degrees of freedom, and $\RLi(\textbf{0}, \cdot)$ follows the normal distribution, the assumptions of Lemma \ref{thm:useful_lemma_2} are satisfied.
\end{cor}

\section{Examples}\label{section:examples} 
In this section we present  empirical results for our method (Algorithm~\ref{alg:mode_jumping} preceded by the Algorithm \ref{alg:burn_in_algorithm}).  We test its performance on three examples -- the first one is a mixture of two Gaussians motivated by \cite{woodard2009sufficient}; the second one is a mixture of fifteen multivariate $t$ distributions and five banana-shaped ones; the third one is a Bayesian model for sensor network localisation. Our implementation admits three versions, varying in the way the jumps between modes are performed. In particular, we consider here the deterministic jump and two independent proposal jumps, with Gaussian and $t$-distributed proposals.

Additionally, we compare the performance of our algorithm against adaptive parallel tempering   \cite{miasojedow2013adaptive}, which was chosen here as it is the refined version of the most commonly used MCMC method for multimodal distributions (parallel tempering). What is more, this algorithm
has a generic implementation, where the user only needs to provide the target density function. In order to make a comparison between the efficiency of these algorithms, among other things, we analyse the Root Mean Square Error  (RMSE) divided by the square root of the dimension of the state space, given a computational budget. We measure the computational cost by the number of evaluations of the target distribution (and its gradient, if applicable), as this is typically the dominating factor in real data examples. Herein we define RMSE as the Euclidian distance between the true $d$-dimensional expected value (if known) and its empirical estimate based on MCMC samples.

In order to depict the variability in the results delivered by both methods, each simulation was repeated 20 times. For exact settings of the experiments, as well as some additional results, we refer the reader to Supplementary Material B.

\subsection{Mixture of Gaussians}\label{section:mixture_of_gaussians}
The following target density was studied by~\cite{woodard2009sufficient}:
\begin{equation}\label{eq:woodard_example_target}
\pi(x) = \frac{1}{2} N\left(-\underbrace{(1, \ldots, 1)}_{d} ,\sigma_1^2 I_d \right) + \frac{1}{2} N\left(\underbrace{(1, \ldots, 1)}_{d} ,\sigma_2^2 I_d \right)
\end{equation}
for $\sigma_1  \neq \sigma_2$. In particular, they showed that the parallel tempering algorithm will tend to stay in the wider mode and, if started in the wider mode, may take a long time before getting to the more narrow one. We looked at the results  for the target distribution \eqref{eq:woodard_example_target} in several different dimensions $d$ ranging between 10 and 200, for $\sigma_1^2 =0.5\sqrt{d/100}$ and   $\sigma_2^2 =\sqrt{d/100}$. The results for our method shown below are based on 500,000 iterations of the main algorithm, preceded by the burn-in algorithm including 1500 BFGS runs. The length of the covariance matrix estimation was chosen automatically using the rule described in Supplementary Material B and varied between 3000 iterations (for $d=10$) to 1,023,000 iterations (for $d=200$) per mode. For dimensions $d=10$ and $d=20$ we ran also the adaptive parallel tempering (APT) algorithm, with 700,000 iterations and 5 temperatures.  Overall this requires 3,500,000 evaluations of the target density that cannot be performed in parallel, despite the name of the method, as the communication between chains running at different temperatures is needed after every iteration. In the light of the tendency of the parallel tempering algorithm to stay in wider modes, each time the APT algorithm was started in $-\underbrace{(1, \ldots, 1)}_{d} \in \mathbb{R}^d$. In order to base our analysis on the same sample size of 500,000 for the two methods, in case of adaptive parallel tempering we applied an initial burn-in period of 200,000 steps.

\begin{figure}[t]
	\begin{center}
		\begin{subfigure}{0.8\textwidth}
			\centering
			\includegraphics[width=10cm]{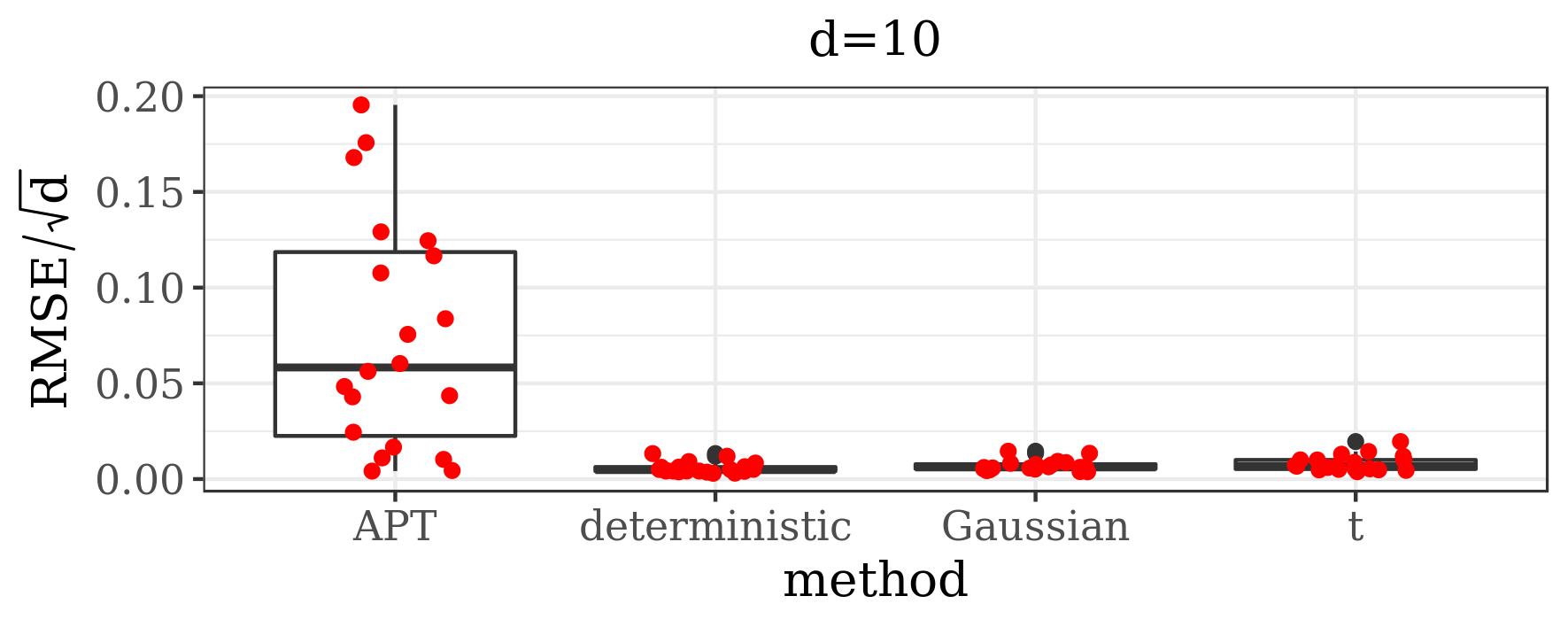} 
		\end{subfigure}
		\begin{subfigure}{0.8\textwidth}
			\centering
			\includegraphics[width=10cm]{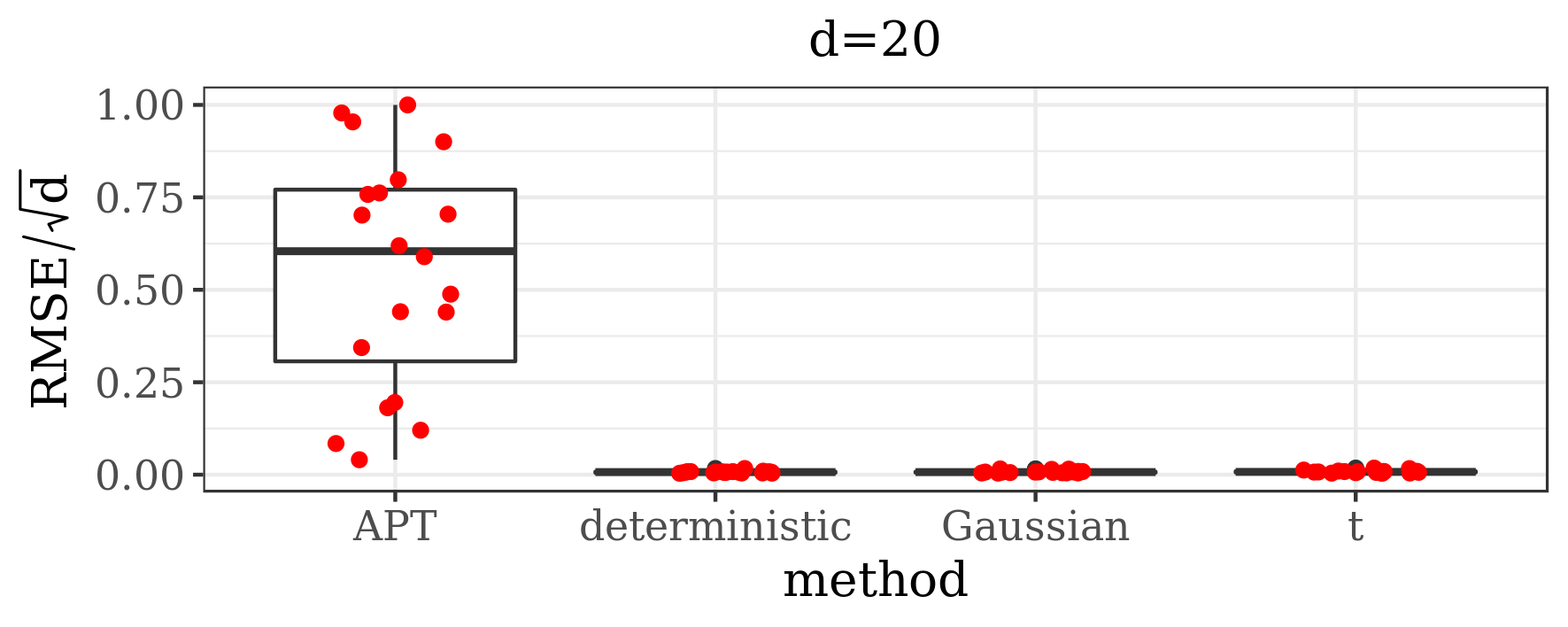} 
			
		\end{subfigure} 
		\caption{\footnotesize Boxplots  of the values of $\text{RMSE}/\sqrt{d}$ for the mixture of Gaussians across 20 runs of the experiment, dimensions 10 and 20. We compare the results of APT with the three JAMS versions: deterministic, Gaussian and $t$-distributed jumps. Note different scales on the $y$-axis.} 
		\label{fig:gaussian_mixture_boxplots_dim_10_20}
	\end{center}
\end{figure}
\FloatBarrier

\begin{table}[ht]
	\centering
	\begin{tabular}{|l|rr|rr|rr|}
		\hline
		& \multicolumn{2}{c|}{deterministic} & \multicolumn{2}{c|}{Gaussian}& \multicolumn{2}{c|}{$t$-distributed} \\ 
		\hline
		& Lowest & Highest & Lowest & Highest & Lowest & Highest \\
		d=10 & 0.98 & 0.99 & 0.85 & 0.87 & 0.71 & 0.73 \\ 
		d=20 & 0.98 & 0.99 & 0.79 & 0.83 & 0.66 & 0.68 \\ 
		d=80 & 0.91 & 0.98 & 0.23 & 0.41 & 0.24 & 0.39 \\ 
		d=130 & 0.72 & 0.98 & 0.04 & 0.13 & 0.06 & 0.15 \\ 
		d=160 & 0.79 & 0.97 & 0.01 & 0.07 & 0.03 & 0.07 \\ 
		d=200 & 0.64 & 0.97 & 0.01 & 0.05 & 0.02 & 0.06 \\ 
		\hline
	\end{tabular}
	\caption{The lowest and the highest value (across 20 runs of the experiment) of the  acceptance rates of jump moves between the two modes for the mixture of Gaussians for different jump methods and dimensions.}\label{table:gaussian_mixture}
\end{table}

\begin{figure}[t]
	\begin{center}
		\begin{subfigure}{0.49\textwidth}
			\centering
			\includegraphics[width=7cm]{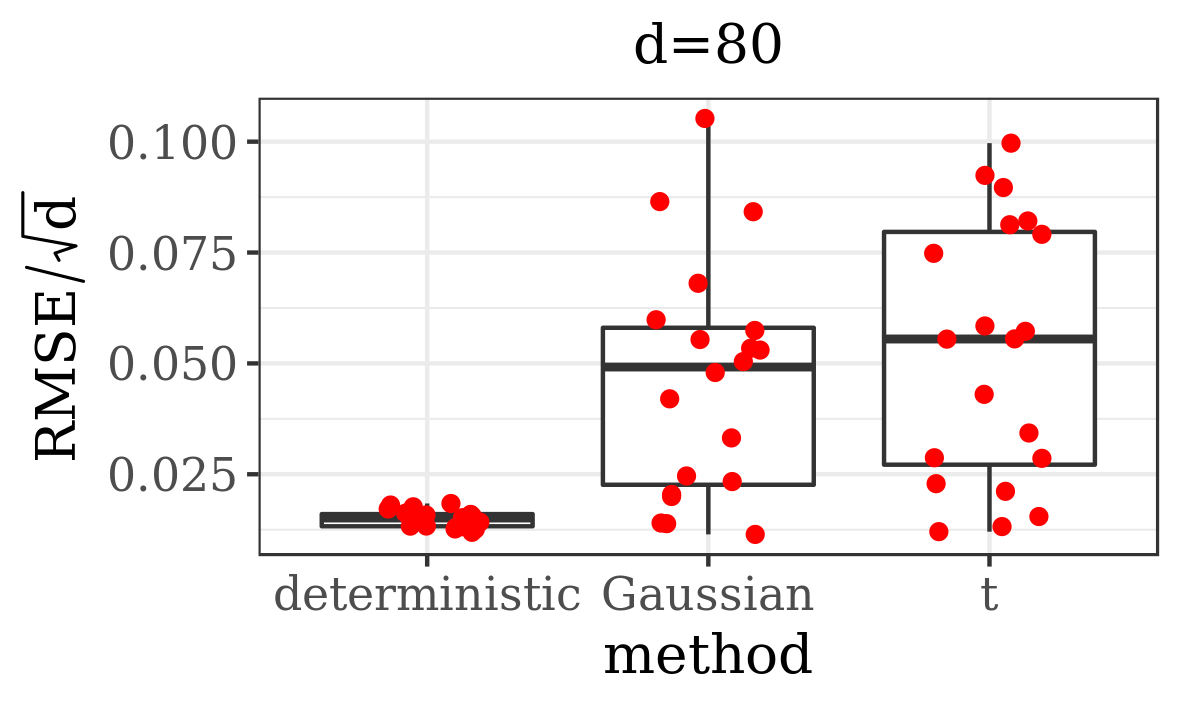} 
		\end{subfigure}
		\begin{subfigure}{0.49\textwidth}
			\centering
			\includegraphics[width=7cm]{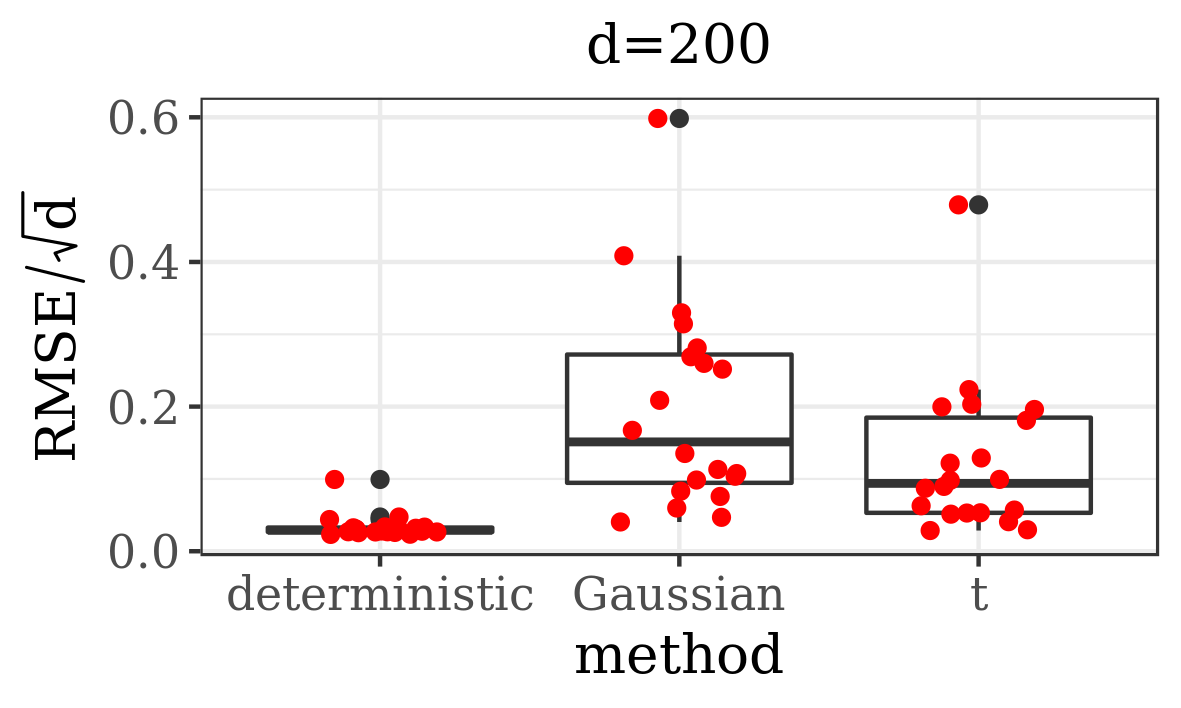} 
		\end{subfigure} 
		\caption{\footnotesize Boxplots  of the values of $\text{RMSE}/\sqrt{d}$ for the mixture of Gaussians across 20 runs of the experiment, dimensions 80 and 200.  Note different scales on the $y$-axis.} 
		\label{fig:gaussian_mixture_boxplots_dim_80_200}
	\end{center}
\end{figure}

\begin{figure}[t]
	\begin{center}
		\begin{subfigure}{1\textwidth}
			\includegraphics[width=14cm]{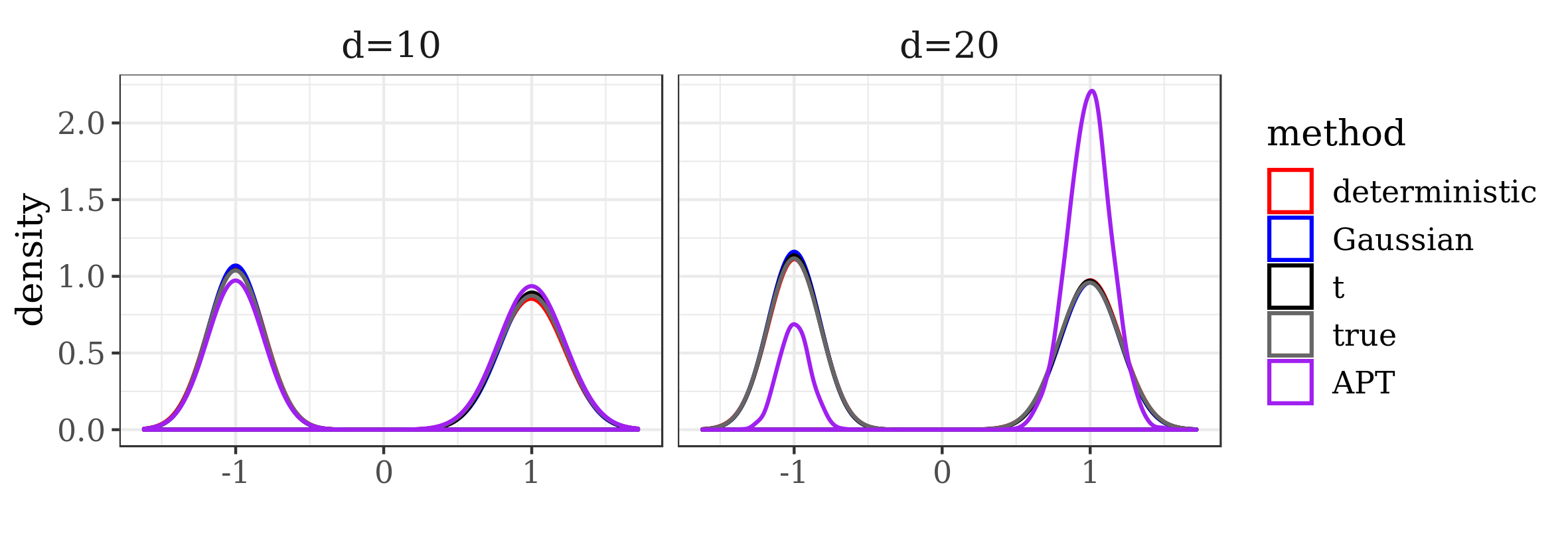} 
		\end{subfigure}
		\begin{subfigure}{1\textwidth}
			\includegraphics[width=14cm]{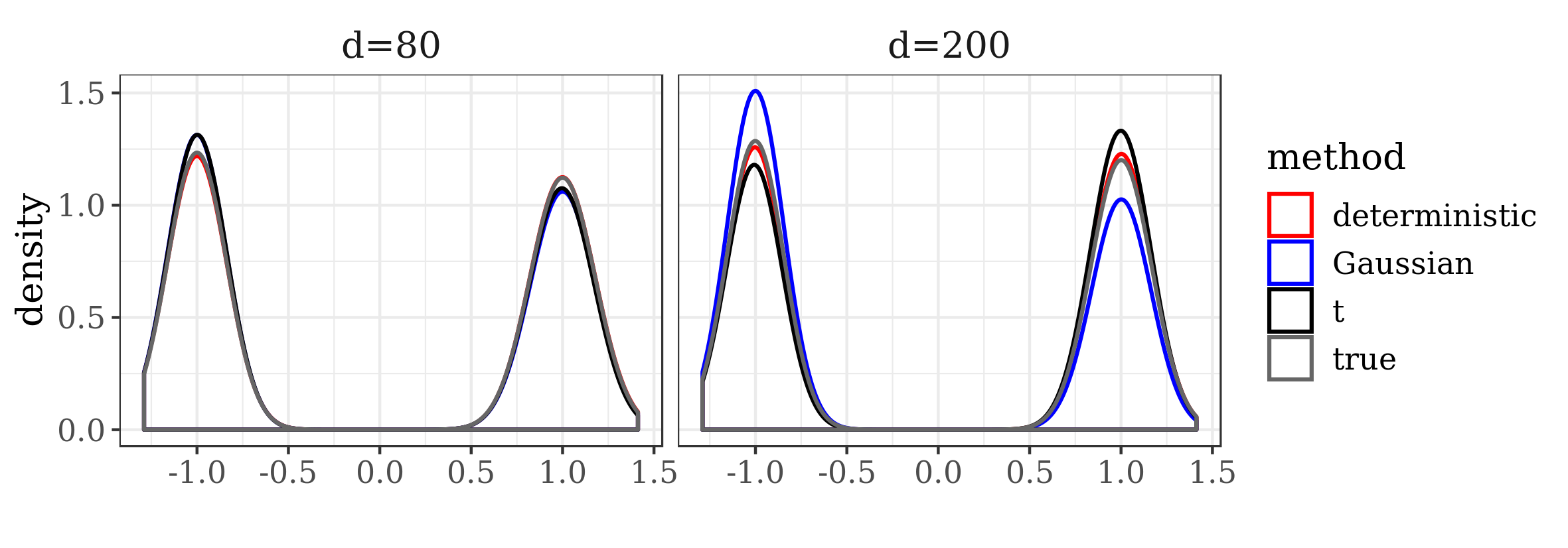} 
		\end{subfigure}
		\caption{\footnotesize Density plots for dimensions 10, 20, 80 and 200. The upper panel shows a comparison between APT and the three JAMS versions. The simulations chosen for the analysis correspond to the median value of RMSE across 20 experiments (the tenth largest value of RMSE).}\label{fig:gaussian_mixture_density_plot}
	\end{center}
\end{figure}
\FloatBarrier
The results presented in the boxplots of Figure \ref{fig:gaussian_mixture_boxplots_dim_10_20}, as well as the 
upper panel of
density plots (Figure \ref{fig:gaussian_mixture_density_plot}) show that our method outperforms adaptive parallel tempering on this example, even when the latter method is given a much larger computational budget. The summary of the acceptance rates of the jump moves presented in Table \ref{table:gaussian_mixture} demonstrates that the algorithm preserves good mixing between the modes in all its jump versions up to dimension 80. It is remarkable that the deterministic jump ensures excellent mixing even in much higher dimensions, outperforming the remaining two methods (see Figure \ref{fig:gaussian_mixture_boxplots_dim_80_200} and the lower panel of Figure \ref{fig:gaussian_mixture_density_plot}), with the acceptance rate between 0.64 and 0.97 in dimension 200.

\subsection{Mixture of $t$ and banana-shaped distributions}\label{section:mixture_bananas}
A classic example of a multimodal distribution is a mixture of 20 bivariate Gaussian distributions introduced in \cite{kou2006discussion} (in two versions, with equal and unequal weights and covariance matrices). It was later studied also by \cite{miasojedow2013adaptive} and \cite{tak2017repelling}.
Our algorithm works well on both versions, however, since the example is relatively simple and the performance of the existing methods on it is already satisfying, we do not expect our method to yield much improvement. Therefore, we decided to modify this example in the way described below in order to make it more challenging. Instead of the Gaussian distribution, the first five modes follow the banana-shaped distribution with $t$ tails and the remaining ones -- multivariate $t$ with 7 degrees of freedom and the covariance matrices $0.01 \sqrt{d}  I_d$, where $d$ is the dimension (the covariance matrices in the original example were given by $0.01 I_2$). The weights are assumed to be equal to 0.05. We consider dimensions $d=10$ and $d=20$ by repeating the original coordinates of the centres of the modes five and ten times, respectively.

Recall the definition of the $d$-dimensional banana-shaped distribution introduced by \cite{haario1999adaptive}\footnote{Originally in the paper by Haario et. al. \cite{haario1999adaptive} the function $f$ was the density of the Gaussian distribution $N(\textbf{0},C)$.}. Let $f$ be the density of the centred $t$ distribution with 7 degrees of freedom and shape matrix $C$, for $C = \text{diag}(100,\underbrace{1, \ldots,1}_{d-1})$. Then the density of the banana-shaped distribution (with $t$-tails) is given by
$$
f_b = f\circ \phi_b,
$$
where 
\begin{equation}\label{eq:banana_distribution_definition}
\phi_b(x_1, \ldots, x_d) = (x_1, x_2+bx_1^2-100b, x_3, \ldots, x_d).
\end{equation}
In order to decrease the variance of the banana-shaped elements of the mixture, we used the following transformation of $f_b$ (setting $b=0.03$)
$$
\tilde{f}_{0.03}(x_1, \ldots, x_d) = \left(20/\sqrt[4]d\right)^d  f_{0.03} \left(20/\sqrt[4]d \cdot(x_1, \ldots,  x_d) \right).
$$
Furthermore,  the formula on the second coordinate of \eqref{eq:banana_distribution_definition} was assigned to coordinate 2, 4, 6, 8, 10 for mode 1, 2, 3, 4 and 5, respectively.

The results below are based on 500,000 iterations, preceded by 40,000 BFGS runs. The number of iterations of the covariance matrix estimation varied between 7,000 and 15,000 steps per mode for dimension $d=10$ and between 15,000 and 63,000 steps per mode for dimension $d=20$. For adaptive parallel tempering we used 2,100,000 iterations and 5 temperatures. We applied an initial burn-in period of 600,000 steps and we thinned the chain keeping every third sample. 

In Supplementary Material B we present results for the same example obtained using JAMS in dimensions $d=50$ and $d=80$ assuming that the modes  of the target distribution are known, since mode finding (in particular, getting to each basin of attraction) is the main bottleneck for this example.

\begin{figure}[t]
	\begin{center}
		\begin{subfigure}{0.8\textwidth}
			\centering
			\includegraphics[width=10cm]{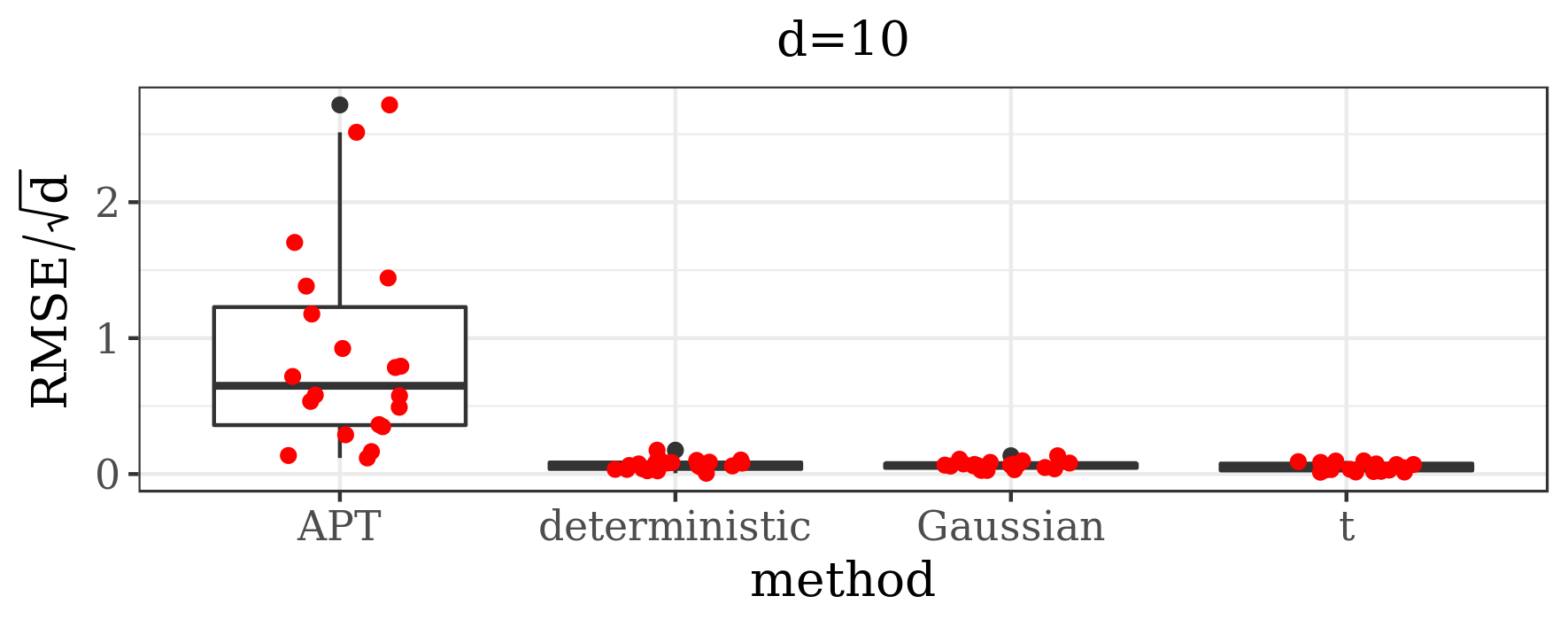} 
		\end{subfigure}
		\begin{subfigure}{0.8\textwidth}
			\centering
			\includegraphics[width=10cm]{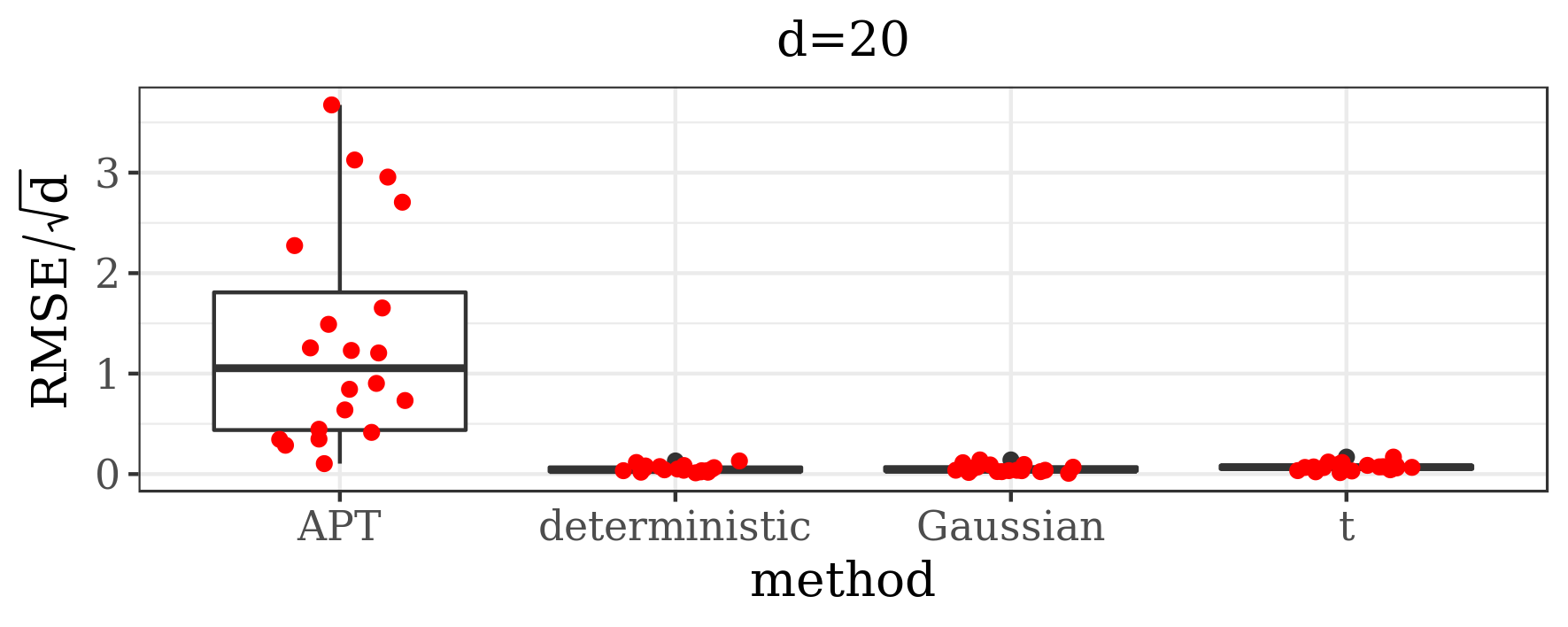} 
			
		\end{subfigure} 
		\caption{\footnotesize Boxplots  of the values of $\text{RMSE}/\sqrt{d}$ for the mixture of banana-shaped and $t$-distributions across 20 runs of the experiment, dimensions 10 and 20. We compare the results of APT with the three JAMS versions: deterministic, Gaussian and $t$-distributed jumps. Note different scales on the $y$-axis.} 
		\label{fig:banana_mixture_boxplots_dim_10_20}
	\end{center}
\end{figure}
\FloatBarrier

\begin{table}[ht]
	\begin{center}
		\begin{tabular}{|l|rr|rr|rr|}
			\hline
			& \multicolumn{2}{c|}{deterministic} & \multicolumn{2}{c|}{Gaussian}& \multicolumn{2}{c|}{$t$-distributed} \\ 
			\hline
			& Lowest & Highest & Lowest & Highest & Lowest & Highest \\
			d=10 & 0.27 & 0.77 & 0.12 & 0.52 & 0.20 & 0.65 \\ 
			d=20 & 0.20 & 0.75 & 0.09 & 0.35 & 0.11 & 0.48 \\ 
			\hline
		\end{tabular}
	\end{center}
	\caption{The lowest and the highest value (across 20 runs of the experiment) of the acceptance rate of jumps from a given mode, dimensions 10 and 20.}\label{table:banana_mixture}
\end{table}

For dimensions $d=10$ and $d=20$ all modes were found by the BFGS runs in each of the 20 simulations. Even though the banana-shaped modes are highly skewed, our method exhibits good between-mode mixing properties, as shown in Table \ref{table:banana_mixture}. Figure  \ref{fig:banana_mixture_boxplots_dim_10_20} illustrates that the empirical means based on JAMS samples approximate well the true expected value of the target distribution, consistently across all experiments, and that our method significantly outperforms APT with a smaller computational cost.

\subsection{Sensor network localisation}\label{section:sensor_network}
We consider here an example from \cite{ihler2005nonparametric}, analysed later by \cite{ahn2013distributed}, \cite{lan2014wormhole} and, in a modified version by \cite{tak2017repelling}. There are 11 sensors with locations $x_1, \ldots, x_{11}$ scattered on a space $[0,1]^2$. The locations of sensors $x_1, \ldots, x_8$ are unknown, the remaining three locations are known. For any two sensors $i$ and $j$ we observe the distance $y_{ij}$ between them with probability $\exp\left(- \frac{\| x_i - x_j \|^2}{2\times 0.3^2}\right)$. Once observed, the distance $y_{ij}$ follows the normal distribution given by $y_{ij} \sim N\left(\| x_i - x_j \|, 0.02^2\right)$. Let $w_{ij}$ be equal to 1 when $y_{ij}$ is observed and $0$ otherwise, and denote $y:= \{y_{ij}\}$ and $w:=\{w_{ij}\}$. The goal of the study is to make inference about the unknown locations $x_i = (z_{i1}, z_{i2})$ for $i = 1, \ldots 8$ given $y$ and $w$. Following \cite{ahn2013distributed} and \cite{lan2014wormhole} we put an improper uniform prior on each of the coordinates $z_{i1}$ and $z_{i2}$ for $i = 1, \ldots 8$. 
The resulting posterior distribution is given by
$$
\pi(x_1, \ldots x_8 \big| y, w) \propto \prod_{\substack{j = 2, \ldots, 11\\
		i = 1, \ldots, 8\\ i <j}} f_{ij}\left(x_i, x_j|y_{ij}, w_{ij}\right),
$$ 
where
$$
f_{ij}\left(x_i, x_j|y_{ij}, w_{ij}\right) = \begin{cases} \exp\left(- \frac{\| x_i - x_j \|^2}{2\times 0.3^2}\right) \exp\left(- \frac{\left(y_{ij} -\| x_i - x_j \|\right)^2}{2\times 0.02^2}\right)
& \text{if } w_{ij}= 1,\\
1-\exp\left(- \frac{\| x_i - x_j \|^2}{2\times 0.3^2}\right)              & \text{otherwise.}
\end{cases}
$$
Since there are few observed distances with known locations (see: top left panel of Figure \ref{fig:sensor_network_scatterplot}), the model is non-identifiable which results in multimodality of the posterior distribution.

We ran JAMS on this example for 500,000 iterations of the main algorithm. This was preceded by 10,000 BFGS runs and covariance matrix estimation (between 7000 and 15,0000 iterations per mode).  For parallel tempering we used 700,000 iterations (with a burn-in period of 200,000) and 4 temperatures. If JAMS is implemented on 8 cores, this means that running an APT simulation is about twice as costly as running a JAMS one (see Supplementary Material B for details).

Despite the fact that for all 20 APT experiments the acceptance rates at all temperature levels, as well as for between-temperature swaps, converged to the optimal acceptance rate 0.234 (see \cite{atchade2011towards}), the behaviour of this algorithm was unstable. As shown in Figure~\ref{fig:sensor_network_scatterplot}, in case of APT the estimation of the location of sensor 1 depends on the starting point. In case of JAMS, both modes for $x_{1}$ (in red) are represented.
\begin{figure}[t]
	\begin{center}
		\includegraphics[width=12cm]{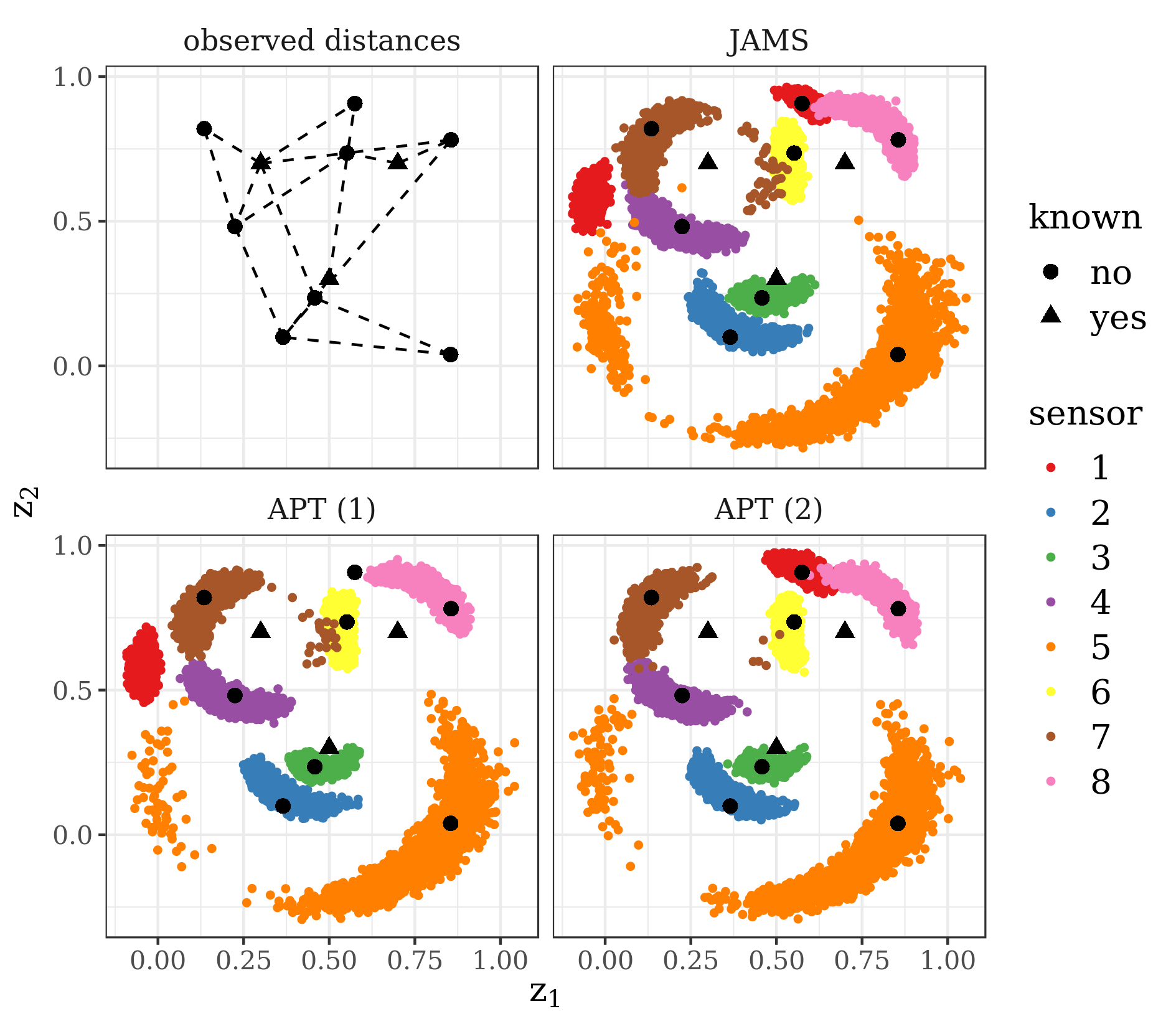} 
	\end{center}
	\caption{\footnotesize Black triangles and dots denote true locations of the sensors with known and unknown locations, respectively. Top left panel: dashed lines represent observed distances between sensors. Top right panel: posterior samples obtained using JAMS (with Gaussian jumps) for locations $x_1, \ldots x_8$. Bottom panels: posterior samples using APT for two different starting points.}
	\label{fig:sensor_network_scatterplot}
\end{figure}
\FloatBarrier
Figure \ref{fig:sensor_network_boxplot} illustrates stability of JAMS across all experiments and jump methods. In Supplementary Material B we assign an even higher computational budget to adaptive parallel tempering allowing for 5 temperatures and observe a substantial improvement in mixing and stability, but the results are still worse than those of JAMS.

\begin{figure}[ht]
	\begin{center}
		\centering
		\includegraphics[width=10cm]{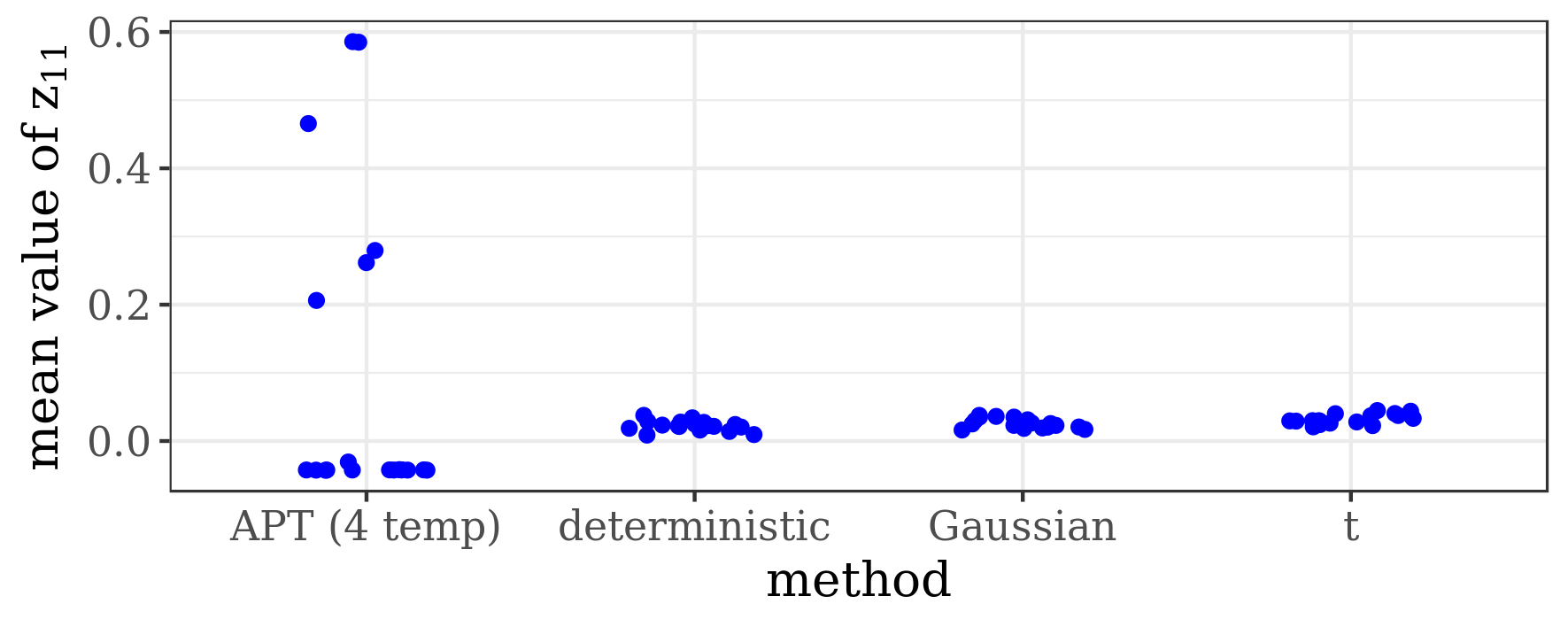} 
		\caption{\footnotesize A boxplot of the mean value of the first coordinate of sensor 1 for 20 runs of the experiment for APT (with 4 temperature levels) and three versions of JAMS.} 
		\label{fig:sensor_network_boxplot}
	\end{center}
\end{figure}

\section{Summary and discussion}\label{section:summary}
The approach we proposed here is based on three fundamental ideas. Firstly, we split the task into mode finding and sampling from the target distribution. Secondly, we base our algorithm on local moves responsible for mixing within the same mode and jumps that facilitate crossing the low probability barriers between the modes. Finally, we account for inhomogeneity between the modes by using different proposal distributions for local moves at each mode and adapting their parameters separately. Similarly, the jump moves account for the difference in geometry of the two involved modes. This is possible thanks to the auxiliary variable approach which enables assigning each MCMC sample to one of the modes and ensuring that it is unlikely to escape to another mode via local moves. This improves over the popular tempering-based approaches, which do not have the mechanism of controlling the mode at each step, and therefore their adaptive versions \cite{miasojedow2013adaptive} only learn the global covariance matrix rather than the local ones. This is highly inefficient if the shapes of the modes are very distinct and results in exponential efficiency decay. 

The optimisation-based approaches are naturally well-suited for the task of collecting the MCMC samples separately for each mode and learning the covariance matrices on this basis. However, the approaches known in the literature do not have a suitable framework for adaptation and tend to be either very costly (e.g. \cite{zhou2011multi}) or to ignore the issue of the possibility of moving between the modes via local steps (e.g.\cite{ahn2013distributed}). Moreover, some of the other fundamental issues of optimisation-based methods have not been systematically addressed by the researchers so far. These include an efficient design of the mode finding phase, distinguishing between newly discovered modes and replicated ones, as well as adapting beyond the infrequent regeneration times, which does not require case-specific calculations. We hope that the method we proposed will fill this gap.

Furthermore, an important advantage of our approach from the point of view of the modern compute resources is that a large part of the algorithm can be implemented on multiple cores.

To develop a methodological approach and prove ergodic results for our algorithm, we  introduced the Auxiliary Variable Adaptive MCMC class. As discussed briefly in Section \ref{section:auxiliary_variable}, there are other adaptive algorithms falling in this category, so our theoretical results may potentially be useful beyond the scope of the Jumping Adaptive Multimodal Sampler. We have also shown that the Auxiliary Variable Adaptive MCMC methods enjoy robust ergodicity properties analogous to Adaptive MCMC under essentially the same well-studied regularity conditions.

Currently the main bottleneck of the method is mode finding, and in particular, sampling  starting points for optimisation runs in such a way that there is at least one point in the basin of attraction of each mode. Therefore in our future work we will focus on designing more efficient algorithms for identifying high probability regions.

\section*{Acknowledgements}
We thank Louis Aslett, Ewan Cameron, Arnaud Doucet, Geoff Nicholls and Jeffrey Rosenthal for helpful comments, and Shiwei Lan for pointing us to the data for the sensor network example. We would also like to thank Radu Craiu for providing the data set for the LOH example considered in Supplementary Material B.

\begin{center}
\large \textbf{\textsc{Supplementary Material A}}
\end{center}
 In Section \ref{section:proofs_3} we present the proofs of our theoretical results stated in Section \ref{section:auxiliary_variable}. In Section \ref{section:proofs_4} we prove the results stated in Section~\ref{section:ergodicty_amcmc}. In Section \ref{section:other_algs} we give some comments about other algorithms in the Auxiliary Variable Adaptive MCMC class.

\section{Proofs for Section 3}\label{section:proofs_3}
To prove our results presented in Section \ref{section:auxiliary_variable} we will use the coupling construction analogous to \cite{roberts2007coupling} (see also
\cite{roberts2013note} for a more rigorous presentation). 

Our proofs will be rigorous and will rely on an explicit coupling construction. The
more complex setting of Auxiliary Variable Adaptive MCMC necessitates a few
preliminary steps: to interpolate between the adaptive process and the target distribution we shall construct two processes to be thought of as "Markovian" and "intermediate". These processes will facilitate application of the triangle inequality in the proofs. 

Recall $\{(\tilde{X}_n,
\Gamma_n)\}_{n=0}^{\infty},$ the adaptive process on $\tilde{X}:= \mathcal{X}\times\Phi$ defined in Section \ref{section:auxiliary_variable} with dynamics governed by equation \eqref{dyn_tilde_x}.
On the same probability space define two
additional sequences, namely $\{(\tilde{X}^{m(t^*)}_n,
\Gamma ^{m(t^*)}_n)\}_{n=0}^{\infty}$ and $\{(\tilde{X}^{i(t^*,\kappa)}_n,
\Gamma ^{i(t^*,\kappa)}_n)\}_{n=0}^{\infty}$  which are identical to 
$\{(\tilde{X}_n,
\Gamma_n)\}_{n=0}^{\infty},$ before
pre-specified time $t^*$, i.e.
\begin{eqnarray*}
	(\tilde{X}^{m(t^*)}_n,
	\Gamma ^{m(t^*)}_n) \; = \; (\tilde{X}^{i(t^*,\kappa)}_n,
	\Gamma ^{i(t^*,\kappa)}_n) & := &
	(\tilde{X}_n,
	\Gamma_n) \qquad  \textrm{ for } \; n \leq t^*. \end{eqnarray*}
After time  $t^*,$ the adaptive parameter $\Gamma^{m(t^*)}$ of  $\{(\tilde{X}^{m(t^*)}_n,
\Gamma ^{m(t^*)}_n)\}_{n=0}^{\infty}$ freezes and
$\tilde{X}^{m(t^*)}_n $ becomes a Markov chain with the marginal
dynamics defined for $n+1 > t^*$ as:
\begin{eqnarray} \label{Markov_seq_G}
\Gamma ^{m(t^*)}_{n+1} & := &  \Gamma ^{m(t^*)}_{n} \; \big(= \Gamma
^{m(t^*)}_{t^*} \big), \\ \label{Markov_seq_X}
\mathbb{P}\left[ \tilde{X}^{m(t^*)}_{n+1} \in \tilde{B} \; \big| \;
\tilde{X}^{m(t^*)}_{n}=\tilde{x}, \; \Gamma
^{m(t^*)}_{t^*} = \gamma \right] & = & \tilde{P}_{\gamma}(\tilde{x}, \tilde{B}),
\qquad \tilde{B} \in \mathcal{B}(\tilde{\mathcal{X}}).
\end{eqnarray}

The second sequence  $\{(\tilde{X}^{i(t^*,\kappa)}_n,
\Gamma ^{i(t^*,\kappa)}_n)\}_{n=0}^{\infty}$ interpolates between  $\{(\tilde{X}_n,
\Gamma_n)\}_{n=0}^{\infty}$ and $\{(\tilde{X}^{m(t^*)}_n,
\Gamma ^{m(t^*)}_n)\}_{n=0}^{\infty}$. We first define the
dynamics of $\Gamma^{i(t^*,\kappa)}$ for $n+1 > t^*,$ as:
\begin{eqnarray} \label{Inter_seq_G}
\Gamma ^{i(t^*,\kappa)}_{n+1} & := &  \left\{  \begin{array}{ll} 
\Gamma_{n+1} & \textrm{if } \; \sup_{\tilde{x} \in
	\tilde{\mathcal{X}}} \| \tilde{P}_{\Gamma_{n+1}}(\tilde{x},
\cdot) - \tilde{P}_{\Gamma ^{i(t^*,\kappa)}_{n}}(\tilde{x},
\cdot) \|_{TV} \leq \kappa, 
\\
\Gamma ^{i(t^*,\kappa)}_{n} & \textrm{otherwise;} \end{array} \right.
\end{eqnarray}
and define an auxiliary stopping time that records decoupling of
$\Gamma_{n}$ and $\Gamma ^{i(t^*,\kappa)}_{n}$ as 
\begin{eqnarray}
\tau_{i(t^*,\kappa)} & :=& \min\{n: \Gamma ^{i(t^*,\kappa)}_{n} \neq \Gamma_{n}\},
\end{eqnarray}
with the convention $\min \emptyset = \infty.$ Now, define the
dynamics of $\tilde{X}^{i(t^*, \kappa)}_{n}$ as:
\begin{eqnarray} \label{inter_seq_X_1}
\tilde{X}^{i(t^*, \kappa)}_{n} & := & \tilde{X}_{n}  \qquad 
\textrm{for } \; n \leq
\tau_{i(t^*,\kappa)}, \textrm{ and}
\\ \label{inter_seq_X_2}
\mathbb{P}\left[ \tilde{X}^{i(t^*, \kappa)}_{n+1} \in \tilde{B} \; \big| \; \tilde{X}^{i(t^*, \kappa)}_{n}=\tilde{x}, \Gamma
^{i(t^*, \kappa)}_n = \gamma \right] & = & \tilde{P}_{\gamma}(\tilde{x},
\tilde{B})  
\end{eqnarray}
for $n+1 >\tau_{i(t^*,\kappa)}$ and $\tilde{B}\in\mathcal{B}(\tilde{\mathcal{X}})$.

Define also the filtration $\{\mathcal{G}_n^{*}\}_{n=0}^{\infty}$ as
an extension of $\{\mathcal{G}_n \}_{n=0}^{\infty}$ by:
\begin{equation} \label{ext_filtr}
\mathcal{G}_n^{*} :=
\sigma \left\{
\{(\tilde{X}_k,
\Gamma_k)\}_{k=0}^{n},\; \{(\tilde{X}^{m(t^*)}_k,
\Gamma ^{m(t^*)}_k)\}_{k=0}^{n}, \; \{(\tilde{X}^{i(t^*,\kappa)}_k,
\Gamma ^{i(t^*,\kappa)}_k)\}_{k=0}^{n}
\right\}.\end{equation}

Let the distributions \[ \tilde{A}_n^{m(t^*), (\tilde{x},\gamma)}(\cdot), \quad
\tilde{A}_n^{m(t^*), \mathcal{G}_t^*}(\cdot),\quad
A_n^{m(t^*), (\tilde{x},\gamma)}(\cdot), \quad A_n^{m(t^*),\mathcal{G}_t^*}(\cdot),\]  and \[ \tilde{A}_n^{i(t^*, \kappa), (\tilde{x},\gamma)}(\cdot),
\quad \tilde{A}_n^{i(t^*, \kappa), \mathcal{G}_t^*}(\cdot),
\quad A_n^{i(t^*, \kappa), (\tilde{x},\gamma)}(\cdot), \quad A_n^{i(t^*, \kappa), \mathcal{G}_t^*}(\cdot), \]  be
analogues of $\tilde{A}_n^{(\tilde{x},\gamma)}(\cdot)$,
$\tilde{A}_n^{\mathcal{G}_t}(\cdot)$,
$A_n^{(\tilde{x},\gamma)}(\cdot)$, $A_n^{\mathcal{G}_t}(\cdot)$,
where in the definitions of the above terms, instead of $\{(\tilde{X}_n,
\Gamma_n)\}_{n=0}^{\infty},$   we use the sequences $\{(\tilde{X}^{m(t^*)}_n,
\Gamma ^{m(t^*)}_n)\}_{n=0}^{\infty}$ and $\{(\tilde{X}^{i(t^*,\kappa)}_n,
\Gamma ^{i(t^*,\kappa)}_n)\}_{n=0}^{\infty}$, respectively, and
condition on the extended $\sigma$-algebras defined in \eqref{ext_filtr}. 
\begin{lemma}\label{lemma:extended_tv}
	Let $\tilde{\nu}_1$ and $\tilde{\nu}_2$ be probability measures on $\tilde{\mathcal{X}}=\mathcal{X}\times \Phi$ and let $\nu_1$ and $\nu_2$ be their marginals on $\mathcal{X}$. Then \begin{equation}
	\|\nu_1 - \nu_2\|_{TV} \leq \|\tilde{\nu}_1 - \tilde{\nu}_2\|_{TV}.
	\end{equation}
\end{lemma}
\begin{proof}
	Total variation distances on $\tilde{\mathcal{X}}$ involve suprema over larger classes of sets than those on $\mathcal{X}$, in particular $|\nu_1(B)-\nu_2(B)| = |\tilde{\nu}_1(B\times\Phi)-\tilde{\nu}_2(B\times \Phi)|$.
\end{proof}
\begin{lemma}\label{lemma:markovian_conv_unif} Let Simultaneous Uniform Ergodicity, i.e. condition (a) of Theorem \ref{thm:uniform}, hold. Then for all $\varepsilon >0,$ there exists $N_0=N_0(\varepsilon)$ such that for all $N\geq N_0$
	\begin{eqnarray} \label{eqn:markovian_conv_unif}
	\| A_{t^*+N}^{m(t^*),
		(\tilde{x},\gamma)}(\cdot) -
	\pi(\cdot)\|_{TV} 
	& \leq & \varepsilon  \qquad
	\textrm{for all} \quad t^* \in \mathbb{N},\; \tilde{x} \in \tilde{\mathcal{X}} \; \textrm{
		and } \; \gamma \in \mathcal{Y}.
	\end{eqnarray}
\end{lemma}
\begin{proof}
	First observe that for any $B \in \mathcal{B}(\mathcal{X})$ the object $A_{t^*+N}^{m(t^*), \mathcal{G}_{t^*}^*}(B)$ is a $\mathcal{G}_{t^*}^*$-measurable random variable and apply Jensen's inequality to obtain \eqref{Jensen1} below. Next, use Lemma \ref{lemma:extended_tv} in \eqref{tilde}. To get \eqref{tv_markov} recall that $\{\tilde{X}_{k}^{m(t^*)}\}_{k=t^*}^{\infty}$ is a Markov chain started from $\tilde{X}_{t^*}$ with dynamics $\tP_{\Gamma_{t^*}}$. Then use monotonicity of the total variation for Markov chains, and finally pick $N_0=N_0(\varepsilon)$ via assumption (a) of Theorem \ref{thm:uniform} to conclude \eqref{lemma:part_3}.
	\begin{eqnarray} \label{Jensen1}
	\| A_{t^*+N}^{m(t^*),
		(\tilde{x},\gamma)}(\cdot) -
	\pi(\cdot)\|_{TV} 
	& \leq & \mathbb{E} \left[ 
	\| A_{t^*+N}^{m(t^*), \mathcal{G}_{t^*}^*}(\cdot) -
	\pi(\cdot)\|_{TV} 
	\right]
	\\ \label{tilde}
	& \leq & \mathbb{E} \left[ 
	\| \tilde{A}_{t^*+N}^{m(t^*), \mathcal{G}_{t^*}^*}(\cdot) -
	\tilde{\pi}_{\Gamma_{t^*}}(\cdot)\|_{TV} \right]
	\\ \label{tv_markov}
	& = & 
	\mathbb{E} \left[ 
	\| \tilde{P}_{\Gamma_{t^*}}^{N}(\tilde{X}_{t^*}, \cdot) -
	\tilde{\pi}_{\Gamma_{t^*}}(\cdot)\|_{TV} \right]
	\\ \label{tv_monotone}
	& \leq & 
	\mathbb{E} \left[ 
	\| \tilde{P}_{\Gamma_{t^*}}^{N_0}(\tilde{X}_{t^*}, \cdot) -
	\tilde{\pi}_{\Gamma_{t^*}}(\cdot)\|_{TV} \right]
	\\ \label{lemma:part_3}
	& \leq & \mathbb{E}[\varepsilon] = \varepsilon.
	\end{eqnarray}
\end{proof}

\begin{lemma}\label{lemma:markovian_conv_cont} Let Containment, i.e. condition (a) of Theorem \ref{thm:non_uniform}, hold. Then for all  $\tilde{x} \in \tilde{\mathcal{X}}, \; \gamma \in \mathcal{Y}$ and $\varepsilon >0,$ there exists $N_0=N_0(\varepsilon, \tilde{x}, \gamma),$ such that for all $t^* \in \mathbb{N}$ and all $N\geq N_0,$
	\begin{eqnarray} \label{eqn:markovian_conv_cont}
	\| A_{t^*+N}^{m(t^*),
		(\tilde{x},\gamma)}(\cdot) -
	\pi(\cdot)\|_{TV} \; \leq \; \mathbb{E} \left[ 
	\| A_{t^*+N}^{m(t^*), \mathcal{G}_{t^*}^*}(\cdot) -
	\pi(\cdot)\|_{TV} \right]
	& \leq & \varepsilon.
	\end{eqnarray}
\end{lemma}
\begin{proof}
	Reiterate the argument in the proof of Lemma \ref{lemma:markovian_conv_unif} to get the first part of \eqref{eqn:markovian_conv_cont} and then to arrive on \eqref{tv_markov}. Next, use the Containment condition, namely  
	choose $N_0 = N_0(\varepsilon, \tilde{x}, \gamma)$ such that 
	\begin{equation}
	\mathbb{P}\left(M_{\varepsilon/2}(\tilde{X}_{n}, \Gamma_n) > N_0 | \tilde{X}_0 =\tilde{x}, \Gamma_0 = \gamma \right) \leq \varepsilon/2
	\end{equation}
	for all $n \in \mathbb{N}$.	
	Let $G := \{M_{\varepsilon/2}(\tilde{X}_{t^*}, \Gamma_{t^*}) \leq N_0\}$, then $\mathbb{P}(G') \leq \varepsilon/2$ and $\| \tilde{P}_{\Gamma_{t^*}}^{N_0}(\tilde{X}_{t^*}, \cdot) -
	\tilde{\pi}_{\Gamma_{t^*}}(\cdot)\|_{TV}  \leq \varepsilon/2$ on $G$. Therefore we get 
	\begin{eqnarray} \nonumber
	\| A_{t^*+N}^{m(t^*),
		(\tilde{x},\gamma)}(\cdot) -
	\pi(\cdot)\|_{TV} 
	& \leq & \mathbb{E} \left[ 
	\| \tilde{P}_{\Gamma_{t^*}}^{N}(\tilde{X}_{t^*}, \cdot) -
	\tilde{\pi}_{\Gamma_{t^*}}(\cdot)\|_{TV} \right] \\
	& \leq & \mathbb{E} \left[ 
	\| \tilde{P}_{\Gamma_{t^*}}^{N_0}(\tilde{X}_{t^*}, \cdot) -
	\tilde{\pi}_{\Gamma_{t^*}}(\cdot)\|_{TV} \right] \\
	& = & \mathbb{E} \left[ 
	\| \tilde{P}_{\Gamma_{t^*}}^{N_0}(\tilde{X}_{t^*}, \cdot) -
	\tilde{\pi}_{\Gamma_{t^*}}(\cdot)\|_{TV}  \mathbb{I}_{G}\right]   \\ 
	&& \qquad+ \mathbb{E} \left[ 
	\| \tilde{P}_{\Gamma_{t^*}}^{N_0}(\tilde{X}_{t^*}, \cdot) -
	\tilde{\pi}_{\Gamma_{t^*}}(\cdot)\|_{TV}  \mathbb{I}_{G'} \right] \nonumber 
	\\ 
	& \leq & (\varepsilon/2) \cdot \mathbb{P}(G)  + 1 \cdot \mathbb{P}(G') \; \leq \; \varepsilon, \label{lemma:part_3_containment}
	\end{eqnarray} as required.
\end{proof}

\begin{lemma}\label{lemma:adap_inter_tv}
	Let Diminishing Adaptation, i.e. condition (b) of Theorem \ref{thm:uniform}, hold. Then for all $\varepsilon>0,$ $\kappa>0$ and $N_0 \in \mathbb{N}$ there exists $t_0 = t_0(\varepsilon, \kappa, N_0)$ such that for every $t^* \geq t_0$ and every $N\leq N_0,$
	\begin{equation}\label{eq:adap_inter_tv}
	\|\tilde{A}_{t^*+N}^{(\tilde{x}, \gamma)}(\cdot ) -
	\tilde{A}_{t^*+N}^{i(t^*, \kappa), (\tilde{x}, \gamma)}(\cdot)\|_{TV} \leq \varepsilon.
	\end{equation}
\end{lemma}
\begin{proof}
	Recall $D_n$ defined in condition (b) of Theorem \ref{thm:uniform} and let \begin{equation}\label{def:H} H_n := \{D_n \geq \kappa \}.\end{equation} 
	Note that by Diminishing Adaptation for every $n \geq
	t_0 = t_0(\varepsilon, \kappa, N_0)$ we have $\mathbb{P}(H_n) \leq \varepsilon / N_0.$  Now, for $t^* \geq t_0,$ define  
	\begin{equation}\label{def:E} E :=
	\bigcap_{n=t^*}^{t^*+N_0-1}H_n^c,
	\qquad \textrm{satisfying} \qquad 
	\mathbb{P}(E) \geq 1- \varepsilon.\end{equation} 
	
	Consider the process $\{(\tilde{X}^{i(t^*,\kappa)}_n,
	\Gamma ^{i(t^*,\kappa)}_n)\}_{n=0}^{\infty}$  with $t^* \geq t_0$. Note that on $E$ we have
	$\tilde{X}_n = \tilde{X}^{i(t^*,\kappa)}_n$, for $n=0,1,\dots, t^*+N_0,$ and therefore the coupling inequality (see e.g. Section 4.1 of
	\cite{roberts2004general}) for every $N\leq N_0$ yields as claimed
	\begin{eqnarray}\nonumber
	\|\tilde{A}_{t^*+N}^{(\tilde{x}, \gamma)}(\cdot ) -
	\tilde{A}_{t^*+N}^{i(t^*, \kappa), (\tilde{x}, \gamma)}(\cdot)\|_{TV} & \leq & \mathbb{P}(E^c) \; \leq  \; \varepsilon.
	\end{eqnarray}
\end{proof}
\begin{lemma}\label{lemma:markov_inter_tv}
	For every $\kappa> 0,$ $t^* \in \mathbb{N}$ and $N \in \mathbb{N}$, the distributions of $\{(\tilde{X}^{i(t^*,\kappa)}_n,
	\Gamma ^{i(t^*,\kappa)}_n)\}_{n=0}^{\infty}$ and $\{(\tilde{X}^{m(t^*)}_n,
	\Gamma ^{m(t^*)}_n)\}_{n=0}^{\infty}$ satisfy the following: 
	\begin{align}
	\label{eqn:inter_markov_Jensen} \|
	\tilde{A}_{t^*+N}^{i(t^*, \kappa), (\tilde{x}, \gamma)}(\cdot) - \tilde{A}_{t^*+N}^{m(t^*), (\tilde{x},\gamma)}(\cdot)\|_{TV}
	&  \leq   \mathbb{E}\Big[ 
	\|
	\tilde{A}_{t^*+N}^{i(t^*, \kappa), \mathcal{G}_{t^*}^*}(\cdot) -
	\tilde{P}_{\Gamma_{t^*}}^{N}(\tilde{X}_{t^*}, \cdot) \|_{TV}\Big] \\ \label{eqn:inter_markov_tv}
	& \leq  \kappa N^2.
	\end{align}
\end{lemma}    

\begin{proof}
	First apply Jensen's inequality
	\begin{equation} \label{Jensen}
	\|
	\tilde{A}_{t^*+N}^{i(t^*, \kappa), (\tilde{x}, \gamma)}(\cdot) - \tilde{A}_{t^*+N}^{m(t^*), (\tilde{x},\gamma)}(\cdot)\|_{TV}
	\leq   \mathbb{E}\Big[ 
	\|
	\tilde{A}_{t^*+N}^{i(t^*, \kappa), \mathcal{G}_{t^*}^*}(\cdot) -
	\tilde{A}_{t^*+N}^{m(t^*), \mathcal{G}_{t^*}^*}(\cdot)  \|_{TV}\Big],
	\end{equation}
	and recall equations \eqref{Markov_seq_G} and \eqref{Markov_seq_X}
	to note that 
	\begin{equation} \label{m_distribuiton}
	\tilde{A}_{t^*+N}^{m(t^*), \mathcal{G}_{t^*}^*}(\cdot) =
	\tilde{P}_{\Gamma_{t^*}}^{N}(\tilde{X}_{t^*}^{m(t^*)}, \cdot) =
	\tilde{P}_{\Gamma_{t^*}}^{N}(\tilde{X}_{t^*}, \cdot),
	\end{equation} 
	that is $\{ \tilde{X}_n^{m(t^*)}\}_{n=t^*}^{t^*+N}$ is a Markov chain
	started from $ \tilde{X}_{t^*}$ with dynamics
	$\tilde{P}_{\Gamma_{t^*}}$. Combining \eqref{Jensen} with \eqref{m_distribuiton}  yields \eqref{eqn:inter_markov_Jensen}.
	
	Now recall \eqref{Inter_seq_G}, \eqref{inter_seq_X_1},
	\eqref{inter_seq_X_2}, i.e.  the dynamics of $\{(\tilde{X}^{i(t^*,\kappa)}_n,
	\Gamma ^{i(t^*,\kappa)}_n)\}_{n=K-N}^{K},$ and observe that 
	\eqref{Inter_seq_G} yields 
	\begin{equation} \label{kernels_similar}
	\sup_{\tilde{x} \in
		\tilde{\mathcal{X}}} \| \tilde{P}_{\Gamma_{t^*}}(\tilde{x},
	\cdot) - \tilde{P}_{\Gamma ^{i(t^*,\kappa)}_{n}}(\tilde{x},
	\cdot) \|_{TV} \leq N\kappa
	\quad \textrm{ for } \; n = t^*, \dots, t^*+N-1.
	\end{equation}
	Hence, for every $n = t^*, \dots, t^*+N-1,\;$ if $\tilde{X}^{m(t^*)}_n =
	\tilde{X}^{i(t^*,\kappa)}_n,$ then by \eqref{kernels_similar} and Proposition~3(g) of
	\cite{roberts2004general}, there exists a coupling of
	$\tilde{X}^{m(t^*)}_{n+1}$ and 
	$\tilde{X}^{i(t^*,\kappa)}_{n+1},$ such that 
	$$\mathbb{P}\big[ \tilde{X}^{m(t^*)}_{n+1} =
	\tilde{X}^{i(t^*,\kappa)}_{n+1}\big] \geq 1- N\kappa.
	$$ Reiterating this construction $N$ times from $n=t^*$ to $n=t^*+N-1$
	implies that there exists a coupling such that 
	\begin{equation} 
	\mathbb{P}\big[ \tilde{X}^{m(t^*)}_{t^*+N} =
	\tilde{X}^{i(t^*,\kappa)}_{t^*+N}\; \big| \; \tilde{X}^{m(t^*)}_{t^*} =
	\tilde{X}^{i(t^*,\kappa)}_{t^*} \big] \geq 1-
	N^2 \kappa . \end{equation}
	Hence by the coupling inequality 
	\begin{equation*}
	\|
	\tilde{A}_{t^*+N}^{i(t^*, \kappa), \mathcal{G}_{t^*}^*}(\cdot) -
	\tilde{P}_{\Gamma_{t^*}}^{N}(\tilde{X}_{t^*}, \cdot)  \|_{TV}
	\leq \kappa N^2,
	\end{equation*}
	as required.
\end{proof}
\noindent\subsubsection*{\bf Proof of Theorem \ref{thm:uniform}}
\begin{proof} 
	By the triangle inequality, for any $n, t^*,$ and $\kappa,$  we have
	\begin{eqnarray} \nonumber 
	T_n(\tilde{x}, \gamma) =\|A_{n}^{(\tilde{x}, \gamma)}(\cdot ) -
	\pi(\cdot)\|_{TV} & \leq & \|A_{n}^{(\tilde{x}, \gamma)}(\cdot ) -
	A_n^{i(t^*, \kappa), (\tilde{x}, \gamma)}(\cdot)\|_{TV} \\ \nonumber && \qquad +  \|
	A_n^{i(t^*, \kappa), (\tilde{x}, \gamma)}(\cdot) - A_n^{m(t^*),(\tilde{x},\gamma)}(\cdot)\|_{TV} \\
	\nonumber  &&
	\qquad
	\qquad
	+ 
	\| A_n^{m(t^*),
		(\tilde{x},\gamma)}(\cdot) -
	\pi(\cdot)\|_{TV} \\ \nonumber 
	& \leq & \|\tilde{A}_{n}^{(\tilde{x}, \gamma)}(\cdot ) -
	\tilde{A}_n^{i(t^*, \kappa), (\tilde{x}, \gamma)}(\cdot)\|_{TV} \\ \nonumber && \qquad +  \|
	\tilde{A}_n^{i(t^*, \kappa), (\tilde{x}, \gamma)}(\cdot) - \tilde{A}_n^{m(t^*),  (\tilde{x},\gamma)}(\cdot)\|_{TV} \\
	\nonumber  &&
	\qquad
	\qquad
	+ 
	\| A_n^{m(t^*),
		(\tilde{x},\gamma)}(\cdot) -
	\pi(\cdot)\|_{TV} \\ \label{decomposition}
	&  =: & \diamondsuit_n^{(1)} + \diamondsuit_n^{(2)} + \diamondsuit_n^{(3)},
	\end{eqnarray}
	where in the second inequality, for the first two terms,  we have used Lemma \ref{lemma:extended_tv}.
	
	Now fix $\delta >0$. To prove the claim, it is enough to construct a target time $K_0 =
	K_0(\delta, \tilde{x}, \gamma)$, s.t.
	\begin{equation}\label{T_to_0}
	T_K(\tilde{x}, \gamma) \leq \delta \qquad \textrm{for all} \quad 
	K > K_0.
	\end{equation}  
	We shall find such target time of the form $K_0=t_0+N_0$. To this end let $\varepsilon = \delta /3$. \\
	First, use Lemma \ref{lemma:markovian_conv_unif} to fix $N_0=N_0(\varepsilon)$ so that 
	\begin{eqnarray}\label{bund_diamond_3}
	\diamondsuit_{t^*+N_0}^{(3)} & \leq & \varepsilon \quad \textrm{for all} \quad t^*\in \mathbb{N}, \tilde{x} \in\tilde{\mathcal{X}} \; \textrm{ and } \; \gamma \in \mathcal{Y}.
	\end{eqnarray}
	Next, take $\kappa:= \varepsilon /N_0^2$ and use Lemma \ref{lemma:markov_inter_tv} to conclude that 
	\begin{eqnarray}\label{bund_diamond_2}
	\diamondsuit_{t^*+N_0}^{(2)} & \leq & \varepsilon \quad \textrm{for all} \quad t^*\in \mathbb{N}, \tilde{x} \in\tilde{\mathcal{X}} \; \textrm{ and } \; \gamma \in \mathcal{Y}.
	\end{eqnarray}
	Finally, use Lemma \ref{lemma:adap_inter_tv} to find $t_0 = t_0(\varepsilon, \kappa, N_0)$ such that     
	\begin{eqnarray}\label{bund_diamond_1}
	\diamondsuit_{t^*+N_0}^{(1)} & \leq & \varepsilon \quad \textrm{for all} \quad t^* \geq t_0, \tilde{x} \in\tilde{\mathcal{X}} \; \textrm{ and } \; \gamma \in \mathcal{Y}.
	\end{eqnarray}
	Letting $K_0:= t_0 + N_0$ allows to decompose every $K\geq K_0$ into $K=t^* + N_0$, so that \eqref{bund_diamond_3}, \eqref{bund_diamond_2}, \eqref{bund_diamond_1} are satisfied, which yields the claim.      
\end{proof}

\noindent\subsubsection*{\bf Proof of Theorem \ref{thm:non_uniform}}
\begin{proof} The proof is identical, except that we use Lemma \ref{lemma:markovian_conv_cont} instead of Lemma \ref{lemma:markovian_conv_unif} to find $N_0 = N_0(\varepsilon, \tilde{x}, \gamma)$ in \eqref{bund_diamond_3}.
\end{proof}

\noindent\subsubsection*{\bf Proof of Theorem \ref{thm:lln}}
\begin{proof} To prove that the Weak Law of Large Numbers holds for the Auxiliary Variable Adaptive MCMC class, recall again the sequences  $\{(\tilde{X}^{m(t^*)}_n,
	\Gamma ^{m(t^*)}_n)\}_{n=0}^{\infty}$ and $\{(\tilde{X}^{i(t^*,\kappa)}_n,
	\Gamma ^{i(t^*,\kappa)}_n)\}_{n=0}^{\infty}$ defined above. Without loss of generality we will assume that $\pi(g) = 0$ and that $|g(x)| < a$. 
	
	By Markov's inequality 
	\begin{eqnarray}\label{eq:lln_proof_Markov}
	\bb{P} \left(\frac{1}{T} \Big| \sum_{i=1}^{T} g\left(X_i\right) \Big| \geq \delta^{1/2}  \right) & \leq & \frac{1}{\delta^{1/2}} \E\left(\frac{1}{T} \Big| \sum_{i=1}^{T} g\left(X_i\right) \Big| \right),
	\end{eqnarray}
	hence to obtain the WLLN it is enough to show that for every $\delta >0$ there exists such $T_0 = T_0(\delta) = T_0(\delta, \tilde{x}, \gamma),$ where $(\tilde{x}, \gamma)$ are the starting points of $(\tilde{X}_n, \Gamma_n),$ that for all $T>T_0$ 
	\begin{eqnarray}\label{eq:lln_goal}
	\E\left(\frac{1}{T} \Big| \sum_{i=1}^{T} g\left(X_i\right) \Big| \right) & \leq & \delta.
	\end{eqnarray}
	
	We shall deal with \eqref{eq:lln_goal} by considering second moments and therefore will have to deal with mixed terms of the form $\E g(X_i)g(X_j)$. Let $\varepsilon >0$ be fixed and we shall pick a specific value later. Firstly, for $i < j,$ consider the following calculation.
	\begin{eqnarray} \nonumber
	\left|\E g(X_i)g(X_j^{m(i)}) \right|
	& = & 
	\left|\E \left(\E \left[ g(X_i)g(X_j^{m(i)}) \big| \mathcal{G}_i^* \right] \right) \right| 
	\\ \nonumber
	& \leq &
	\left|\E \left(g(X_i) \E \left[ g(X_j^{m(i)}) \big| \mathcal{G}_i^* \right] \right) \right| 
	\\ \nonumber
	& \leq &
	\E \left( |g(X_i)| \left| \E \left[ g(X_j^{m(i)}) -\pi(g) \big| \mathcal{G}_i^* \right] \right| \right) 
	\\ \nonumber
	& \leq &
	a^2 \E \left[ 
	\| A_{j}^{m(i), \mathcal{G}_{i}^*}(\cdot) -
	\pi(\cdot)\|_{TV} \right]
	\\ \label{eq:mixed_markov}
	& \leq &
	a^2 \varepsilon, \qquad \textrm{for all } \; i\in \mathbb{N} \; \textrm{ and } \; j-i \geq N_0=N_0(\varepsilon, \tilde{x}, \gamma),
	\end{eqnarray}
	where $N_0(\varepsilon, \tilde{x}, \gamma)$ has been obtained from Lemma \ref{lemma:markovian_conv_unif}, if assuming Simultaneous Uniform Ergodicity, or from Lemma \ref{lemma:markovian_conv_cont}, if assuming Containment. 
	
	Secondly, given $N_0$ in \eqref{eq:mixed_markov}, fix $N_1 \geq N_0$ such that $1/N_1 < \varepsilon$ and consider pairs $i,j$ satisfying $N_1 \leq j-i \leq N_1^2.$ Set $\kappa:= \varepsilon / N_1^4,$ and compute 
	\begin{eqnarray} \nonumber
	\left|\E g(X_i)\big(g(X_j^{i(i,\kappa)}) -g(X_j^{m(i)})\big) \right|
	& \leq & \E |g(X_i)|\big|g(X_j^{i(i,\kappa)})-g(X_j^{m(i)})\big| 
	\\ \nonumber
	& \leq &
	a^2 \| A_{j}^{i(i, \kappa), (\tilde{x},\gamma)}(\cdot) -
	A_{j}^{m(i), (\tilde{x},\gamma)}(\cdot)\|_{TV} 
	\\ \nonumber
	& \leq &
	a^2 \kappa (j-i)^2, \quad \textrm{for all } \; \kappa, i,j, \; \textrm{by Lemmas \ref{lemma:extended_tv} and  \ref{lemma:markov_inter_tv}, } 
	\\ 
	& \leq & \label{eq:mixed_inter_markov}
	a^2 \varepsilon
	\quad \textrm{since } \;N_1 \leq j-i \leq N_1^2 \; \textrm{ and } \kappa= \varepsilon / N_1^4.
	\end{eqnarray}
	
	Finally, use Lemma \ref{lemma:adap_inter_tv} to find $t_0=t_0(\varepsilon, \kappa, N_1^2)$ and conclude
	\begin{eqnarray} \nonumber
	\left|\E g(X_i)\big(g(X_j)-g(X_j^{i(i,\kappa)})\big) \right|
	& \leq & \E |g(X_i)|\big|g(X_j)-g(X_j^{i(i,\kappa)})\big| 
	\\ \nonumber
	& \leq &
	a^2 \| A_{j}^{(\tilde{x}, \gamma)}(\cdot) - A_{j}^{i(i, \kappa), (\tilde{x}, \gamma)}(\cdot)
	\|_{TV} 
	\\ \label{eq:mixed_inter_adap}
	& \leq & 
	a^2 \varepsilon, \quad \textrm{for all } \; i> t_0 = t_0(\varepsilon, \kappa, N_1^2) \textrm{ and } j-i \leq N_1^2. 
	\end{eqnarray}
	
	We are ready to address the mixed term. Since $|g|<a$, trivially for any $i,j$
	\begin{eqnarray} \label{eq:mixed_trivial}
	|\E g(X_i)g(X_j)| & \leq & a^2.
	\end{eqnarray}
	Moreover, for $\varepsilon >0,$ $N_1 \geq N_0(\varepsilon, \tilde{x}, \gamma),$ $\kappa = \varepsilon/N_1^4$, pairs $i,j$ such that $N_1 \leq j-i \leq N_1^2$, and $i> t_0(\varepsilon, \kappa, N_1^2)$ equations \eqref{eq:mixed_markov}, \eqref{eq:mixed_inter_markov} and \eqref{eq:mixed_inter_adap} yield
	\begin{eqnarray} \nonumber
	|\E g(X_i)g(X_j)| &\leq & \left|\E g(X_i)\big(g(X_j)-g(X_j^{i(i,\kappa)})\big) \right| \\ \nonumber
	&&\quad + \left|\E g(X_i)\big(g(X_j^{i(i,\kappa)}) -g(X_j^{m(i)})\big) \right| \\ \nonumber
	& &\quad \quad+ \left|\E g(X_i)g(X_j^{m(i)}) \right| 
	\\ \label{eq:mixed_combined}
	& \leq & 3a^2 \varepsilon.
	\end{eqnarray}
	Consequently, for any $N_1 > \max\{1/\varepsilon, N_0\}$ and  $t_1 > t_0$, chosen as above, we can compute
	\begin{eqnarray} \nonumber
	\E \left(\frac{1}{N_1^2} \sum_{k=t_1+1}^{t_1+N_1^2} g(X_k)\right)^2 
	& \leq & 
	\frac{1}{N_1^4}\Bigg(
	\sum_{i,j = t_1+1}^{t_1 + N_1^2}\mathbb{I}_{\{|j-i|\geq N_1\}}|\E g(X_i)g(X_j)| \\ \nonumber
	&& \quad+  \sum_{i,j = t_1+1}^{t_1 + N_1^2}\mathbb{I}_{\{|j-i| < N_1\}}|\E g(X_i)g(X_j)| 
	\Bigg) \\
	& \leq &
	\frac{2}{N_1^4}\left(
	N_1^43a^2\varepsilon + N_1^3a^2
	\right)
	\; \leq \; 8a^2 \varepsilon, 
	\label{eq:second_moment}
	\end{eqnarray}
	where we have used \eqref{eq:mixed_combined} and \eqref{eq:mixed_trivial} to bound the first and second summation, respectively.
	
	By the Cauchy-Schwartz inequality \eqref{eq:second_moment} implies
	\begin{eqnarray} 
	\E \left|\frac{1}{N_1^2} \sum_{k=t_1+1}^{t_1+N_1^2} g(X_k)\right| 
	& \leq & 
	2 \sqrt{2} a \varepsilon^{1/2} 
	\label{eq:one_block}
	\end{eqnarray}
	for $N_1 > \max\{1/\varepsilon, N_0(\varepsilon, \tilde{x}, \gamma)\}$ and   $t_1 > t_0(\varepsilon, \kappa, N_1^2)$.
	
	Following the proof of Theorem 5 of \cite{roberts2007coupling}, fix  $T$ so large that 
	\begin{equation}\label{eq:choice_of_T}
	\max \left[\frac{at_0}{T}, \frac{aN_1^2}{T}  \right] \; \leq \; \varepsilon^{1/2}.
	\end{equation}
	Use \eqref{eq:one_block} and \eqref{eq:choice_of_T} to observe that
	\begin{eqnarray} \nonumber 
	\E\left(\frac{1}{T} \Big| \sum_{i=1}^{T} g\left(X_i\right) \Big| \right)  & \leq & \E\left(\frac{1}{T} \Big| \sum_{i=1}^{t_0} g\left(X_i\right) \Big| \right)  \\
	&& \label{eq:lln_proof_part_4}
	\qquad
	+\E\left(\frac{N_1^2}{T} \sum_{j=1}^{\lfloor\frac{T-t_0}{N_1^2} \rfloor} \frac{1}{N_1^2} \Big|\sum_{k=1}^{N_1^2} g\left(X_{t_0+(j-1)N_1^2+k}\right) \Big| \right) 
	\\  \nonumber
	&& \qquad \qquad
	+  \E\left(\frac{1}{T} \Big| \sum_{i=t_0 + \lfloor\frac{T-t_0}{N_1^2} \rfloor N_1^2 + 1}^{T} g\left(X_i\right) \Big| \right) \\
	&\nonumber \leq& \frac{at_0}{T} + 2\sqrt{2} a\varepsilon^{1/2} + \frac{aN_1^2}{T} \\
	& \leq & 2(\sqrt{2}a+1) \varepsilon^{1/2}.\label{eq:sum_final}
	\end{eqnarray}
	Setting $\varepsilon:= (\delta/2(\sqrt{2}a+1))^2$ in the above argument yields \eqref{eq:lln_goal} as desired. \end{proof}

\subsubsection*{\bf Proof of Lemma \ref{thm:drift_condition}}
\noindent \begin{proof}
	We will begin the proof by showing that assumption \eqref{eq:main_drift_condition} implies that an analogous drift condition is satisfied for $\tPg^{n_0}$, $n_0$ defined in \eqref{eq:minorisation_condititon}, perhaps with different constants $\lambda$ and $b$, which we define below. For any $k \in \{1, \ldots, n_0\}$ we have
	$$
	\tPg^{k}(\tilde{x}): = \E \left(V_{\tpig}(\tilde{X}_{n+k}) \big| \tilde{X}_n = \tilde{x}, \Gamma_n = \gamma, \Gamma_{n+1} = \gamma, \ldots, \Gamma_{n+k-1} = \gamma\right).
	$$
	For $k=2$ we have
	\begin{align*}
	\tPg^{2}(\tilde{x}) = &\E \left(V_{\tpig}(\tilde{X}_{n+2}) \big| \tilde{X}_n = \tilde{x}, \Gamma_n = \gamma, \Gamma_{n+1} = \gamma\right)\\ = 
	&\E \left(\E \left(V_{\tpig}(\tilde{X}_{n+2}) \big| \tilde{X}_{n+1}, \tilde{X}_n = \tilde{x}, \Gamma_n = \gamma, \Gamma_{n+1} = \gamma\right) \big| \tilde{X}_n = \tilde{x}, \Gamma_n = \gamma, \Gamma_{n+1} = \gamma\right) \\
	\leq & \E \left( \lambda V_{\tpig}(\tilde{X}_{n+1}) + b \big| \tilde{X}_n = \tilde{x}, \Gamma_n = \gamma, \Gamma_{n+1} = \gamma\right) \\
	= & b + \lambda \E \left( V_{\tpig}(\tilde{X}_{n+1})\big| \tilde{X}_n = \tilde{x}, \Gamma_n = \gamma, \Gamma_{n+1} = \gamma\right) \\
	\leq &  b+ \lambda\left(\lambda V_{\tpig}(\tilde{x})+b\right) \leq \lambda^2 V_{\tpig}(\tilde{x}) +2b.
	\end{align*}
	By similar calculations and induction we obtain
	\begin{equation}\label{eq:additional_drift_condition}
	\tPg^{n_0}(\tilde{x}) \leq \lambda^{n_0} V_{\tpig}(\tilde{x}) +n_0 b,
	\end{equation}
	as required.
	
	By Theorem 12 of \cite{rosenthal1995minorization} and conditions \eqref{eq:minorisation_condititon} and \eqref{eq:additional_drift_condition},  there exists $K < \infty$ and $\rho<1$, depending only on $\lambda$, $b$, $v$, $n_0$ and $\delta$, such that for each $\gamma \in \mathcal{Y}$ and for any $k\in \mathbb{N}$ we have 
	
	\begin{equation}\label{eq:proposition_3_drift_condition}
	\| \tilde{P}_{\gamma}^{n_0 \cdot k}(\tilde{x}, \cdot)-\tilde{\pi}_{\gamma}(\cdot)  \|_{TV} \leq \left(K + V_{\tpig}(\tilde{x})\right)\rho^k.
	\end{equation}
	We now use the  monotonicity of 
	$\| \tilde{P}_{\gamma}^{n}(\tilde{x}, \cdot)-\tilde{\pi}_{\gamma}(\cdot)  \|_{TV}$ in $n$  (see Proposition 3b) of \cite{roberts2004general})
	to argue that
	\begin{align*}
	\| \tilde{P}_{\gamma}^{n_0 \cdot k +n_0-1}(\tilde{x}, \cdot)-\tilde{\pi}_{\gamma}(\cdot)  \|_{TV} \leq \ldots \leq & \| \tilde{P}_{\gamma}^{n_0 \cdot k +1}(\tilde{x}, \cdot)-\tilde{\pi}_{\gamma}(\cdot)  \|_{TV} \\
	\leq & \| \tilde{P}_{\gamma}^{n_0 \cdot k}(\tilde{x}, \cdot)-\tilde{\pi}_{\gamma}(\cdot)  \|_{TV} \\
	\leq & \left(K + V_{\tpig}(\tilde{x}) \right) \rho^k \\
	= & \left(K/\rho + V_{\tpig}(\tilde{x})/\rho \right) \rho^{k+1}.
	\end{align*}
	Let $\tilde{\rho}:=\rho^{\frac{1}{n_0}}$. It follows that for every $m \in \mathbb{N}$
	\begin{equation}\label{eq:transformed_drift_condition_bound}
	\| \tilde{P}_{\gamma}^{m}(\tilde{x}, \cdot)-\tilde{\pi}_{\gamma}(\cdot)  \|_{TV} \leq \left(K/\rho + V_{\tpig}(\tilde{x})/\rho \right) \tilde{\rho}^{m}.
	\end{equation}
	
	The next step of the proof will be to show that the sequence $V_{\tpi_{\Gamma_n}}(\tilde{X}_n)$ is bounded in probability. By Lemma 3 in \cite{roberts2007coupling}, it suffices to show that $\sup_{n} \E V_{\tpi_{\Gamma_n}}(\tilde{X}_n) < \infty$. Firstly, let us show that $\tP V_{\tpig}(\tilde{x})$ is bounded for $\gamma \in \mathcal{Y}$ and $\tilde{x} \in A$. Note that 
	\begin{align*}
	\sup_{\gamma \in \mathcal{Y}} \sup_{\tilde{x} \in A} \tPg V_{\tpig}(\tilde{x}) = & \sup_{\gamma \in \mathcal{Y}} \sup_{\tilde{x} \in A} \left(\frac{\tPg V_{\tpig}(\tilde{x})}{ V_{\tpig}(\tilde{x})} V_{\tpig}(\tilde{x})  \right) \leq \sup_{\gamma \in \mathcal{Y}} \sup_{\tilde{x} \in \mathcal{X}} \frac{\tPg V_{\tpig}(\tilde{x})}{ V_{\tpig}(\tilde{x})} \sup_{\gamma \in \mathcal{Y}} \sup_{\tilde{x} \in A}  V_{\tpig}(\tilde{x}).
	\end{align*}
	Since $A$  and $\mathcal{Y}$ were assumed to be  compact, $\sup_{\gamma \in \mathcal{Y}} \sup_{\tilde{x} \in A} V_{\tpig}(\tilde{x}) < \infty$. Additionally, the drift condition  \eqref{eq:main_drift_condition} yields
	$$
	\sup_{\gamma \in \mathcal{Y}} \sup_{\tilde{x} \in \mathcal{X}} \frac{\tPg V_{\tpig}(\tilde{x})}{V_{\tpig}(\tilde{x})} \leq \sup_{\gamma \in \mathcal{Y}} \sup_{\tilde{x} \in \mathcal{X}} \frac{ \lambda V_{\tpig}(\tilde{x}) + b}{V_{\tpig}(\tilde{x})} \leq \lambda + b.
	$$
	Therefore we can define $M:= \sup_{\gamma \in \mathcal{Y}} \sup_{\tilde{x} \in A} \tPg V_{\tpig}(\tilde{x}) < \infty$. It follows that
	\begin{align}\label{eq:drift_estimation}
	\begin{split}
	E \left(V_{\tpi_{\Gamma_{n+1}}}(\tilde{X}_{n+1})\big| \tilde{X}_n, \Gamma_n\right) = & \E \left(V_{\tpi_{\Gamma_{n+1}}}(\tilde{X}_{n+1})\big| \tilde{X}_n, \Gamma_n\right) \mathbb{I}_{\tilde{X}_n \in A}  \\
	& \quad +\E \left(V_{\tpi_{\Gamma_{n+1}}}(\tilde{X}_{n+1})\big| \tilde{X}_n, \Gamma_n\right) \mathbb{I}_{\tilde{X}_n \notin A} \\
	= & \E \left(V_{\tpi_{\Gamma_{n+1}}}(\tilde{X}_{n+1})\big| \tilde{X}_n, \Gamma_n\right) \mathbb{I}_{\tilde{X}_n \in A}  \\
	& \quad + \E \left(V_{\tpi_{\Gamma_{n}}}(\tilde{X}_{n+1})\big| \tilde{X}_n, \Gamma_n\right) \mathbb{I}_{\tilde{X}_n \notin A}\\
	\leq & \sup_{\tilde{x} \in A} \tP_{\Gamma_{n+1}} V_{\tpi_{\Gamma_{n+1}}}(\tilde{x}) + \E \left(V_{\tpi_{\Gamma_{n}}}(\tilde{X}_{n+1})\big| \tilde{X}_n, \Gamma_n\right) \\
	\leq & \sup_{\gamma \in \mathcal{Y}}\sup_{\tilde{x} \in A} \tPg V_{\tpig}(\tilde{x}) + \lambda V_{\tpi_{\Gamma_{n}}}(\tilde{X}_n) + b \\
	\leq & M+ \lambda V_{\tpi_{\Gamma_{n}}}(\tilde{X}_n) + b.
	\end{split}
	\end{align}
	By the law of total expectation,
	\begin{align*}
	\E V_{\tpi_{\Gamma_{n+1}}}(\tilde{X}_{n+1}) = & \E \E \left(V_{\tpi_{\Gamma_{n+1}}}(\tilde{X}_{n+1})\big| \tilde{X}_n, \Gamma_n\right), 
	\end{align*}
	which combined with \eqref{eq:drift_estimation} gives 
	\begin{align*}
	\E V_{\tpi_{\Gamma_{n+1}}}(\tilde{X}_{n+1}) \leq \lambda\E V_{\tpi_{\Gamma_{n}}}(\tilde{X}_{n}) + M + b.
	\end{align*}
	This implies, using Lemma 2 in \cite{roberts2007coupling}
	that
	\begin{equation}\label{eq:like_lemma_2}
	\sup_{n} \E V_{\tpi_{\Gamma_n}}(\tilde{X}_n) \leq \max \left[\E V_{\tpi_{\Gamma_0}}(\tilde{X}_0), \frac{M+b}{1-\lambda} \right].
	\end{equation}
	Lemma \ref{thm:drift_condition} will now follow from combining the fact that the sequence $V_{\tpi_{\Gamma_n}}(\tilde{X}_n)$ is bounded in probability with \eqref{eq:transformed_drift_condition_bound}. Note that for any fixed $\varepsilon$ and $\tilde{\delta}$, there exists $N$ such that 
	
	\begin{align*}
	\mathbb{P}\left(M_{\varepsilon}(\tilde{X}_{n}, \Gamma_n) \leq N  \right) & = \mathbb{P}\left(\| \tilde{P}_{\Gamma_n}^N(\tilde{X}_n, \cdot)-\tilde{\pi}_{\gamma}(\cdot) \|_{TV} \leq \varepsilon \right) \\
	&\geq \mathbb{P}\left( \left(K/\rho + V_{\tpi_{\Gamma_n}}(\tilde{X}_n)/\rho \right) \tilde{\rho}^{N} \leq \varepsilon  \right) \\
	& = \mathbb{P}\left( V_{\tpi_{\Gamma_n}}(\tilde{X}_n) \leq  \varepsilon \tilde{\rho}^{-N}\rho  - K \right) \\
	& \geq 1- \tilde{\delta}
	\end{align*}
	for all $n\in \mathbb{N}$. The last inequality holds since $ \varepsilon \tilde{\rho}^{-N}\rho  - K  \to  \infty$  as $N \to \infty$ and  $V_{\Gamma_n}(\tilde{X}_n)$ is bounded in probability.
\end{proof}

\section{Proofs for Section 4}\label{section:proofs_4}

\subsection*{\bf Proof of Theorem \ref{thm:heavy_tailed_proposal}}
\begin{proof}
	The aim of the proof is to verify the assumptions of Theorem \ref{thm:uniform} and conclude. Diminishing Adaptation has been addressed in Section \ref{subsection:airmcmc}, so it is enough to prove that Simultaneous Uniform Ergodicity holds. Note that assumption \eqref{eq:heavy_tailed_proposal} implies that for some positive constant~$c_1$
	\begin{equation}\label{eq:heavy_tailed_assumption_transformed}
	\frac{\RJk(y)}{\tpig(y,k)} =\frac{\RJk(y)}{\pi(y)} \frac{\sum_{j\in \mathcal{I}} \wgj Q_j(\mu_j, \Sigma_{\gamma,j})(y)}{\wgk Q_k(\mu_k, \Sigma_{\gamma,k})(y)}  >   \frac{\RJk(y)}{\pi(y)} > c_1 \quad 
	\end{equation}
	for each $k\in \mathcal{I}$, $y \in \mathcal{X}$ and $\giY$.
	For any $(x,i) \in \mathcal{X} \times \mathcal{I}$, any set $\hat{C} \subset \mathcal{X}$  and any $k \in \mathcal{I}$ we can compute
	\begin{align*}
	\int_{\hat{C} } \tPg \left((x,i), (dy,k) \right)   \geq &  \epsilon \agik  \int_{\hat{C} } \tpigJk \left((x,i), (dy,k) \right)   \\
	=& \epsilon \agik  \int_{\hat{C}} \RJk(y) \min \left[1, \frac{\tpig(y,k) \agki \RJi(x) }{\tpig(x,i)\agik\RJk(y)}\right] dy  \\
	\geq & \epsilon   \int_{\hat{C} } \min \left[\agik \RJk(y) , \tpig(y,k) \agki \frac{\RJi(x)}{\tpig(x,i)}\right] dy \\
	\geq & \epsilon \epsilon_a  \int_{\hat{C} }\min \left[c_1 \tpig(y,k), c_1 \tpig(y,k) \right] dy \\
	= & \epsilon \epsilon_a c_1 \int_{\hat{C} } \tpig(y,k) dy,
	\end{align*}
	where $\epsilon_a$ is as in equation \eqref{eq:weights_epsilon}. \\
	Furthermore, any set $C \subset \mathcal{X}\times\mathcal{I}$ may be decomposed as $C = \bigcup_{k \in \mathcal{I}} \hat{C}_k \times \{k\}$, therefore
	\begin{equation}\label{eq:space_small_set}
	\sum_{k \in \mathcal{I}} \int_{\hat{C}_k } \tPg \left((x,i), (dy,k) \right) \geq \sum_{k \in \mathcal{I}} \epsilon \epsilon_a c_1 \int_{\hat{C}_k } \tpig(y,k) dy = \epsilon \epsilon_a c_1 \tpig(C).
	\end{equation}
	Since $\tpig$ is a probability measure on $\mathcal{X} \times {I}$  for each $\giY$ and  \eqref{eq:space_small_set}
	holds for all $(x,i) \in \mathcal{X} \times {I}$, by Theorem 8 of \cite{roberts2004general} we have
	$$
	\|\tPg^n((x,i), \cdot) -
	\tilde{\pi}_{\gamma}(\cdot)\|_{TV} \leq \left(1-\epsilon \epsilon_a c_1\right)^n  \quad \text{for all } (x,i) \in \mathcal{X} \times \mathcal{I} \text{ and } \giY,
	$$
	which completes the proof.
\end{proof}

\subsubsection*{\bf Proof of Theorem \ref{thm:light_tailed_proposal}}

We will show that the assumptions of Theorem \ref{thm:non_uniform} are satisfied. Since Diminishing Adaptation was discussed in Section \ref{subsection:airmcmc}, it suffices to prove that the Containment condition holds, which we will do using Lemma \ref{thm:drift_condition}. Assumptions a) and d) were discussed in Section \ref{section:overview_assumptions}. Assumption e) follows directly from the construction of the algorithm for $A:= \bigcup_{i \in \mathcal{I} } A_i \times \{i\}$. Assumption f) holds trivially, since $\tilde{X}_0$ and $\Gamma_0$ are deterministic (chosen by the user of the algorithm). The remaining part of the proof is organised as follows.

We show that the drift condition expressed in assumption b) of Lemma~\ref{thm:drift_condition} is satisfied  under the assumptions  of Theorem \ref{thm:light_tailed_proposal}. To this end, we consider a drift function of the form 
\begin{equation} \label{equation_drift}
V_{\tpig}(\tilde{x}) := c\tpig(\tilde{x})^{-s} = c\tpig((x,i))^{-s}\end{equation} for some $s\in (0,1)$ and $c$ such that $c\pi(x)^{-s}\geq 1$ (thus enforcing $V_{\tpig}(\tilde{x}) >1$). We first focus on obtaining the appropriate result for the local kernels and subsequently we combine it with the result for jumps. Finally, we prove that assumption c) of Lemma~\ref{thm:drift_condition} is satisfied for $n_0=3$.

\noindent \subsubsection*{\bf Assumption b) of Lemma \ref{thm:drift_condition} (local kernels)}
\begin{proof}
	The drift function $V_{\tpig}$ defined as above is a jointly continuous  function of $(x,\gamma)$ so it is bounded on compact sets in $\mathcal{X} \times \mathcal{Y}$ for each $i \in \mathcal{I}$, as required by assumption~b) of Lemma \ref{thm:drift_condition}. Therefore, it is also bounded on compact sets in  $\mathcal{\tilde{X}} \times \mathcal{Y}$. The proof will be continued for $s=\frac{1}{2}$ but analogous reasoning would be valid for any~$s\in(0,1)$.

	We will prove that there exists $\lambda_L<1$ such that for the local move kernels we have 
	\begin{equation}\label{eq:local_kernels_limsup}
	\limsup_{|x|\to \infty} \sup_{\gamma \in \mathcal{Y}} \frac{\tpigLi V_{\tpig}\left((x,i)\right)}{V_{\tpig}\left((x,i)\right)} \leq \lambda_{L}
	\end{equation}
	for all $i \in \mathcal{I}$. We will refer multiple times to the proof of Theorem 4.1 of \cite{jarner2000geometric}. Following the notation used there, let $C_{\pi(x)}(\delta)$ denote the radial $\delta$-zone around $C_{\pi(x)}$, where $C_{\pi(x)}$ is the contour manifold corresponding to $\pi(x)$.  Firstly, there exists $R_0$ such that for $|x|>R_0$ the contour manifold $C_{\pi(x)}$ is parametrised by $S^{d-1}$ and encloses the acceptance set for $\pi$ defined as $a(x):=\{y \in \mathcal{X}: \pi(y) > \pi(x) \}$ (we refer to Section 4 of \cite{jarner2000geometric} for the details of this argument). In our proof we will only consider $|x| >R_0$. Define also
	\begin{equation}\label{eq:lambda_L_i_def}
	\lambda_{L,i} :=\limsup_{|x| \to \infty} \sup_{\gamma \in \mathcal{Y}} \int_{r_{\gamma,i}(x)}\RLi(x,y) dy.
	\end{equation}
	By assumption \eqref{eq:assumption_limsup} $\lambda_{L,i}<1$.
	
	Fix $i \in \mathcal{I}$ and $\epsilon >0$. We will show that for sufficiently large $x$  
	\begin{equation}\label{eq:epsilon_bounding}
	\frac{\tpigLi V_{\tpig} \left((x,i)\right)}{V_{\tpig}\left((x,i)\right)} \leq \lambda_{i,L}  + 3\epsilon  + \epsilon^{1/2}.
	\end{equation}
	The idea of this proof is to split $\mathcal{X}$ into disjoint sets $\mathcal{X} \setminus B(x,K)$, $B(x,K) \cap C_{\pi(x)}(\delta)$ and $B(x,K) \setminus C_{\pi(x)}(\delta)$ and show that for any $x$ with a sufficiently large norm the integral representing acceptance, that is, of the function $\RLi(x,y) \min \left[1, \frac{\tpig(y,i)}{\tpig(x,i)}\right] \frac{V_{\tpig}\left((y,i)\right)}{V_{\tpig}\left((x,i)\right)}$ on those sets is bounded from above by $\epsilon$, $\epsilon$ and $\epsilon^{1/2}$, respectively. We fix the values of $K$ and $\delta$ below. As for the rejection part, we use \eqref{eq:lambda_L_i_def} to show that the corresponding integral is bounded by $\lambda_{L,i} + \epsilon$, for all $x$ at a sufficient distance from \textbf{0}. Putting all these upper bounds together, we obtain the required $\lambda_{i,L}  + 3\epsilon  + \epsilon^{1/2}$.

	Firstly, observe that by assumption a) and condition \eqref{eq:matrices_estimation} the family of distributions $\RLi(x,\cdot)$ is tight. Thus, there exists $K$ such that 
	\begin{equation}\label{eq:local_kernel_inequality_part_1}
	\sup_{\giY}  \int_{\mathcal{X}\setminus B(x,K)} \RLi(x,y) dy \leq \epsilon.
	\end{equation}
	Furthermore, as shown the proof of Theorem 4.1 of \cite{jarner2000geometric}, under assumption c) of Theorem~\ref{thm:light_tailed_proposal} for any positive $\delta$ and $K$ 
	\begin{equation}\label{eq:jarner_hansen_formula}
	\mu^{\text{Leb}}\left(B(x,K)  \cap C_{\pi(x)}(\delta) \right) \leq \delta \left(\frac{|x|+K}{|x|-K}\right)^{d-1} \frac{\mu^{\text{Leb}} \left(B(x, 3K)\right)}{K}.
	\end{equation}
	Fix $K$ satisfying \eqref{eq:local_kernel_inequality_part_1}. Since $\lim_{x\to \infty} \left(\frac{|x|+K}{|x|-K}\right)^{d-1} = 1$, there exists $R_1>0$ such that for  $|x| > \max[R_0, R_1]$  
	$$
	\left(\frac{|x|+K}{|x|-K}\right)^{d-1} < 1+\epsilon.
	$$
	
	Recall that by assumption a)  for any $x\in \mathcal{X}$ we have $\sup_{\gamma \in \mathcal{Y}}\sup_{y \in \mathcal{X}} \RLi(x,y)>0$. Now let us choose $\delta$ such that for $|x| > R_1$
	$$
	\mu^{\text{Leb}}\left(C_{\pi(x)}(\delta) \cap B(x,K)\right) \leq  \frac{ \epsilon}{ \sup_{\gamma \in \mathcal{Y}}\sup_{y \in \mathcal{X}} \RLi(x,y)}, 
	$$ 
	therefore getting
	\begin{equation}\label{eq:local_kernel_inequality_part_2}
	\sup_{\gamma \in \mathcal{Y}} \int_{C_{\pi(x)}(\delta) \cap B(x,K)}\RLi(x, y)dy \leq \epsilon.
	\end{equation}
	Let $r(x) = \{y \in \mathcal{X}: \pi(y)< \pi(x)\}$ and $a(x) = \{y \in \mathcal{X}: \pi(y)\geq \pi(x)\}$. 
	We now split $B(x,K)\setminus C_{\pi(x)}(\delta)$ into $\left(r(x) \cap B(x,K)\right)\setminus C_{\pi(x)}(\delta)$ and $\left(a(x) \cap B(x,K)\right)\setminus C_{\pi(x)}(\delta)$ and we estimate the value of $\min \left[1, \frac{\tpig(y,i)}{\tpig(x,i)}\right] \frac{V_{\tpig}\left((y,i)\right)}{V_{\tpig}\left((x,i)\right)}$ on each of those sets separately. Fix $\tilde{K}$ such that
	\begin{equation}\label{eq:t_dist_bound}
	\frac{\sum_{j\in \mathcal{I}} \wgj Q_j(\mu_j, \Sigma_{\gamma,j})(x)}{\wgi Q_i(\mu_i, \Sigma_{\gamma,i})(x)} \leq \tilde{K} \quad \text{for all } x \in \mathcal{X} \text{ and } \gamma \in \mathcal{Y}.
	\end{equation}
	This is possible by assumption d) combined with conditions \eqref{eq:matrices_estimation} and \eqref{eq:weights_epsilon}.
	Since $\pi$ is super-exponential, there exists $R_2$ so large that for   $|x| > \max[R_0, R_2]$: 
	\begin{enumerate}[1)]
		\item If $y \in \left(r(x) \cap B(x,K)\right)\setminus C_{\pi(x)}(\delta)$, then $\frac{\pi(y)}{\pi(x)} \leq \frac{\epsilon}{\tilde{K}}$.
		\item If $y \in \left(a(x) \cap B(x,K)\right)\setminus C_{\pi(x)}(\delta)$, then  $\frac{\pi(x)}{\pi(y)} \leq \frac{\epsilon}{\tilde{K}}$.
	\end{enumerate}
	In the first case we have (using \eqref{eq:t_dist_bound}):
	\begin{align*}
	\frac{\tpig(y,i)}{\tpig(x,i)} = & \frac{\pi(y)}{\pi(x)} \frac{\sum_{j\in \mathcal{I}} \wgj Q_j(\mu_j, \Sigma_{\gamma,j})(x)}{ \wgi Q_i(\mu_i, \Sigma_{\gamma,i})(x)}\frac{\wgi Q_i(\mu_i, \Sigma_{\gamma,i})(y)}{\sum_{j\in \mathcal{I}} \wgj Q_j(\mu_j, \Sigma_{\gamma,j})(y)} \\
	\leq & \frac{\pi(y)}{\pi(x)} \frac{\sum_{j\in \mathcal{I}} \wgj Q_j(\mu_j, \Sigma_{\gamma,j})(x)}{ \wgi Q_i(\mu_i, \Sigma_{\gamma,i})(x)} \leq \tilde{K} \frac{\pi(y)}{\pi(x)} \leq \epsilon.
	\end{align*}
	Similarly for $y \in \left(a(x) \cap B(x,K)\right)\setminus C_{\pi(x)}(\delta)$ we get
	$\frac{\tpig(x,i)}{\tpig(y,i)} \leq \epsilon$. Hence, on $B(x,K)\setminus C_{\pi(x)}(\delta)$ we have
	\begin{equation}\label{eq:B_minus_C_pi}
	\min \left[1, \frac{\tpig(y,i)}{\tpig(x,i)}\right] \frac{V_{\tpig}\left((y,i)\right)}{V_{\tpig}\left((x,i)\right)}= \min \left[\frac{\tpig\left(x,i\right)^{1/2}}{\tpig(y,i)^{1/2}} , \frac{\tpig\left(y,i\right)^{1/2}}{\tpig(x,i)^{1/2}}\right] \leq \epsilon^{1/2}.
	\end{equation}
	Furthermore, by assumption \eqref{eq:lambda_L_i_def} we can choose $R_3$ such that for $|x|> R_3$
	\begin{equation}\label{eq:local_kernel_inequality_part_3}
	\sup_{\gamma \in \mathcal{Y}} \int_{r_{\gamma,i}(x)}\RLi(x,y) dy \leq \lambda_{L,i} +\epsilon.
	\end{equation}
	Finally, for  $|x| > \max[R_0, R_1, R_2,R_3]$  we obtain
	\begin{align*}
	\frac{\tpigLi V_{\tpig} \left((x,i)\right)}{V_{\tpig}\left((x,i)\right)} = & \int_{\mathcal{X}}\RLi(x,y) \min \left[1, \frac{\tpig(y,i)}{\tpig(x,i)}\right] \frac{V_{\tpig}\left((y,i)\right)}{V_{\tpig}\left((x,i)\right)}dy &\\
	& \quad  +   
	\int_{\mathcal{X}}\RLi(x,y)
	\left(1 - \min\left[1 , \frac{\tpig\left(y,i\right)}{\tpig\left(x,i\right)}\right] \right) dy  \\
	=& \int_{\mathcal{X}}\RLi(x,y) \min \left[\frac{\tpig\left(x,i\right)^{1/2}}{\tpig(y,i)^{1/2}}, \frac{\tpig\left(y,i\right)^{1/2}}{\tpig(x,i)^{1/2}}\right] dy & \\
	& \quad  +   
	\int_{r_{\gamma,i}(x)}\RLi(x,y)
	\left(1 - \min\left[1 , \frac{\tpig\left(y,i\right)}{\tpig\left(x,i\right)}\right] \right) dy \\
	=& \int_{\mathcal{X}\setminus B(x,K)}\RLi(x,y) \min \left[\frac{\tpig\left(x,i\right)^{1/2}}{\tpig(y,i)^{1/2}}, \frac{\tpig\left(y,i\right)^{1/2}}{\tpig(x,i)^{1/2}}\right] dy \\
	& \quad + 
	\int_{ B(x,K)\cap C_{\pi(x)}(\delta)}\RLi(x,y) \min \left[\frac{\tpig\left(x,i\right)^{1/2}}{\tpig(y,i)^{1/2}}, \frac{\tpig\left(y,i\right)^{1/2}}{\tpig(x,i)^{1/2}}\right]dy \\
	\text{(see \eqref{eq:B_minus_C_pi})} \quad & \quad \quad + 
	\int_{B(x,K)\setminus C_{\pi(x)}(\delta)}\RLi(x,y) \min \left[\frac{\tpig\left(x,i\right)^{1/2}}{\tpig(y,i)^{1/2}}, \frac{\tpig\left(y,i\right)^{1/2}}{\tpig(x,i)^{1/2}}\right]dy  \\
	& \quad \quad \quad  + 
	\int_{r_{\gamma,i}(x)}\RLi(x,y) 
	\left(1 -  \frac{\tpig\left(y,i\right)}{\tpig\left(x,i\right)} \right) dy \\
	(\leq \epsilon \text{ by \eqref{eq:local_kernel_inequality_part_1}}) \leq & \int_{\mathcal{X}\setminus B(x,K)}\RLi(x,y) dy  \\
	(\leq \epsilon \text{ by \eqref{eq:local_kernel_inequality_part_2}}) \quad & \quad  + 
	\int_{ B(x,K)\cap C_{\pi(x)}(\delta)}\RLi(x,y) dy & \\
	(\leq \epsilon^{1/2}) \quad & \quad \quad + 
	\int_{B(x,K)\setminus C_{\pi(x)}(\delta)}\RLi(x,y) \epsilon^{1/2} dy & \\
	(\leq \lambda_{L,i} +\epsilon  \text{ by \eqref{eq:local_kernel_inequality_part_3}}) \quad & \qquad \qquad \qquad +
	\int_{r_{\gamma,i}(x)}\RLi(x,y) dy\\
	\leq & \lambda_{L,i}  + 3\epsilon  + \epsilon^{1/2},
	\end{align*}
	which ends the proof of \eqref{eq:epsilon_bounding}. Consequently, by setting  $\lambda_L$ such that $\max_{i\in \mathcal{I}} \lambda_{i,L} < \lambda_L <1$, we obtain~\eqref{eq:local_kernels_limsup}. Observe that there exists $R_L>0$ such that if $|x| >R_L$, then
	$$
	\tpigLi V_{\tpig}\left((x,i)\right) \leq \lambda_{L} V_{\tpig}\left((x,i)\right).
	$$
	For $|x| \leq R_L$ we have
	\begin{align*}
	\sup_{|x| <R_L} \sup_{\gamma \in \mathcal{Y}} \tpigLi V_{\tpig}\left((x,i)\right) \leq &  \sup_{|x| <R_L} \sup_{\gamma \in \mathcal{Y}} \frac{\tpigLi V_{\tpig} \left((x,i)\right)}{V_{\tpig}\left((x,i)\right)} \sup_{|x| <R_L} \sup_{\gamma \in \mathcal{Y}}  V_{\tpig}\left((x,i)\right).
	\end{align*}
	Now analogously to $r_{\gamma,i}(x)$, let us define the acceptance region for $\tpig$ as
	\begin{equation}\label{eq:acc_region_gamma}
	a_{\gamma,i}(x) = \{y \in \mathcal{X}: \tpi_{\gamma}(y,i)\geq \tpi_{\gamma}(x,i)\}.
	\end{equation}
	Note that 
	\begin{align*}
	\frac{\tpigLi V_{\tpig} \left((x,i)\right)}{V_{\tpig}\left((x,i)\right)} = & \int_{a_{\gamma,i}(x)} R_{\gamma,i}(x,y) \frac{V_{\tpig}\left((y,i)\right)}{V_{\tpig}\left((x,i)\right)} dy & \\ 
	& \quad  + 
	\int_{r_{\gamma,i}(x)} \RLi(x,y)  \frac{\tpig(y,i)}{\tpig(x,i)}\frac{V_{\tpig}\left((y,i)\right)}{V_{\tpig}\left((x,i)\right)} dy & \\
	& \quad \quad + 
	\int_{r_{\gamma,i}(x)} \RLi(x,y)  \left(1 - \frac{\tpig(y,i)}{\tpig(x,i)} \right)dy \\ 
	=  & \int_{a_{\gamma,i}(x)} \RLi(x,y) \frac{\tpig\left(x,i\right)^{1/2}}{\tpig(y,i)^{1/2}} dy \\
	& \quad  +
	\int_{r_{\gamma,i}(x)} \RLi(x,y)  \left(1 - \frac{\tpig(y,i)}{\tpig(x,i)}  + \frac{\tpig(y,i)^{1/2}}{\tpig(x,i)^{1/2}}\right)dy  \\
	\leq & 2 \int_{\mathcal{X}} \RLi(x,y) dy =2.
	\end{align*}
	Besides
	$$
	\sup_{|x| <R_L} \sup_{\gamma \in \mathcal{Y}}  V_{\tpig}\left((x,i)\right) <\infty 
	$$
	as for each $i$ the function $V_{\tpig}\left((x,i)\right)$  is  jointly continuous with respect to $x$ and $\gamma$.
	By setting 
	$$ 
	b_L:=2\max_{i \in \mathcal{I}}\sup_{|x| <R_L} \sup_{\gamma \in \mathcal{Y}}  V_{\tpig}\left((x,i)\right)
	$$ we obtain
	
	\begin{equation}\label{eq:local_drift_condition}
	\tpigLi V_{\tpig}\left((x,i)\right) \leq \lambda_{L} V_{\tpig}\left((x,i)\right) + b_L
	\end{equation}
	for all $(x,i)\in \mathcal{X} \times \mathcal{I}$.
\end{proof}
\newpage

\noindent\subsubsection*{\bf Assumption b) of Lemma \ref{thm:drift_condition} (jump kernels)}
\begin{proof}
	Firstly recall that under assumption e1) of Theorem \ref{thm:light_tailed_proposal} we have:
	\begin{align}\label{eq:P_combined}
	\begin{split}
	\tPg V_{\tpig}\left((x,i)\right)  =& \sum_{k \in \mathcal{I}} \int_{\mathcal{X}} \tPg\left((x,i), (dy,k)\right)V_{\tpig}\left((y,k)\right)\\
	=&  (1-\epsilon) \int_{\mathcal{X}} \tpigLi\left((x,i), (dy,i)\right) V_{\tpig}\left((y,i)\right)  \\
	& \quad + \epsilon\sum_{k\in \mathcal{I}} a_{\gamma,ik}  \int_{\mathcal{X}} \tpigJk\left((x,i), (dy,k)\right) V_{\tpig}\left((y,k)\right).
	\end{split}
	\end{align}
	if $x$ belongs to the jumping region $JR_{\gamma,i}$, and $\tPg V_{\tpig}\left((x,i)\right) = \int_{\mathcal{X}} \tpigLi\left((x,i), (dy,i)\right) V_{\tpig}\left((y,i)\right)$ otherwise. 
	Recall as well that all the jumping regions $JR_{\gamma,i}$ for $\gamma \in \mathcal{Y}$, $i \in \mathcal{I}$ are contained within a compact set $D$ and consequently any point $(y,k)$ proposed in a deterministic jump satisfies $(y,k) \in D\times \{k\}$. Let us now define
	$$
	b_J : = \sup_{\giY} \max_{k \in \mathcal{I}} \sup_{y \in D}V_{\tpig}\left((y,k)\right) < \infty.
	$$
	Observe that 
	$$
	\sup_{\giY} \max_{k \in \mathcal{I}}\sup_{x\in D} \int_{\mathcal{X}} \tpigJk\left((x,i), (dy,k)\right) V_{\tpig}\left((y,k)\right) \leq b_J
	$$
	and so  for all $(x,i)$ 
	\begin{align*}	 
	\tPg V_{\tpig}\left((x,i)\right)  \leq & \int_{\mathcal{X}} \tpigLi\left((x,i), (dy,i)\right) V_{\tpig}\left((y,i)\right) \\ &\quad+ \sum_{k\in \mathcal{I}} a_{\gamma,ik}   \int_{\mathcal{X}} \tpigJk\left((x,i), (dy,k)\right) V_{\tpig}\left((y,k)\right)\\
	\leq & \lambda_{L} V_{\tpig}\left((x,i)\right) + b_L +b_J.
	\end{align*}
	Finally, setting $\lambda:= \lambda_L$ and $b:=b_L+b_J$ yields \eqref{eq:main_drift_condition}  under assumption e1).
	
	Let us now consider assumption e2).  Recall that for any $s \in (0,1)$ if $V_{\tpig}\left((x,i)\right) = c \tpig(x,i)^{-s}$, then \eqref{eq:local_drift_condition} holds for some $\lambda_L$, $b_L$ and~$R_L$.
	Furthermore,
	\begin{align}\label{eq:jump_drift_condition}
	\begin{split}
	&\int_{\mathcal{X}} \tpigJk\left((x,i), (dy,k)\right) V_{\tpig}\left((y,k)\right)   \\
	= & \int_{\mathcal{X}}\RJk(y)\min \left[1, \frac{\tpig(y, k)}{\tpig(x, i)} \frac{\agki \RJi(x)} {\agik \RJk(y)} \right] V_{\tpig}\left((y,k)\right) dy  \\
	& \quad +
	\left(1 - \int_{\mathcal{X}} \RJk(y) \min \left[1, \frac{\tpig(y, k)}{\tpig(x, i)} \frac{\agki \RJi(x)} {\agik \RJk(y)} \right] dy\right) V_{\tpig}\left((x,i)\right) \\
	\leq &  \int_{\mathcal{X}}\RJk(y)  V_{\tpig}\left((y,k)\right) dy +  V_{\tpig}\left((x,i)\right).
	\end{split}
	\end{align}
	By assumption \eqref{eq:light_tailed_proposal} there exists a constant $c_2$ such that  $\frac{\RJi(x)}{\pi(x)^{s_J}} < c_2$  for  each $x \in \mathcal{X}$, $i \in \mathcal{I}$ and $\gamma \in \mathcal{Y}$ and as a consequence,
	$$
	\frac{\RJi(x)}{\tpig(x,i)^{s_J}} = \frac{\RJi(x)}{\pi(x)^{s_J}} \left(\frac{\sum_{j\in \mathcal{I}} \wgj Q_j(\mu_j, \Sigma_{\gamma,j})(x)}{\wgi Q_i(\mu_i, \Sigma_{\gamma,i})(x)} \right)^{s_J} \leq c_2 \tilde{K}^{s_J},
	$$
	where the last inequality follows from \eqref{eq:t_dist_bound}. Fix $s<s_J$ and observe that 
	\begin{align}\label{eq:jump_estimation_b} 
	\begin{split}
	b_J:= \sup_{\giY} \max_{k\in \mathcal{I}}\int_{\mathcal{X}} \RJk(y)  V_{\tpig}\left((y,k)\right) dy  = & \sup_{\giY} \max_{k\in \mathcal{I}}\int_{\mathcal{X}} \RJk(y) \tpig(y,k)^{-s} dy \\
	= & \sup_{\giY} \max_{k\in \mathcal{I}} \int_{\mathcal{X}} \frac{\RJk(y)}{\tpig(y,k)^{s_J}} \tpig(y,k)^{s_J-s} dy \\
	\leq &\sup_{\giY} \max_{k\in \mathcal{I}} c_2 \tilde{K}^{s_J} \int_{\mathcal{X}} \tpig(y,k)^{s_J-s} dy \\
	\leq  & c_2 \tilde{K}^{s_J}\int_{\mathcal{X}} \pi(y)^{s_J-s} dy  < \infty,
	\end{split}
	\end{align}
	where the last inequality follows from $\pi$ being super-exponential and $s_J-s$ positive.	
	
	Recall additionally that under e2) \eqref{eq:P_combined} holds for all $(x,i) \in \mathcal{X}\times\mathcal{I}$.
	Putting together \eqref{eq:local_drift_condition}, \eqref{eq:jump_drift_condition}, \eqref{eq:jump_estimation_b} and \eqref{eq:P_combined} yields	
	\begin{align*}
	\tPg V_{\tpig} \left((x,i)\right) \leq &(1-\epsilon) \lambda_L V_{\tpig} \left((x,i)\right) + (1-\epsilon) b_L + \epsilon\sum_{k \in \mathcal{I}} a_{\gamma,ik}  b_{J} + \epsilon V_{\tpig} \left((x,i)\right) \\
	= & \left((1-\epsilon)\lambda_L + \epsilon\right) V_{\tpig} \left((x,i)\right) +  (1-\epsilon) b_L + \epsilon b_{J}.
	\end{align*}
	By setting $\lambda: = (1-\epsilon)\lambda_L + \epsilon$ and $b:= (1-\epsilon) b_L + \epsilon b_J$, we
	obtain the drift condition as given by \eqref{eq:main_drift_condition}.
\end{proof}
\noindent\subsubsection*{ \bf Assumption c) of Lemma \ref{thm:drift_condition}}
\begin{proof}
	Proving the minorisation condition \eqref{eq:minorisation_condititon} amounts to specifying $n_0,$ $\delta$, $\nu_{\gamma}$ and $v$, and verifying that $\tPg^{n_0}(\tilde{x}, B) \geq \delta \nu_{\gamma}(B)$ for all measurable sets $B$ and all $\tilde{x}$ satisfying $V_{\tpig}(\tilde{x}) \leq v$.  Let $C_v$ be defined as $C_v:= \{x\in \mathcal{X}: c\pi(x)^{-s} \leq v\}$, where $c$ is defined in \eqref{equation_drift}. We specify the value of $v$ below, separately for assumptions e1) and e2). 
	
	Note that if $x \in C_v$, then $V_{\tpig}\left((x,i)\right) \leq v$ for each $i\in \mathcal{I}$ and each $\gamma \in \mathcal{Y}$.  Observe also that $C_v$ is a compact set. Let $\nu_{\gamma}$ be the uniform distribution on $C_v \times \mathcal{I}$ (and 0 everywhere else)
	i.e. for $A \subseteq C_v$ we have $\nu_{\gamma}(A\times\{i\}) = \frac{1}{N}\frac{\mu^{\text{Leb}}(A)}{\mu^{\text{Leb}}(C_v)}$. To prove the claim, it is enough to show that $$\tPg^{n_0}(\tilde{x}, \hat{B}) \geq \delta\nu_{\gamma}\left(\hat{B}\right)$$ for $\hat{B}$ of the form $B \times \{k\}$, for any $B\subseteq C_v$ and any $k\in \mathcal{I}$. 
	
	Firstly, note that for any $i, k \in \mathcal{I}$
	\begin{align}\label{eq:assumption_c_augmented_target}
	\begin{split}
	\inf_{x,y \in C_v} \inf_{\giY}  \min \left[1, \frac{\tpig(y,k)}{\tpig(x,i)}\right] \geq \inf_{x,y \in C_v} \inf_{\giY}  \min \left[1, \frac{\pi(y)}{\pi(x)} \frac{\wgk Q_k(\mu_k, \Sigma_{\gamma,k})(y)}{\sum_{j\in \mathcal{I}} \wgj Q_j(\mu_j, \Sigma_{\gamma,j})(y)} \right] \geq 0,
	\end{split}
	\end{align}
	where the last inequality is satisfied by assumption c) and equation \eqref{eq:t_dist_bound}.

	We will first focus on verifying the minorisation condition under assumption~e1). Recall that $D$ is a compact set in $\mathcal{X}$ such that for each $\gamma \in \mathcal{Y}$ and  $i \in \mathcal{I}$ we have $JR_{\gamma, i} \subseteq D$.  Recall also that by the construction of the jumping regions there exists $r_1$ such that for each $\gamma \in \mathcal{Y}$ and  $i \in \mathcal{I}$ the ball $B(\mu_i, r_1) \subseteq JR_{\gamma, i}$. Let us now pick $v$ so large that $D \subseteq C_v$ and $v >2n_0b/(1-\lambda^{n_0})$ for $n_0 =3$. 
	
	The minorisation condition will be proved for $n_0=3$.	Indeed, three steps of the algorithm are enough to get from a point $(x,i)$ to a set $B\times \{k\}$ (a local step within mode $i$ to reach its jumping region, a jump to mode $k$ and a local move within mode $k$ to set $B$).  
	
	Fix $(x,i) \in C_v\times \mathcal{I}$ and a set $\hat{B} = B \times \{k\}$ for $B \subseteq C_v$ (we allow for the case $k=i$). Note that since $JR_{\gamma,i} \subset C_v$ for all $i\in \mathcal{I}$ and $\giY$ we have
	\begin{align}\label{eq:proof_ass_c_e1_1}
	\begin{split}
	& \inf_{x\in C_v} \inf_{\giY} \tpigLi\left((x,i),JR_{\gamma, i} \times \{i\}\right) \\
	\geq &\inf_{x\in C_v}\inf_{\giY}\int_{JR_{\gamma, i}}\RLi(x,y) \min \left[1, \frac{\tpig(y,i)}{\tpig(x,i)}\right] dy\\
	\geq & \inf_{x, y\in C_v}\inf_{\giY} \RLi(x,y) \inf_{x, y\in C_v} \inf_{\giY}  \min \left[1, \frac{\tpig(y,i)}{\tpig(x,i)}\right] \mu^{\text{Leb}}(JR_{\gamma, i}) \\
	\geq &	\inf_{x, y\in C_v}\inf_{\giY} \RLi(x,y) \inf_{x, y\in C_v} \inf_{\giY}  \min \left[1, \frac{\tpig(y,i)}{\tpig(x,i)}\right] \mu^{\text{Leb}}\left(B(\mu_i, r_1)\right) \\
	:= & p_{1,i}.
	\end{split}
	\end{align}
	By equations \eqref{eq:local_proposals_compact_set_assumption} and \eqref{eq:assumption_c_augmented_target} we get that $p_{1,i}$ defined above is strictly positive for $i \in \mathcal{I}$. Considering the probability of accepting a deterministic jump from mode $i$ to mode $k$, we obtain  
	\begin{align}\label{eq:proof_ass_c_e1_2}
	\begin{split}
	& \inf_{x\in JR_{\gamma, i}} \inf_{\giY} \tilde{P}_{\gamma, J, k} \left((x,k), B \times\{k\} \right) \\
	\geq	&\inf_{x\in JR_{\gamma,i}} \inf_{ y \in  JR_{\gamma,k} }  \inf_{\giY} \min\left[1, \frac{\tilde{\pi}(y,k)}{\tilde{\pi}(x,i)} \frac{ a_{\gamma, ki} \sqrt{\det \Sigma_{\gamma,k}}}{ a_{\gamma, ik} \sqrt{\det \Sigma_{\gamma,i}} } \right] \\
	\geq & \inf_{x, y\in C_v} \inf_{\giY} \min\left[1, \frac{\tilde{\pi}(y,k)}{\tilde{\pi}(x,i)} \frac{ a_{\gamma, ki} \sqrt{\det \Sigma_{\gamma,k}}}{ a_{\gamma, ik} \sqrt{\det \Sigma_{\gamma,i}} } \right] \\
	:= & p_{2, ik}.
	\end{split}
	\end{align}
	It follows from equations \eqref{eq:assumption_c_augmented_target}, \eqref{eq:weights_epsilon} and \eqref{eq:matrices_estimation} that $p_{2, ik}>0$ for  $i, k\in \mathcal{I}$.
	Analogous arguments show that	
	\begin{align}\label{eq:proof_ass_c_e1_3}
	\begin{split}
	& \inf_{x\in JR_{\gamma, k}} \inf_{\giY} \tpigLk\left((x,k), B \times\{k\} \right) \\
	\geq &\inf_{x\in C_v}\inf_{\giY}\int_{B}\RLk(x,y) \min \left[1, \frac{\tpig(y,k)}{\tpig(x,k)}\right] dy\\
	\geq & \underbrace{\inf_{x, y\in C_v}\inf_{\giY} \RLk(x,y) \inf_{x, y\in C_v} \inf_{\giY}  \min \left[1, \frac{\tpig(y,k)}{\tpig(x,k)}\right]}_{:=p_{3,k}} \mu^{\text{Leb}}(B).
	\end{split}
	\end{align}
	Combining \eqref{eq:proof_ass_c_e1_1}, \eqref{eq:proof_ass_c_e1_2} and \eqref{eq:proof_ass_c_e1_3} yields
	$$
	\tPg^{3}((x,i), B\times \{k\}) \geq (1-\epsilon)^2 \epsilon \epsilon_a p_{1,i}p_{2, ik} p_{3,k} \mu^{\text{Leb}}(B).
	$$
	Setting $\delta:=  (1-\epsilon)^2 \epsilon \epsilon_a \min_{i, k \in \mathcal{I}}p_{1,i}p_{2, ik} p_{3,k} N \mu^{\text{Leb}}(C_v)$ ends the proof.
	
	We will now verify the minorisation condition under assumption~e2). Let $v$ be so large that $B(\mu_i, r)\subseteq C_v$ (see assumption \eqref{eq:jump_proposals_compact_set_assumption}) for $i \in \mathcal{I}$ and $v >2n_0b/(1-\lambda^{n_0})$, for $n_0=3$.  We will prove that indeed the minorisation condition holds for $n_0=3$. Note that if we want to move from $(x,i)$ to a set $B \times \{k\}$, it is enough to make a local step to $B(\mu_i,r)$, and then a jump to $B(\mu_k,r)$ followed by a local step to $B$.
	
	As before,  fix $(x,i) \in C_v\times \mathcal{I}$ and a set $\hat{B} = B \times \{k\}$ for $B \subseteq C_v$. Again we include the case $k=i$. Analogous calculations to \eqref{eq:proof_ass_c_e1_1} show that 
	\begin{align}\label{eq:proof_ass_c_e2_1}
	\begin{split}
	& \inf_{x\in C_v} \inf_{\giY} \tpigLi \left((x,i), B(\mu_i, r) \times\{i\} \right) \\
	\geq &	\inf_{x, y\in C_v}\inf_{\giY} \RLi(x,y) \inf_{x, y\in C_v} \inf_{\giY}  \min \left[1, \frac{\tpig(y,i)}{\tpig(x,i)}\right] \mu^{\text{Leb}}\left(B(\mu_i, r)\right) \\
	:= & p_{4,i}>0.
	\end{split}
	\end{align}
	For the jump kernel involved we obtain
	\begin{align}\label{eq:proof_ass_c_e2_2}
	\begin{split}
	& \inf_{x\in B(\mu_i,r)} \inf_{\giY} \tilde{P}_{\gamma, J, k} \left((x,i), B(\mu_k, r) \times\{k\} \right) \\
	\geq &  \inf_{x\in B(\mu_i,r)} \inf_{\giY} \int_{ B(\mu_k, r)} \RJk(y)  \min \left[1, \frac{\tpig(y, k)}{\tpig(x, i)} \frac{\agki \RJi(x)} {\agik \RJk(y)} \right] dy \\
	\geq &  \inf_{x\in B(\mu_i,r)} \inf_{y \in B(\mu_k, r) } \inf_{\giY} \RJk(y) \min \left[1, \frac{\tpig(y, k)}{\tpig(x, i)} \frac{\agki \RJi(x)} {\agik \RJk(y)} \right] \mu^{\text{Leb}} \left( B(\mu_k, r) \right)\\
	:= & p_{5,ik}.
	\end{split}
	\end{align}
	Note that $p_{5,ik}$ is positive by equations  \eqref{eq:assumption_c_augmented_target} and \eqref{eq:weights_epsilon}, and assumption e2). Finally, similar calculations to \eqref{eq:proof_ass_c_e1_3} yield
	\begin{align}\label{eq:proof_ass_c_e2_3}
	\begin{split}
	& \inf_{x\in B(\mu_k, r)} \inf_{\giY} \tpigLk\left((x,k), B \times\{k\} \right) \\
	\geq &\inf_{x\in C_v}\inf_{\giY}\int_{B}\RLk(x,y) \min \left[1, \frac{\tpig(y,k)}{\tpig(x,k)}\right] dy\\
	\geq & \underbrace{\inf_{x, y\in C_v}\inf_{\giY} \RLk(x,y) \inf_{x, y\in C_v} \inf_{\giY}  \min \left[1, \frac{\tpig(y,k)}{\tpig(x,k)}\right]}_{=p_{3,k}} \mu^{\text{Leb}}(B).
	\end{split}
	\end{align}
	for $p_{3,k}$ defined in the previous part of the proof. We now combine \eqref{eq:proof_ass_c_e2_1}, \eqref{eq:proof_ass_c_e2_2} and \eqref{eq:proof_ass_c_e2_3} to get
	$$
	\tPg^{3}((x,i), B\times \{k\}) \geq (1-\epsilon)^2 \epsilon \epsilon_a p_{4,i}p_{5, ik} p_{3,k} \mu^{\text{Leb}}(B).
	$$
	Setting $\delta:=  (1-\epsilon)^2 \epsilon \epsilon_a \min_{i, k \in \mathcal{I}}p_{4,i}p_{5, ik}p_{3,k} N \mu^{\text{Leb}}(C_v)$ ends the proof.
\end{proof}	
\subsubsection*{\bf Proof of Theorem \ref{thm:lln_jams}}
\begin{proof}
	This theorem is a direct corollary from Theorem \ref{thm:lln}. The assumptions of this theorem were verified in the proofs of Theorems \ref{thm:heavy_tailed_proposal} or  \ref{thm:light_tailed_proposal}, under the uniform and the non-uniform scenario, respectively. 
\end{proof}
\subsubsection*{\bf Proof of Lemma \ref{thm:useful_lemma_1}}
\begin{proof}
	Fix $i \in \mathcal{I}$ and let $\epsilon_L$ be such that for $|x|$ larger than some $R_0$
	$$
	\int_{a(x)}R_{\gamma^{*}, L, i}(x,y) \geq \epsilon_L,
	$$
	(such $\epsilon_L$ can be found due to assumption \eqref{eq:assumption_single_gamma}). Hence, for $K$ sufficiently large 
	$$
	\int_{a(x) \cap B(x, K)}R_{\gamma^{*}, L, i}(x,y) \geq \frac{\epsilon_L}{2 },
	$$
	which implies that for any $|x|>R_0$ 
	$$
	\mu^{\text{Leb}}\left(a(x) \cap B(x, K)\right) \geq \frac{\epsilon_L}{2\sup_{y \in B(x,K) }R_{\gamma^{*}, L, i}(x,y) } = \frac{\epsilon_L}{2\sup_{y \in B(\textbf{0},K) }R_{\gamma^{*}, L, i}(\textbf{0},y) }
	$$
	and consequently 
	\begin{align}\label{eq:lemma_upper_bound}
	\begin{split}
	& \inf_{\giY}\int_{a(x) \cap B(x, K)}\RLi(x,y)dy \\
	\geq  & \mu^{\text{Leb}}\left(a(x) \cap B(x, K)\right) \inf_{\giY} \inf_{y \in B(x,K)}\RLi(x,y) \\
	\geq & \frac{\epsilon_L \inf_{\giY} \inf_{y \in B(\textbf{0},K)}\RLi(\textbf{0},y)}{2\sup_{y \in B(\textbf{0},K) }R_{\gamma^{*}, L, i}(\textbf{0},y) }=: \tilde{\epsilon}_L>0. 
	\end{split}
	\end{align}
	By assumption a) of Theorem \ref{thm:light_tailed_proposal} $\tilde{\epsilon}_L$ is indeed positive. 
	
	Let the acceptance region $a_{\gamma,i}(x)$ be given by \eqref{eq:acc_region_gamma}. We will show that for $|x|$ sufficiently large and for each $\giY$
	\begin{equation*}
	\int_{a_{\gamma,i}(x)} \RLi(x,y) dy \geq \frac{\tilde{\epsilon}_L}{2},
	\end{equation*}
	which will prove the claim. We shall now repeat similar arguments to those used in the proof of Theorem~\ref{thm:light_tailed_proposal}, in the part for the local kernels. Firstly, we use formula \eqref{eq:jarner_hansen_formula} to conclude that for $|x|$ larger than some $R_1$ (which may depend on $K$) and for sufficiently small $\delta$ (which may depend on $K$, $R_1$ and $\tilde{\epsilon}_L$ ), we have 
	\begin{equation}\label{eq:lemma_lower_bound}
	\sup_{\gamma \in \mathcal{Y}}\int_{C_{\pi(x)}(\delta) \cap B(x,K)}\RLi(x,y ) dy\leq \frac{\tilde{\epsilon}_L}{2}.
	\end{equation}
	We put \eqref{eq:lemma_upper_bound} together with \eqref{eq:lemma_lower_bound} to obtain
	\begin{align*}
	\inf_{\giY} \int_{\left(a(x) \cap B(x,K)\right)\setminus C_{\pi(x)}(\delta) }\RLi\left(x,y\right)dy \geq \frac{\tilde{\epsilon}_L}{2}
	\end{align*}
	for $|x|> \max\left[R_0,R_1\right]$. Now recall that for each $\delta$ there exists $R_2$ such that for $|x|>R_2$ if
	$y \in \left(a(x) \cap B(x,K)\right)\setminus C_{\pi(x)}(\delta)$ then $\frac{\pi(y)}{\pi(x)} \geq \tilde{K}$  for $\tilde{K}$ defined in  \eqref{eq:t_dist_bound}. Therefore in particular $y \in a_{\gamma,i}(x)$ for each $\giY$. Finally, for $|x|> \max\left[R_0,R_1, R_2\right]$ we have
	
	$$
	\inf_{\giY} \int_{a_{\gamma,i}(x) }\RLi\left(x,y\right) dy \geq \inf_{\giY} \int_{\left(a(x) \cap B(x,K)\right)\setminus C_{\pi(x)}(\delta)}\RLi\left(x, y \right) dy \geq \frac{\tilde{\epsilon}_L}{2},
	$$
	which ends the proof.
\end{proof}
\subsubsection*{\bf Proof of Lemma \ref{thm:useful_lemma_2}}
\begin{proof}
	Fix any $\giY$ and $i \in \mathcal{I}$.  
	To prove the required result, we will use analogous arguments to those from the proof of Theorem 4.3 of \cite{jarner2000geometric}. Let $\epsilon>0$ and $R$ be such that for $|x| >R$
	$$
	\frac{x}{|x|} \cdot \frac{\nabla \pi(x)}{|\nabla \pi(x)|} \leq - \epsilon.
	$$
	Fix $K>0$ and define the cone $W(x)$ as
	$$
	W(x) := \left\{x- a \xi: 0<a<K, \xi \in S^{d-1}, \left|\xi - \frac{x}{|x|}\right| \leq \frac{\epsilon}{2} \right\}.
	$$
	We now refer to the proof of Theorem 4.3 of \cite{jarner2000geometric} to see that for $x$ sufficiently large $W(x) \subset a(x)$. What is more, 
	$$
	\liminf_{|x| \to \infty} \int_{W(x)} \RLi(x,y) dy >0 \quad\text{and so} \quad \liminf_{|x| \to \infty} \int_{a(x)} \RLi(x,y) dy >0.
	$$
	Hence, since $i$ was chosen arbitrarily, assumption \eqref{eq:assumption_single_gamma} is satisfied for $\gamma^{*} := \gamma$.
	
	We would like to point out here that originally Theorem 4.3 of \cite{jarner2000geometric} was proved under a stronger assumption, that is, $\RLi(x, y) = \RLi(|x-y|)$. However, careful inspection of this proof shows that it is enough to assume that $\RLi(x, y)= \RLi(y, x)$, which is satisfied in our case as $\RLi$ follows an elliptical distribution.
\end{proof}

\subsubsection*{ \bf Proof of Corollary \ref{thm:corrollary}}
\begin{proof}
	We will again refer multiple times to \cite{jarner2000geometric}. Firstly, by Theorem 4.4 of this paper, if $\pi_1$ and $\pi_2$ are super-exponential and satisfy \eqref{eq:simplifying_assumption}, then also $a_1 \pi_1 + a_2 \pi_2$ is super-exponential and satisfies \eqref{eq:simplifying_assumption} for positive $a_1$ and $a_2$. By Theorem 4.6 of the same paper, each density of the form $\pi(x) \propto\exp\left(-p(x) \right)$  is super-exponential and satisfies \eqref{eq:simplifying_assumption}, if $p$ is a polynomial of order $\geq 2$. Therefore, the assumptions of Lemma \ref{thm:useful_lemma_2} hold, as required.
\end{proof}

\section{Other algorithms in the Auxiliary Variable Adaptive MCMC class}\label{section:other_algs}
As mentioned in Section \ref{section:auxiliary_variable}, an instance of an algorithm in the Auxiliary Variable Adaptive MCMC class is adaptive parallel tempering introduced by \cite{miasojedow2013adaptive}. Indeed, let us consider $\Phi:= \mathcal{X}^N$ and $\tilde{X} := \mathcal{X} \times \Phi = \mathcal{X} \times \mathcal{X}^N$ and
$$
\tpig\left(x_N, (x_0, \ldots, x_{N-1})\right): = \prod_{i=0}^N \pi(x_i)^{\beta_{j, \gamma}}.
$$
Then for any $B \in \mathcal{B}(\mathcal{X})$ we have
$$
\tpig\left(B \times \Phi\right) = \int_{B} \pi(x_N)^{\beta_{N, \gamma}} dx_N \int_{\Phi}\prod_{i=0}^{N-1} \pi(x_i)^{\beta_{j, \gamma}}dx_0 \ldots d x_{N-1} = \int_{B} \pi(x_N)^{\beta_{N, \gamma}} dx_N = \pi(B),
$$
where the last equality follows since $\beta_{N, \gamma}=1$ for all $\gamma$. Additionally, the transition kernels used in adaptive parallel tempering $\{ \tilde{P}_{\gamma} \}_{\gamma \in \mathcal{Y}}$ are defined in such a way that  detailed balance holds. 

Another example of a group of algorithms in the Auxiliary Variable Adaptive MCMC class is an adaptive version of pseudomarginal algorithms. Recall that pseudomarginal algorithms are a powerful tool used in situations when the target  density $\pi(x)$ on $\mathcal{X}$ cannot be evaluated pointwise or this evaluation would be very expensive, but an unbiased estimator of $\pi(x)$ is available. In the simplest setting an importance sampling estimator is used for this purpose. Then the pseudomarginal algorithm is equivalent to the Metropolis-Hastings algorithm targeting a distribution $\tilde\pi^N(x, z)$ on an augmented state space $\mathcal{X} \times \mathcal{Z}$, where $Z\in \mathcal{Z}$ is a vector representing $N$ samples on which the importance sampling estimator is based. A remarkable property of the pseudomarginal algorithms is that $\pi(x)$ is the marginal distribution of  $\tilde\pi^N(x, z)$ regardless of $N$ (see \cite{andrieu2009pseudo} and \cite{andrieu2015convergence}) . The number of samples $N$ and, in more complex settings, the amount of correlation between those samples (see  \cite{deligiannidis2018correlated}), may follow an adaptive scheme. Therefore, conditions $\eqref{ker_erg}$ and $\eqref{marginal_ok}$ are satisfied for $\Phi = \mathcal{Z}^{\mathbb{N}} \times \mathbb{N}$, where $N \in \mathbb{N}$ corresponds to the number of samples used for estimation.

\begin{center}
	\large \textbf{\textsc{Supplementary Material B}}
\end{center}
In Sections \ref{section:updating} and \ref{section:burnin} we present details of the implementation of our method.  An additional simulation example  and settings of our numerical experiments are shown in Section \ref{section:details_experiments}.

\section{Updating $w_{\gamma,i}$ and  $a_{\gamma,ik}$}\label{section:updating}
Recall that $N$ denotes the number of modes. The weights $w_{\gamma,i}$ are set to $1/N$ at the beginning of the algorithm and they are adapted while the algorithm runs in such a way that they represent the proportion of samples observed so far in each mode. At the same time we do not allow any of the weights to get below some pre-specified value $\epsilon_w$; otherwise the target distribution $\tpig$ could run the risk of being severely distorted by weights very close or equal to 0. In particular we use the update scheme described below.

Let $n_{i, \text{obs}}$ be the number of samples in mode $i$ for $i=1, \ldots, N$ observed after $n$ iterations of the main algorithm. Then $n= \sum_{i=1}^N n_{i, \text{obs}}$. Define 
$$
w_{\text{add}}:= \frac{n}{\frac{1}{\epsilon_w} - N} \quad \text{and set} \quad w_i := \frac{n_{i, \text{obs}} + w_{\text{add}}}{n + N w_{\text{add}}}\quad \text{for } i=1, \ldots, N.
$$
It is easily checked that if there are no observations in mode $i$, i.e. $n_{i, \text{obs}}$ is equal to 0, then $w_i = \epsilon_w$. Since $\epsilon_w$ must satisfy $N \epsilon_w <1$ and the number of modes $N$ is typically unknown in advance, in our implementation the user provides $\tilde{\epsilon}_w$ and the algorithm sets $\epsilon_w := \tilde{\epsilon}_w/N$.

Even though the theory we present in Section \ref{section:ergodicty_amcmc} holds when  $\agik$ follow some adaptive rule, we propose to keep these values fixed throughout the run of the algorithm, with a default choice $a_{\gamma, ii}=0$  and $a_{\gamma, ik} = 1/(N-1)$ for $\giY$, $i,k\in \{1, \ldots, N\}$ and $i \neq k$. If $N$ is not very large the benefit of adapting $\agik$ is rather marginal while it may add to finite sample instability. A natural alternative improving acceptance rates would be to keep $a_{\gamma, ii}=0$ and $a_{\gamma, ik} = w_{\gamma,k}/\sum_{j \neq i} w_{\gamma,j}$. However, consider a scenario when a mode with a significant weight in the target distribution is particularly difficult to jump into (for example, because the covariance matrix estimation has not been run for long enough). The jumps to this mode will very likely get rejected many times before we observe the first sample in this mode and start adapting its covariance matrix. In such case proposing modes proportionally to the number of samples collected so far in those modes would effectively make moves to this "difficult" mode even less frequent. Consequently we could face the risk of underestimating the weight of this mode for a fixed computational budget. Hence, we adopt the more conservative approach of keeping these values fixed to avoid the risk described above.

Note also that we in our implementation we use $a_{\gamma, ii}=0$ even though formally we assumed in Section \ref{section:ergodicty_amcmc} that $a_{\gamma, ii}> \epsilon_a$. This is because in practice we do not want to propose jumps to the same mode. In case of deterministic jumps this would mean proposing a move to the same state (recall equation \eqref{eq:deterministic_proposal}), which would have a negative impact on the mixing of the algorithm.

\section{Burn-in algorithm}\label{section:burnin}
For the mode finding part in our implementation we use the BFGS  method from the \texttt{optimx} package in \texttt{R} \cite{R}. We only pass to the next stage of the burn-in algorithm (mode-merging) those vectors for which first and second order Kuhn, Karush, Tucker (KKT) optimality conditions are satisfied. Checking these conditions is necessary  in order to avoid including points that are not local minima of $-\log(\pi)$ (but, for example, saddle points) in the list of modes. Besides, we recommend that the user codes up their own function for calculating the gradient and the Hessian, whenever possible, or uses packages that compute those values with high numerical precision. This will help ensure numerical stability of the optimisation runs. What is more, working with variables with bounded support tends to be problematic -- the optimisation algorithm will typically struggle in the neighbourhood of the boundary. In such cases it is usually beneficial to work with transformed variables, defined on the whole space (see Section \ref{section:loh_supplementary}).

Recall that the initial value of the matrix corresponding to mode $i$ at the beginning of round 1 of the covariance matrix estimation is the inverse of the Hessian evaluated at $\mu_i$ (see line 17 of Algorithm \ref{alg:burn_in_algorithm}). The heuristics behind this idea is that in case of the Gaussian distribution the inverse of the Hessian of $-\log(\pi)$ would correspond to the covariance matrix, so intuitively for a large class of target distributions this will be a good starting value.

As mentioned in Section \ref{subsection:burn_in_algorithm}, we propose a semi-automatic way of choosing the number of rounds of the covariance matrix estimation, denoted by $K$. Recall that $\Sigma_{k,i}$ is the matrix corresponding to mode $i$ updated during round $k$. The choice of $K$ is based on monitoring the following quantity, called  inhomogeneity factor (see \cite{roberts2009examples} and \cite{rosenthal2011optimal}), given by
\begin{equation}\label{eq:inhomogeneity_factor}
b_{k,i} := d\frac{\sum_{j=1}^d \lambda_j^{-1}}{\left( \sum_{j=1}^d \lambda_j^{-1/2}\right)^2},
\end{equation}
where $d$ is the dimension of the state space of $\pi$ and $\lambda_j$ for $j=1, \ldots, d$ are the eigenvalues of $\Sigma_{k-1,i}^{-1} \Sigma_{k,i}$. Note that this factor is always a real number even though  $\Sigma_{k-1,i}^{-1} \Sigma_{k,i}$ does not need to be symmetric. If $\lambda$ is a complex eigenvalue of $\Sigma_{k-1,i}^{-1} \Sigma_{k,i}$, its conjugate $\bar{\lambda}$ is also an eigenvalue of $\Sigma_{k-1,i}^{-1} \Sigma_{k,i}$ and so the imaginary components cancel both in the numerator and the denominator of $\eqref{eq:inhomogeneity_factor}$. Moreover, by Jensen's inequality $b_{k,i}$ satisfies
$b_{k,i} \geq 1,$ and $b_{k,i}=1$ if and only if $\Sigma_{k-1,i}$ and $\Sigma_{k,i}$ are proportional to each other. In particular, the value of $b_{k,i}$ is always equal to 1 in the scaling phase.

The procedure we propose is the following: perform $AC_1$ scaling steps for each matrix (perhaps split into several rounds). Then perform at least one covariance-based round for each mode. Continue running covariance-matrix rounds until the inhomogeneity factor drops below a certain threshold $b_{\text{acc}}$ for all matrices. In other words, having performed $AC_1$ scaling steps and at least one covariance-based round, we set $K$ to the smallest value of $k$ satisfying
$\max_{i \in \{1, \ldots, N \}} b_{k,i} \leq b_{\text{acc}}.
$
In the version in which the modes can be added when the main algorithm runs, one could consider stopping the burn-in algorithm separately for each mode and passing the covariance matrix to the main algorithm once its corresponding inhomogeneity factor goes below $b_{\text{acc}}$.

As for the choice of the lengths of the rounds, by default we use a geometric sequence with a common ratio 2. The number $AC_1$ should grow with the dimension of the state space $d$ since the initial covariance matrix will be based on $AC_1$ samples for each mode. In our experiments $AC_1$ is equal to $\max(1000, d^2/2)$.

Note that this construction  implies that adapting the matrices by scaling will happen only during the burn-in algorithm, as the number of samples in each mode at the beginning of the main algorithm will be equal to the total length of the number of iterations in the burn-in algorithm, so in particular this number will exceed $AC_1$. 

The adaptation scheme of the main algorithm is based on updating the covariance matrices passed from the burn-in algorithm.

\section{Examples --  further details}\label{section:details_experiments}
Below we present one more example, a hierarchical Bayesian model for cancer data. We also  discuss some further details related to the simulations described in Section \ref{section:examples}. The exact parameters settings of our experiments are summarised in Table \ref{table:experimens_settings}.  For all examples shown in this paper we used an implementation of the algorithm in which the burn-in algorithm runs only before the main algorithm (without adding modes on the fly).

\begin{table}
	\begin{center}
		\begin{tabular}{|c|c|c|c|c|c|}
			\hline
			\multicolumn{2}{|c|}{} &\makecell{Mixture \\of Gaussians} & \makecell{Mixture 
				\\of banana-shaped  \\ and t-distributions} &   \makecell{Sensor \\ network} & \makecell{LOH \\ example}\\
			\hline
			\multicolumn{2}{|c|}{\textit{Main algorithm}} & \multicolumn{4}{|c|}{}  \\
			\hline
			\multicolumn{2}{|c|}{number of iterations} & 500,000 & 500,000  & 500,000 & 200,000 \\
			\multicolumn{2}{|c|}{$\alpha$} & 0.7 &  0.7 & 0.7 &0.7 \\
			\multicolumn{2}{|c|}{$\beta$} & 0.0001 & 0.0001 &0.0001 &0.0001\\
			\multicolumn{2}{|c|}{$\epsilon$} & 0.1 & 0.1 & 0.1 & 0.1\\
			\multicolumn{2}{|c|}{$\tilde{\epsilon_w}$} & 0.01 & 0.01 & 0.01 & 0.01 \\
			\multicolumn{2}{|c|}{$AC_2$} & 1000 &  1000 & 500 &500\\
			\multicolumn{2}{|c|}{optimal acceptance rate} &  0.234 & 0.234 & 0.234 & 0.234\\
			\multicolumn{2}{|c|}{local proposal} & Gaussian & Gaussian & Gaussian & \makecell{Gaussian/ \\ $t$-distributed}\\
			\multicolumn{2}{|c|}{distributions $Q_i$} & $t$ with 7 df & $t$ with 7 df & $t$ with 7 df& $t$ with 7 df\\
			\multicolumn{2}{|c|}{\makecell{df of the proposal  \\ (if $t$-distributed)}} &7 &7 &7 &7 \\
			\hline
			\multicolumn{2}{|c|}{\textit{Burn-in algorithm}} & \multicolumn{4}{|c|}{} \\
			\hline
			\multicolumn{2}{|c|}{number of BFGS runs} & 1500 & 40,000 &  10,000 & 500  \\
			\multicolumn{2}{|c|}{$b_{\text{acc}}$} &1.1  & 1.1 & 1.1 & 1.1 \\
			\hline	
		\end{tabular}
	\end{center}
	\caption{Settings of the parameters used for the examples presented in this paper.}\label{table:experimens_settings}
\end{table}
\subsection{Hierarchical Bayesian model for LOH data}\label{section:loh_supplementary}
The example presented here is based on  the Seattle Barrett’s Oesophagus study (see \cite{barrett1996determination})  analysed later by \cite{warnes2000normal}, \cite{craiu2009learn} and \cite{bai2011divide}. Loss of Heterozygosity is the process by which a region of the genome is deleted on either the paternal or maternal inherited chromosomes leading to a loss of diversity. Loss of Heterozygosity (LOH) rates were collected from oesophageal cancers for 40 regions,
each on a distinct chromosome arm. They are of interest since chromosome regions with high rates of LOH are thought to
contain so-called Tumour Suppressor Genes (TSGs) whose functionality is adversely affected by the reduction in genetic diversity. There exists also a proportion of "background" (not cancer-related) LOH. The aim of this study is to provide, for each LOH rate, the probability of being in the TSG group and in the "background" group. 
Following the approach adopted in the above papers, we consider the following mixture model:
$$
x_i \sim\eta \text{Binomial}(N_i, \pi_1)  + (1-\eta) \text{Beta-Binomial}(N_i, \pi_2, \gamma) \quad\text{for } i = 1, \ldots 40,
$$
where $x_i$ is the number of events of interest (Loss of Heterozygosity) observed in region $i$, and $N_i$ -- the corresponding sample size. Besides, $\eta$ denotes the probability of a location being a member
of the binomial group, $\pi_1$ is the probability of LOH in the
binomial group, $\pi_2$ is the probability of LOH in the beta-binomial
group, and $\gamma$ controls the variability of the beta-binomial
group. That is, the likelihood function for this model is given by
$ \prod_{i=1}^{40} f(x_i, N_i| \eta, \pi_1, \pi_2, \gamma)$ for
$$
f(x_i, N_i| \eta, \pi_1, \pi_2, \gamma)= \eta   {N_i \choose x_i} \pi_1^{x_i} (1- \pi_1)^{n_i - x_i} + (1-\eta) {N_i \choose x_i} \frac{\text{B}\left(x_i + \frac{\pi_2}{\omega_2}, n_i-x_i + \frac{1-\pi_2} {\omega_2} \right)}{\text{B}\left( \frac{\pi_2}{\omega_2},\frac{1-\pi_2} {\omega_2} \right)},
$$
where $\text{B}$ denotes the beta function and $\omega_2:= \frac{\exp(\gamma)}{2\left(1+\exp(\gamma)\right)}$.
The following prior distributions were used for the parameters of interest:
\begin{align*}
\eta & \sim \text{Unif}(0,1)\\
\pi_1 & \sim \text{Unif}(0,1)\\
\pi_2 & \sim \text{Unif}(0,1)\\
\gamma & \sim \text{Unif}(-30, 30).
\end{align*}
The resulting target distribution has two non-symmetric and well-separated modes, as depicted in Figure  \ref{fig:loh_scatterplot}, one of which has a significantly bigger weight than the other; below we denote them by mode 1 and 2, respectively. We based our analysis on 200,000 steps of the main algorithm and 500 BFGS runs for the mode-finding stage. The length of the covariance matrix estimation in burn-in algorithm  was equal to 3000 iterations for each experiment (chosen automatically). Table  \ref{table:loh_t} summarizes the acceptance rates of jumps between the modes for the three versions of the implementation of the algorithm.
\begin{figure}[t]
	\begin{center}
		\includegraphics[width=7cm]{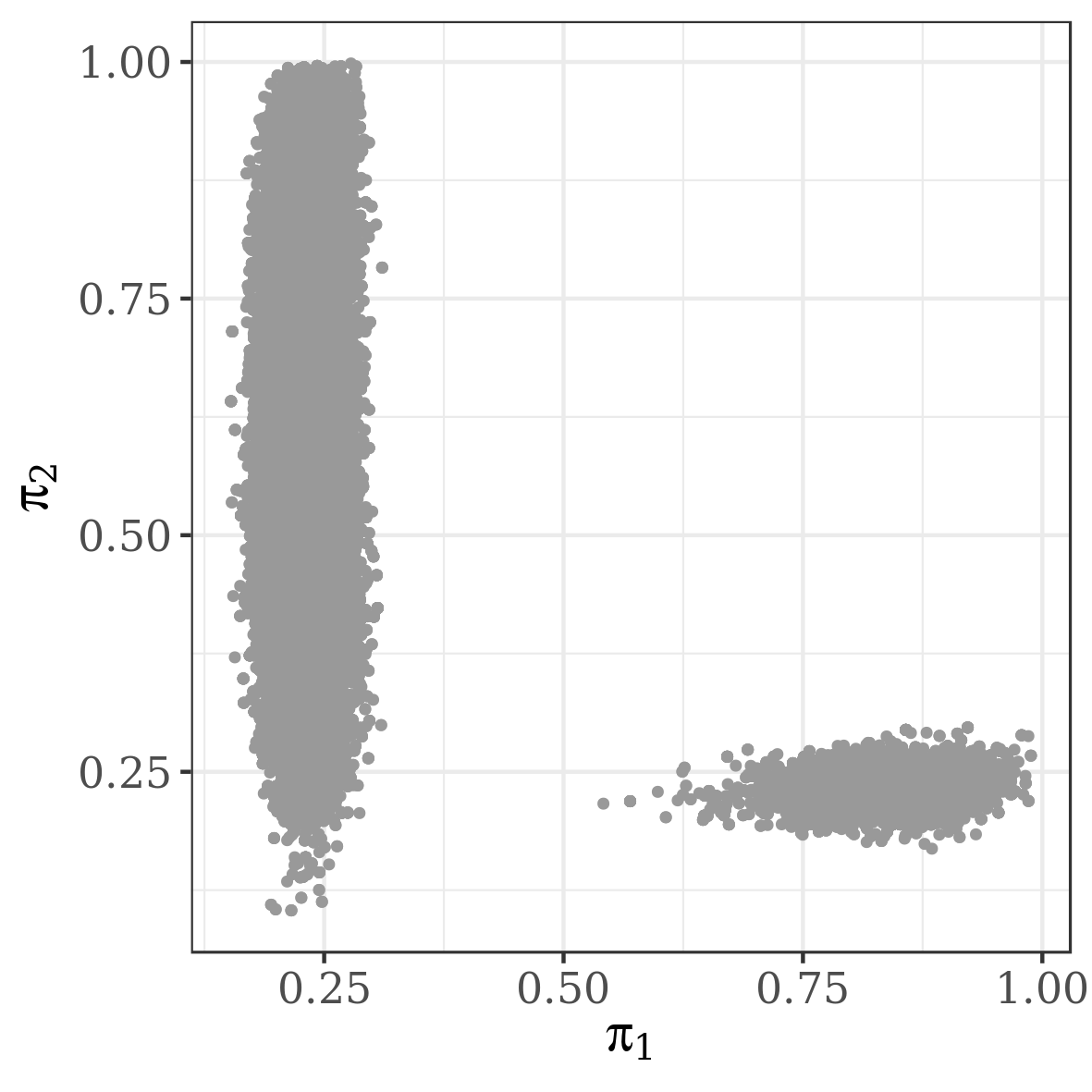} 
	\end{center}
	\caption{\footnotesize Scatterplot of coordinates $\pi_1$ and $\pi_2$ in the LOH study.}\label{fig:loh_scatterplot}
\end{figure}
\FloatBarrier

As stated above, the prior distribution for all the variables has its support on a compact set.  Since this typically has an adverse effect on both mode-finding and sampling, we decided to work with transformed variables, which live on the real line. For the first three variables we applied the logit transformation, i.e. we transformed them 
using a function $t_1(x)= \log(x) - \log(1-x)$. For the last variable we used the transformation given by $t_2(x)= \log(30+x) - \log(30-x)$.

The starting points for the optimisation runs were sampled from the prior distribution and transformed the way described above. The number of function and gradient evaluations for the $20 \times 500$ BFGS runs varied between 27 and 428, with an average of~73.

\begin{table}[ht]
	\begin{center}
		\begin{tabular}{|l|rr|rr|rr|}
			\hline
			& \multicolumn{2}{c|}{deterministic} & \multicolumn{2}{c|}{Gaussian}& \multicolumn{2}{c|}{$t$-distributed} \\ 
			& Lowest & Highest & Lowest & Highest & Lowest & Highest \\
			\hline
			Gaussian local proposal & & & & & & \\
			\hline
			mode 1 to mode 2 & 0.01 & 0.02 & 0.02 & 0.02 & 0.02 & 0.03 \\ 
			mode 2 to mode 1 & 0.44 & 0.71 & 0.70 & 0.76 & 0.71 & 0.77 \\ 
			\hline
			$t$-distributed local proposal & & & & & & \\
			\hline
			mode 1 to mode 2& 0.01 & 0.03 & 0.02 & 0.03 & 0.02 & 0.02 \\ 
			mode 2 to mode 1 & 0.50 & 0.78 & 0.61 & 0.76 & 0.63 & 0.76 \\ 
			\hline
		\end{tabular}
	\end{center}
	\caption{The lowest and the highest value (across 20 runs of the experiment) of the  acceptance rates of jump moves between the two modes of the posterior distribution in the LOH study for different jump methods (for the Gaussian and $t$-distributed local proposal).} \label{table:loh_t}
\end{table}
\FloatBarrier

\subsection{Mixture of Gaussians}
The starting points for the optimisations runs were sampled uniformly on  $[-2,2]^d$. In Table \ref{table:gaussian_mixture_bfgs} we gathered information about the number of the target density and its gradient evaluations (jointly) in the BFGS runs. We reported the minimum, the mean and the maximum value required for the optimisation algorithm to converge. The last two columns show the minimum and the maximum number of iterations used for the estimation of the covariance matrices in the burn-in algorithm. These figures show that indeed the computational budget used by our method for dimensions $d=10$ and $d=20$ was significantly smaller than the budget of APT (see Section \ref{section:mixture_of_gaussians}). 

Figures \ref{fig:gaussian_mixture_boxplots_dim_130_160} and \ref{fig:gaussian_mixture_density_plot_high} illustrate good performance of our method in dimensions $d=130$ and $d=160$, especially for the deterministic jumps. Interestingly, the Gaussian proposal for jumps gives results of the poorest quality on this example.

\begin{table}[ht]
	\centering
	\begin{tabular}{|l|rrr|rr|}
		\hline
		& \multicolumn{3}{|c|}{Optimisation runs} & \multicolumn{2}{|c|}{\makecell{Covariance matrix \\ estimation}} \\
		\hline
		& minimum & mean & maximum & minimum & maximum \\ 
		\hline
		d=10 & 9 & 11.39 & 42 & 3000 & 3000 \\ 
		d=20 & 9 & 10.61 & 39 & 3000 & 7000\\ 
		d=80 & 6 & 7.36 & 27 & 255,000 & 511,000 \\ 
		d=130 & 8 & 8.07 & 24 & 511,000 & 511,000 \\ 
		d=160 & 6 & 8.02 & 22 & 511,000 & 511,000 \\ 
		d=200 & 6 & 6.85 & 23 & 1,023,000 & 1,023,000 \\ 
		\hline
	\end{tabular}
	\caption{First part: number of the target density and its gradient evaluations in the optimisation runs for the mixture of Gaussians. Second part: number of iterations used for the estimation of the covariance matrices in the burn-in algorithm.}\label{table:gaussian_mixture_bfgs}
\end{table}

\begin{figure}[t]
	\begin{center}
		\begin{subfigure}{0.49\textwidth}
			\centering
			\includegraphics[width=7cm]{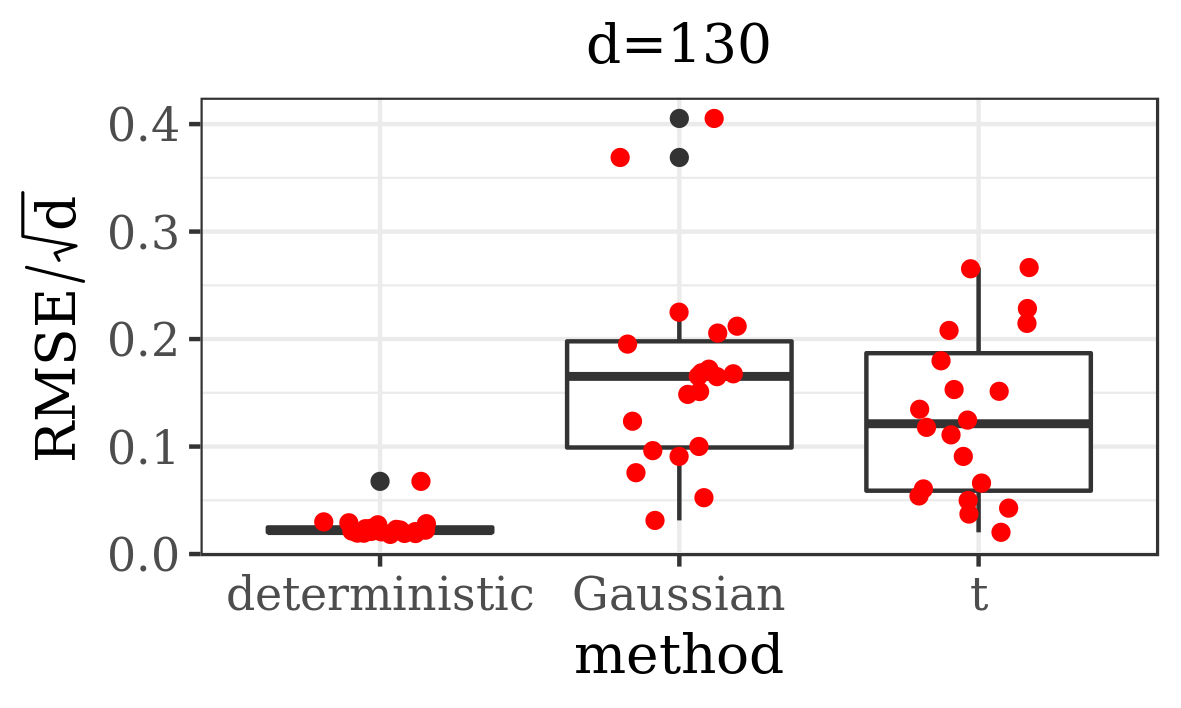} 
			
		\end{subfigure}
		\begin{subfigure}{0.49\textwidth}
			\centering
			\includegraphics[width=7cm]{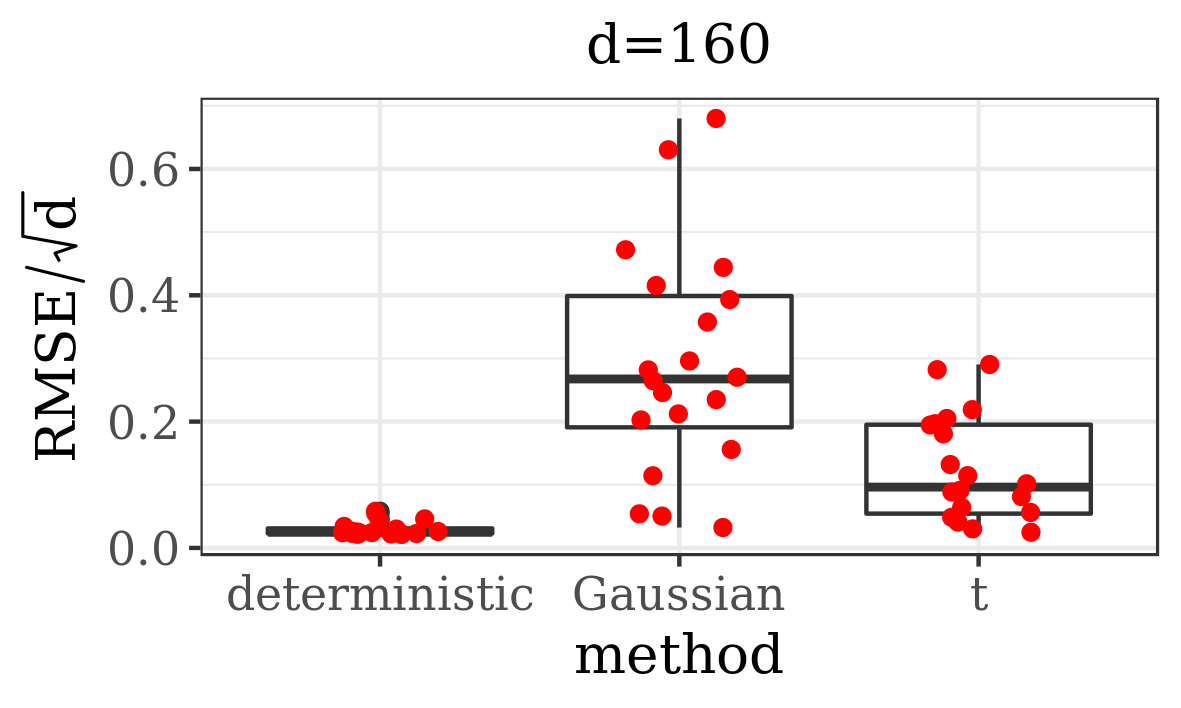} 
			
		\end{subfigure} 
		\caption{\footnotesize Boxplots  of the values of $\text{RMSE}/\sqrt{d}$ for the mixture of Gaussians across 20 runs of the experiment, dimensions 130 and 160. Note different scales on the $y$-axis.} 
		\label{fig:gaussian_mixture_boxplots_dim_130_160}
	\end{center}
\end{figure}
\FloatBarrier
\begin{figure}[t]
	\begin{center}
		
		\includegraphics[width=14cm]{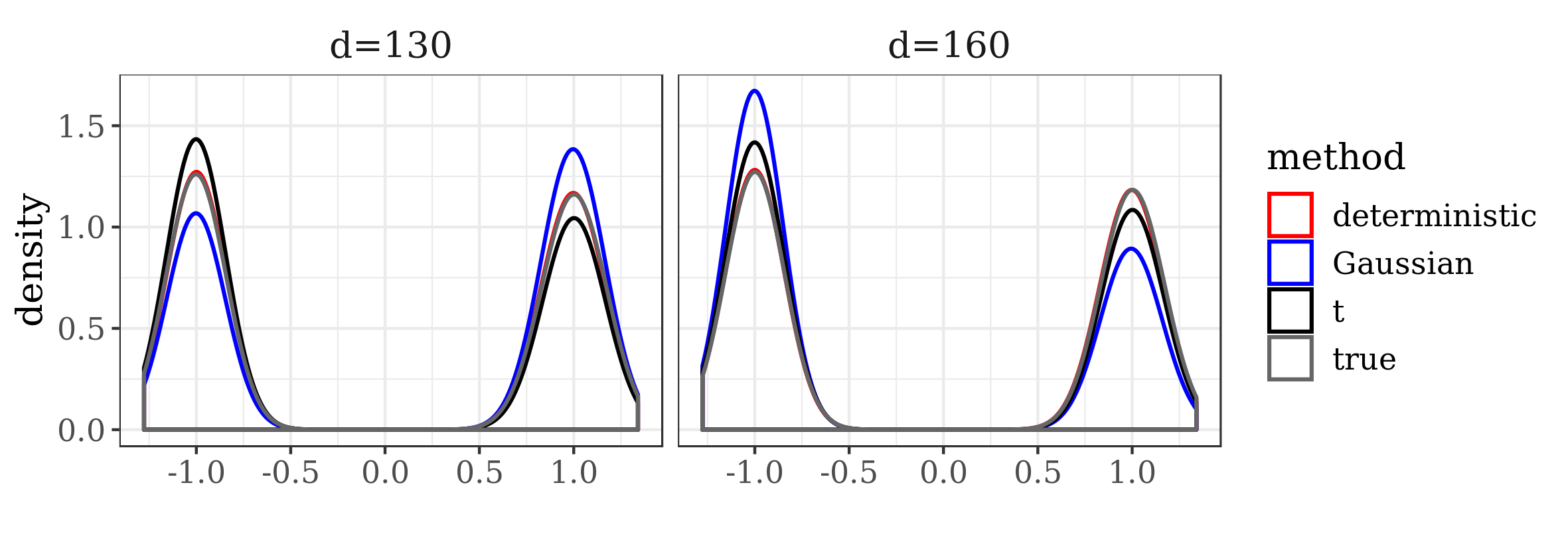} 
		
		\caption{\footnotesize Density plots for dimensions 130 and 160.  The simulations chosen for the analysis correspond to the median value of RMSE across 20 experiments (the tenth largest value of RMSE).}\label{fig:gaussian_mixture_density_plot_high}
	\end{center}
\end{figure}
\FloatBarrier

\subsection{Mixture of $t$ and banana-shaped distributions}
The starting points for the optimisation runs were sampled uniformly on  $[-2, 12]^d$. Table \ref{table:banana_mixture_bfgs} presents analogous information to Table \ref{table:gaussian_mixture_bfgs}, for the mixture of $t$ and banana-shaped distributions considered in Section \ref{section:mixture_bananas}. For dimensions $d=50$ and $d=80$ we did not run the mode-finding part, assuming the locations of the modes were known. Overall, our method in all its versions proved to  perform well on this high-dimensional example, despite the complicated shapes of the modes. Table \ref{table:banana_mixture_high_dimensions} shows that, as before, the deterministic jump method ensures best between-mode mixing. However, it can be noticed that given 20 runs of the experiment, a few times this method delivered results that deviated significantly from the truth (see Figure \ref{fig:banana_mixture_boxplots_dim_50_80}).

\begin{table}[ht]
	\centering
	\begin{tabular}{|l|rrr|rr|}
		\hline
		& \multicolumn{3}{|c|}{Optimisation runs} & \multicolumn{2}{|c|}{\makecell{Covariance matrix \\ estimation}} \\
		\hline
		& minimum & mean & maximum & minimum & maximum \\ 
		\hline
		d=10 & 21 & 49 & 220 & 7000 & 15,000 \\ 
		d=20 & 21 & 48 & 216 & 15,000 & 63,000 \\ 
		d=50 & - & - & - & 255,000 & 255,000 \\ 
		d=80 &  - & - & - & 511,000 & 511,000 \\ 
		\hline
	\end{tabular}
	\caption{First part: number of the target density and its gradient evaluations in the optimisation runs for the  mixture of banana-shaped and $t$-distributions. Second part: number of iterations used for the estimation of the covariance matrices in the burn-in algorithm.}\label{table:banana_mixture_bfgs}
\end{table}

\begin{figure}[t]
	\begin{center}
		\begin{subfigure}{0.49\textwidth}
			\centering
			\includegraphics[width=7cm]{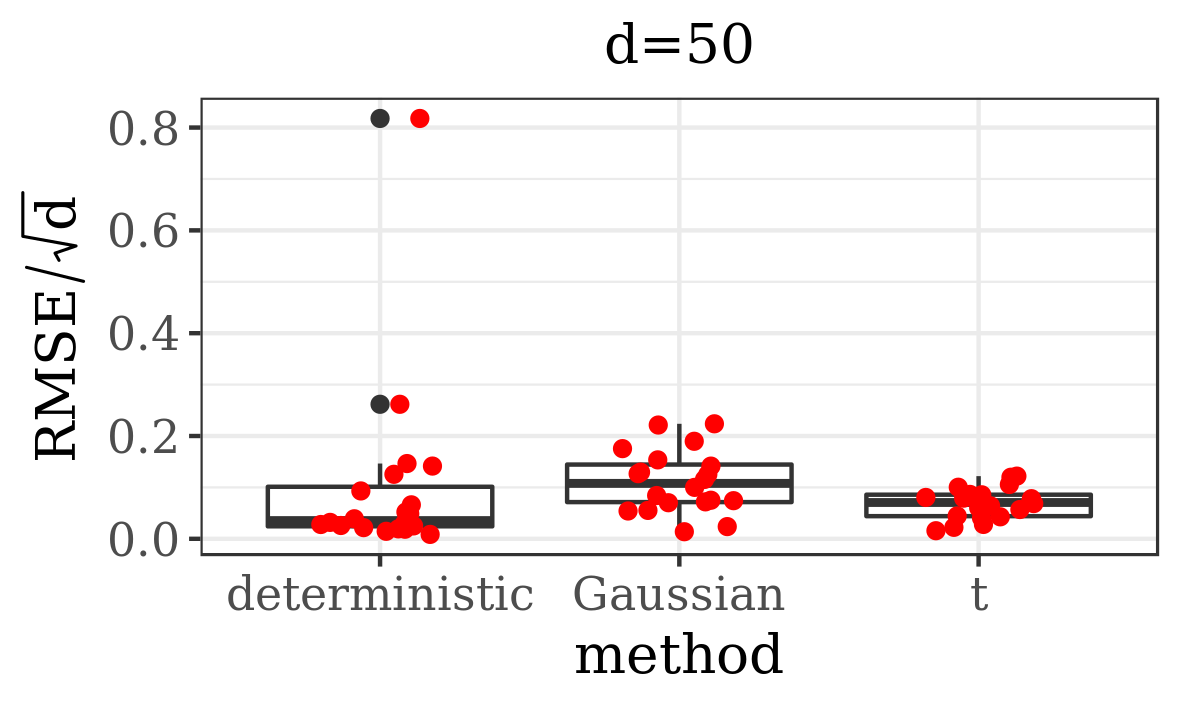} 
			
		\end{subfigure}
		\begin{subfigure}{0.49\textwidth}
			\centering
			\includegraphics[width=7cm]{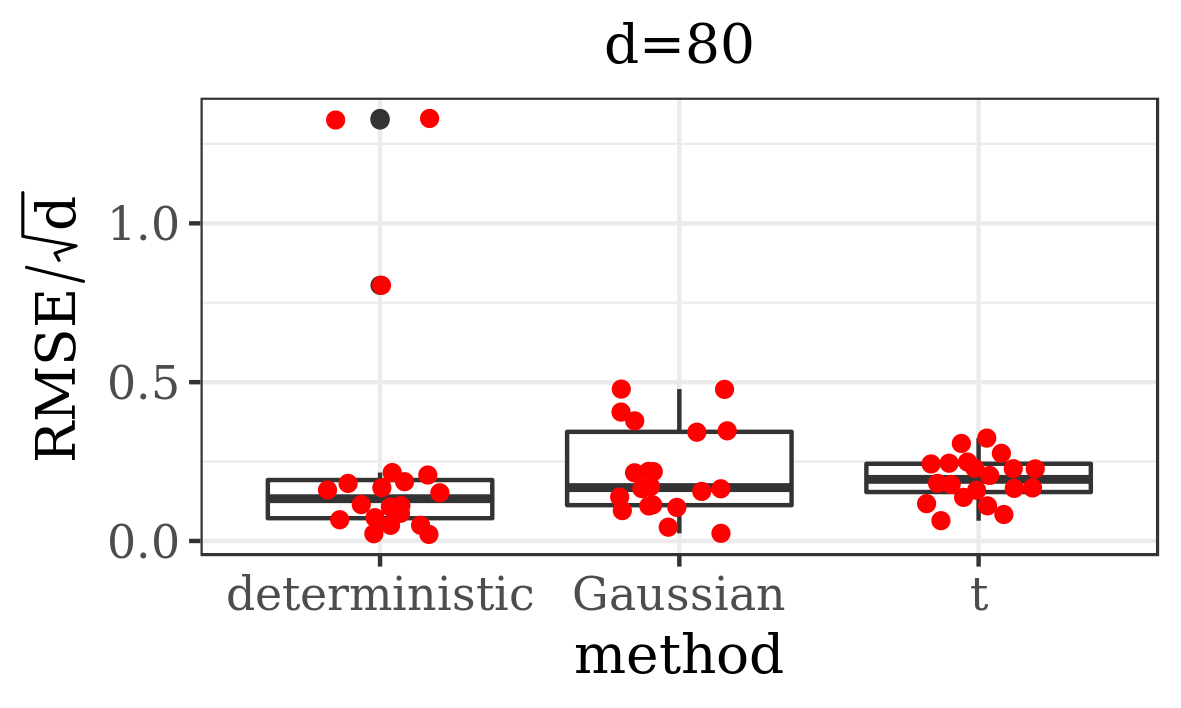} 
			
		\end{subfigure} 
		\caption{\footnotesize Boxplots  of the values of $\text{RMSE}/\sqrt{d}$ for the mixture of banana-shaped and $t$-distributions across 20 runs of the experiment, dimensions 50 and 80. Note different scales on the $y$-axis.} 
		\label{fig:banana_mixture_boxplots_dim_50_80}
	\end{center}
\end{figure}
\FloatBarrier

\begin{table}[ht]
	\begin{center}
		\begin{tabular}{|l|rr|rr|rr|}
			\hline
			& \multicolumn{2}{c|}{deterministic} & \multicolumn{2}{c|}{Gaussian}& \multicolumn{2}{c|}{$t$-distributed} \\ 
			\hline
			& Lowest & Highest & Lowest & Highest & Lowest & Highest \\
			d=50 & 0.26 & 1.00 & 0.03 & 0.14 & 0.06 & 0.23 \\ 
			d=80 & 0.16 & 1.00 & 0.01 & 0.07 & 0.03 & 0.12 \\ 
			\hline
		\end{tabular}
	\end{center}
	\caption{The lowest and the highest value (across 20 runs of the experiment) of the acceptance rate of jumps from a given mode, dimensions 50 and 80.}\label{table:banana_mixture_high_dimensions}
\end{table}
\subsection{Sensor network localisation}
Starting points for the BFGS procedures were sampled uniformly on $[0,1]^{16}$. The number of function and gradient evaluations for these runs varied between 175 and 876, with an average of~400. The starting points for  the APT simulations were the 14 modes identified by the BFGS optimiser and 6 points sampled uniformly on $[0,1]^{16}$. 

Recall that the results for APT presented in Section \ref{section:sensor_network} were based on 4 temperatures. In Figure \ref{fig:sensor_network_boxplot_5_temp} we present analogous results to those shown in Figure \ref{fig:sensor_network_boxplot}, with the number of temperatures increased to 5 (and the same number of iterations equal to 700,000). Under such settings the APT algorithm mixes better between the modes, however, as illustrated in Figure \ref{fig:sensor_network_boxplot_5_temp}, it still yields less stability than our method. Note that APT required $4\times 700,000 = 2,800,000$ or $5\times 700,000 = 3,500,000$ target evaluations, for $4$ and $5$ temperature levels, respectively. Assuming an implementation of JAMS on a standard desktop computer with 8 cores, the computational cost measured by the number of target and gradient evaluations per core would be at most:
\begin{itemize}
	\item for mode finding: $10,000/8 \times 875$ (as for each BFGS run we had at most 875  such evaluations);
	\item for the burn-in-algorithm:  $2 \times 15,000$ target evaluations (as the estimation of 14 covariance matrices needed to be split across 8 cores);
	\item for the main algorithm: $500,000$ target evaluations.
\end{itemize}
Altogether this would give $1,623,750$ evaluations, and the additional overhead resulting from the communication between cores. This figure shows that APT required a larger computational cost than JAMS in our setup even though the above analysis was carried out under a pessimistic scenario. Firstly, on average there were 400 evaluations per BFGS procedure and plugging this value into the above calculations would decrease the overall number of evaluations to 1,000,030. What is more, typically a user would run our algorithm on a server or a cloud service, which we in fact did as well. This would allow to split the computational cost (in particular, that of BFGS runs) across a much larger number of cores.

\begin{figure}[ht]
	\begin{center}
		\centering
		\includegraphics[width=10cm]{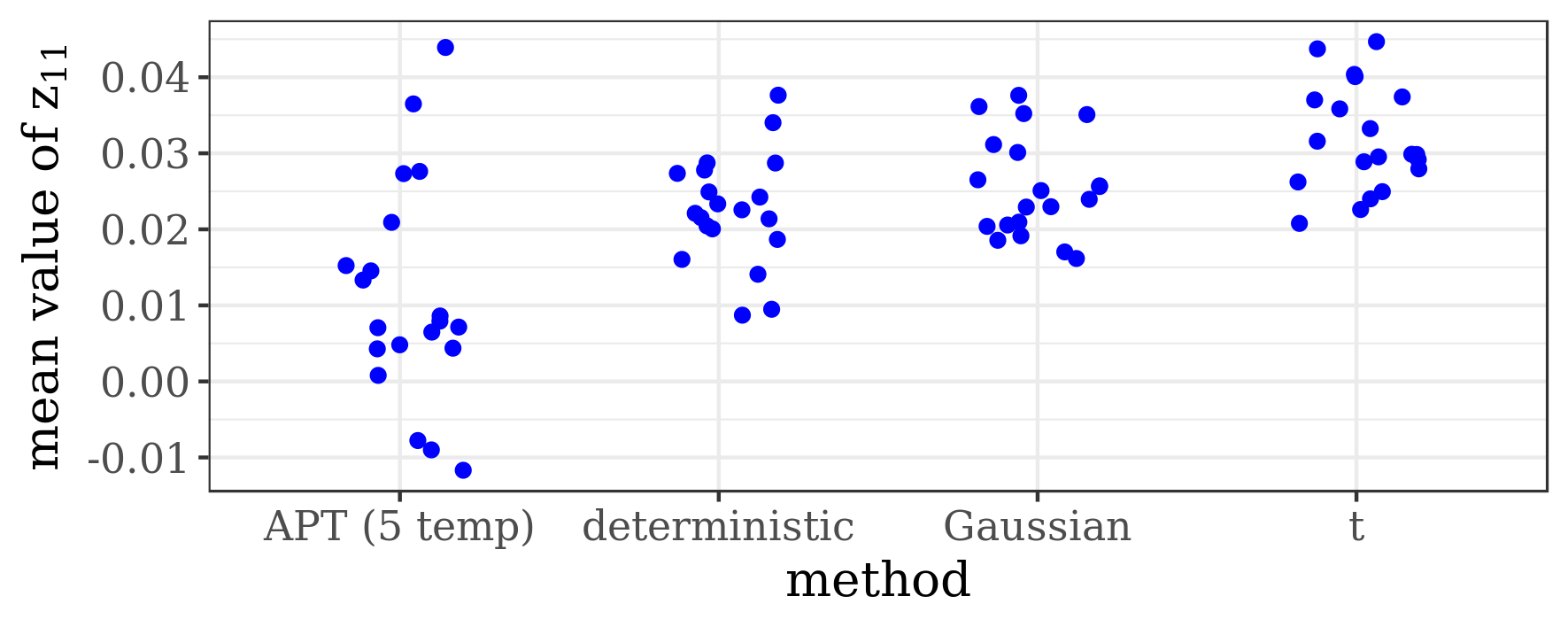} 
		\caption{\footnotesize A boxplot of the mean value of the first coordinate of sensor 1 for 20 runs of the experiment for APT (with 5 temperature levels) and three versions of JAMS.} 
		\label{fig:sensor_network_boxplot_5_temp}
	\end{center}
\end{figure}
\FloatBarrier

\bibliographystyle{apalike}

\end{document}